\PassOptionsToPackage{final}{microtype} 
\documentclass[UKenglish,thm-restate,cleveref,final]{lipics-v2019}
\def\rakStandalone{1}
\ifdefined\rakStandalone
\hideLIPIcs
\makeatletter
\def\@oddfoot{}
\def\doiHeading{}
\makeatother
\nolinenumbers
\else
\fi

\usepackage[left]{showlabels}

\usepackage{mathabx} 

\usepackage{graphicx}
\usepackage{proof}
\usepackage{mathtools}
\usepackage{stmaryrd}
\usepackage{amssymb}
\usepackage{microtype}
\usepackage{doi}

\usepackage[all,cmtip,2cell]{xy}
\usepackage{tikz-cd}

\tikzcdset{
  arrow style=tikz,
  diagrams={>={Computer Modern Rightarrow[scale=0.8]}},
  diagrams={>={To}}
}

\usepackage[draft]{fixme}
\fxsetup{theme=color}
\FXRegisterAuthor{rk}{anrk}{RAK}

\usepackage{xspace}
\newcommand{\ie}{i.e.\@\xspace}
\newcommand{\eg}{e.g.\@\xspace}

\makeatletter
\newcommand*{\etc}{%
  \@ifnextchar{.}%
  {etc}%
  {etc.\@\xspace}%
}
\makeatother

\newcommand{\ms}{\mathsf}
\newcommand{\mb}{\mathbf}

\newcommand{\mt}{\mathtt}

\newcommand{\defin}[1]{{\usefont{T1}{lmss}{b}{n}{#1}}}

\newcommand{\rn}[1]{(\textsc{#1})}

\newcommand{\uscore}{\mbox{\tt\char`\_}}
\newcommand{\limplies}{\supset}

\newcommand{\jisst}[2][{}]{#2\ \ms{type}_{\ms{s}}^{#1}}

\newcommand{\jisft}[1]{#1\ \ms{type}_{\ms{f}}}
\newcommand{\jstype}[3][{}]{{#2} \vdash \jisst[#1]{#3}}
\newcommand{\jtypem}[5]{{#1}\mathrel{;} {#2} \vdash {#3} \mathrel{{:}{:}} {#4} : {#5}}

\newcommand{\jtypef}[3]{#1 \Vdash #2 : #3}

\newcommand{\jcms}[3]{#1 :_{\ms{s}} #2 \leadsto #3}
\newcommand{\jcmf}[3]{#1 :_{\ms{f}} #2 \leadsto #3}

\newcommand{\tSendC}[3]{\mathsf{send}\ #1\ #2;\ #3}
\newcommand{\tSendL}[3]{#1.#2;#3}
\newcommand{\tSendS}[2]{\ms{send}\ #1\ \mathsf{shift}; #2}
\newcommand{\tSendV}[3]{\uscore \leftarrow \ms{output}\ #1\ #2; #3}
\newcommand{\tSendU}[2]{\ms{send}\ #1\ \mathsf{unfold}; #2}
\newcommand{\tRecvC}[3]{#1 \leftarrow \mathsf{recv}\ #2;\ #3}
\newcommand{\tRecvS}[2]{\ms{shift} \leftarrow \ms{recv}\ #1; #2}
\newcommand{\tRecvV}[3]{#1 \leftarrow \ms{input}\ #2; #3}
\newcommand{\tRecvU}[2]{\ms{unfold} \leftarrow \ms{recv}\ #1; #2}
\newcommand{\tClose}[1]{\mathsf{close}\ #1}
\newcommand{\tWait}[2]{\mathsf{wait}\ #1;#2}
\newcommand{\tCase}[2]{\mathsf{case}\ #1\ #2}
\newcommand{\tFwd}[2]{#1 \leftarrow #2}
\newcommand{\tCut}[3]{#1 \leftarrow #2;\ #3}
\newcommand{\tFix}[2]{\ms{fix}\ #1.#2}
\newcommand{\tProc}[3]{#1 \leftarrow \left\{ #2 \right\} \leftarrow #3}

\newcommand{\Tu}{\mathbf{1}}
\newcommand{\Tplus}{\oplus}
\newcommand{\Tot}{\otimes}
\newcommand{\Tamp}{\&}
\newcommand{\Tlolly}{\multimap}
\newcommand{\Tand}[2]{#1 \land #2}
\newcommand{\Timp}[2]{#1 \limplies #2}
\newcommand{\Trec}[2]{\rho #1.#2}
\newcommand{\Tus}[1]{{{\uparrow} #1}}
\newcommand{\Tds}[1]{{{\downarrow} #1}}
\newcommand{\Tproc}[2]{\{#1 \leftarrow #2\}}

\newcommand{\N}{\mathbb{N}}

\newcommand{\FIX}{\mathsf{FIX}}

\newcommand{\fix}{\mathsf{fix}}
\newcommand{\sfix}[1]{\ensuremath{#1^\dag}}
\newcommand{\Tr}[2]{\ensuremath{\ms{Tr}^{#2}#1}}
\newcommand{\Trop}{\ms{Tr}}

\newcommand{\up}{\ms{up}}
\newcommand{\down}{\ms{down}}

\usepackage{stackengine,scalerel}
\newcommand\operatorupX[1]{\,\ThisStyle{\ensurestackMath{%
      #1\stackengine{-0pt}{\,}{\SavedStyle\!^{\mathord{\uparrow}}}{O}{l}{F}{T}{S}}}\!}
\makeatletter
\newcommand\operatorup[1]{\!\mathop{\operatorupX{#1}}\@ifnextchar\{{\,}{\@ifnextchar[{\,}{\@ifnextchar({\,}{}}}}
\makeatother
\newcommand{\dirsqcup}{\operatorup{\bigsqcup}}

\newcommand{\dirsup}{\dirsqcup}

\newcommand{\sto}{\xrightarrow{\bot{!}}}

\DeclareMathOperator{\strictfn}{strict}

\newcommand{\strc}[1]{{#1}_{\bot!}}
\newcommand{\oabc}{$\omega$-\textbf{aBC}}
\newcommand{\oabcs}{\oabc${}_{\bot!}$}
\newcommand{\moabc}{\text{\oabc}}
\newcommand{\moabcs}{\text{\oabcs}}
\newcommand{\Cell}[1]{\mb{Cell}_{#1}}
\newcommand{\CFP}{\mb{CFP}}

\newcommand{\op}[1]{{#1}^{\mathrm{op}}}

\newcommand{\sembr}[1]{\llbracket {#1} \rrbracket}

\newcommand{\upd}[2]{[#1 \mid #2]}

\newcommand{\nto}{\Rightarrow} 

\newcommand{\upim}[1]{\left[#1\right]}

\newcommand{\R}{\mathrel{\mathcal{R}}}
\newcommand{\Rh}{\mathrel{\hat \R}}

\usepackage{adjustbox}

\usepackage{cleveref}
\crefname{diagram}{Diagram}{Diagrams}
\crefname{property}{property}{properties}
\crefname{intn}{Equation}{Equations}
\creflabelformat{intn}{#2(#1)#3}
\crefrangelabelformat{intn}{(#3#1#4) to~(#5#2#6)}
\crefname{figure}{Fig.}{Figs.}

\usepackage{enumitem}
\newlist{thmlist}{enumerate}{1}
\setlist[thmlist]{label=(\arabic{thmlisti}), ref=\thetheorem(\arabic{thmlisti}),noitemsep}
\newlist{lemlist}{enumerate}{1}
\setlist[lemlist]{label=(\arabic{lemlisti}), ref=\thelemma(\arabic{lemlisti}),noitemsep}
\newlist{proplist}{enumerate}{1}
\setlist[proplist]{label=(\arabic{proplisti}), ref=\theproposition(\arabic{proplisti}),noitemsep}
\crefname{thmlisti}{theorem}{theorems}
\crefname{lemlisti}{lemma}{lemmas}
\crefname{proplisti}{proposition}{propositions}

\newcommand{\citeN}{\cite}

\newcommand{\Unfold}{\ms{Unfold}}
\newcommand{\Fold}{\ms{Fold}}

\newcommand{\fd}[1]{#1^\degree}
\newcommand{\cons}[2]{\mt{#1} \mt{::} #2}

\title{A Domain Semantics for Higher-Order Recursive Processes}
\author{Ryan Kavanagh}{Carnegie Mellon University}{rkavanagh@cs.cmu.edu}{https://orcid.org/0000-0001-9497-4276}{}
\authorrunning{R.~Kavanagh}
\Copyright{Ryan Kavanagh}
\relatedversion{A full version of this paper is available at \url{https://arxiv.org/abs/2002.01960}.}

\begin{document}
\maketitle              
\begin{abstract}
  The polarized SILL programming language~\cite{pfenning_griffith_2015:_polar_subst_session_types,toninho_2013:_higher_order_proces_funct_session} uniformly integrates functional programming and session-typed message-passing concurrency.
  It supports general recursion, asynchronous and synchronous communication, and higher-order programs that communicate channels and processes.
  We give polarized SILL a domain-theoretic semantics---the first denotational semantics for a language with this combination of features.
  Session types in polarized SILL denote pairs of domains of unidirectional communications.
  Processes denote continuous functions between these domains, and process composition is interpreted by a trace operator~\cite{joyal_1996:_traced_monoid_categ}.
  We illustrate our semantics by validating expected program equivalences.

\end{abstract}

\section{Introduction}
\label{sec:introduction}

The proofs-as-programs correspondence between linear logic and the session-typed $\pi$-calculus is the foundation of many programming languages for message-passing concurrency.
It has led to an arsenal of techniques for reasoning about processes~\cite{perez_2012:_linear_logic_relat,perez_2014:_linear_logic_relat,toninho_2015:_logic_found_session_concur_comput}.
Denotational semantics~\cite{atkey_2017:_obser_commun_seman_class_proces} and game semantics~\cite{castellan_yoshida_2019:_two_sides_same_coin} have further enriched our understanding.
However, giving a denotational semantics to session-typed languages with recursion has remained~difficult.

To illustrate the difficulty, consider a process that flips bits in a bit stream.
A bit stream is a sequence of labels $\mt{0}$ or $\mt{1}$ sent over a channel.
This communication protocol is specified by a recursive session type ``$\mt{bits}$''.\footnote{Explicitly, $\mt{bits} = \Trec{\beta}{\Tplus\{ \mt{0} : \beta, \mt{1} : \beta \}}$, where $\Trec{\alpha}{A}$ forms recursive types and $\Tplus\{l : A_l\}_{l \in L}$ forms internal choice types.}
The following recursive process, $\mt{flip}$, \textit{uses} a channel \texttt{b} (the input stream) to \textit{provide} a channel \texttt{f} (the flipped stream), both satisfying the type $\mt{bits}$:
\begin{verbatim}
b : bits |- fix F . case b { 0 => f.1; f <- F <- b
                           | 1 => f.0; f <- F <- b } :: f : bits
\end{verbatim}
It cases on the label received over $\mt{b}$ and sends the complementary label over $\mt{f}$.
Subsequent communication is handled by the tail call \verb!f <- F <- b!.
We use the concrete syntax \verb!f <- P <- b! to spawn the process $\mt{P}$ that uses the channel $\mt{b}$ and provides the channel $\mt{f}$.

Operational semantics for session-typed processes generally satisfy a confluence property that implies that their input-output behaviour is deterministic.
This suggests that $\mt{flip}$ denotes a function $\sembr{\mt{flip}}$ from bit streams on $\mt{b}$ to bit streams on~$\mt{f}$.

This processes-as-functions interpretation raises many questions.
The process providing the bit stream on $b$ could get stuck and only send a finite prefix of this bit stream.
How should $\sembr{\mt{flip}}$ handle these finite prefixes?
Computationally, $\sembr{\mt{flip}}$ should be monotone: a longer input prefix should result in no less output.
It should also be continuous: $\sembr{\mt{flip}}$ should not be able to observe an entire infinitely-long bit stream before sending output.

Bidirectional communication further complicates this interpretation.
Indeed, session types specify communication protocols on channels that carry both input to and output from processes.
How do we decompose session-typed communications into the ``inputs'' and ``outputs'' for processes-as-functions in a principled way?
Finally, what does it mean to compose processes $P$ and $Q$ communicating over $c$?
How do we capture the feedback caused by feeding $P$'s output on $c$ into $Q$ and vice-versa?

Our thesis is that a domain-theoretic denotational semantics elucidates the structure of higher-order session-typed languages with recursion.
Domains are well-studied mathematical structures and we leverage their rich theory to study these languages.
A denotational semantics also induces a notion of semantic equivalence.
It is automatically compositional because denotational semantics are defined by recursion on the structure of programs.

\textbf{We give a domain semantics for polarized SILL}~\cite{pfenning_griffith_2015:_polar_subst_session_types,toninho_2013:_higher_order_proces_funct_session}.
Polarized SILL integrates functional and message-passing concurrent programming.
It supports recursive session types and processes, and higher-order programs that communicate channels and processes.
Communication is asynchronous, and synchronous communication is encoded using polarity~shifts.
Though denotational semantics~\cite{atkey_2017:_obser_commun_seman_class_proces,castellan_yoshida_2019:_two_sides_same_coin,kokke_2019:_better_late_than_never} exist for session-typed languages in restricted settings (\eg, in the absence of recursion, functional constructs, or higher-order features), ours is the first denotational semantics for a full-featured session-typed~language.

For readers familiar with session types, we hope they can take away the high-level ideas of their semantic interpretation in the presence of (nonlinear) functions and arbitrary recursion and how it might be used to reason about process equivalence.
A particular phenomenon not usually addressed is that processes may fail to communicate along a given channel in the presence of recursively defined types and processes.
This phenomenon is easily addressed domain-theoretically: because processes denote continuous (so monotone) functions, they uniformly treat complete and incomplete communications.

For readers familiar with denotational semantics, we hope they can take away the ideas behind its application to bidirectional, session-typed communication in the presence of recursion.
The key insights here, when compared to the denotational semantics of functional languages, are that (session) types denote pairs of domains of unidirectional communications instead of domains of values, and that program (process) composition is given by a trace operator instead of by function composition.

We give an overview of polarized SILL and of our semantics in \cref{sec:overview-sill,sec:overview-semantics}.
The details of our semantics are given in \cref{sec:semantic-clauses}.
We revisit the bit-flipping process in \cref{sec:bit-streams}.
All typing rules and semantic clauses can be found in \cref{sec:rules-polarized-sill,sec:summ-interpr}, respectively.


\section{Overview of Polarized SILL}
\label{sec:overview-sill}

The polarized SILL programming language cohesively integrates functional computation and message-passing concurrent computation.
Its concurrent computation layer arises from a proofs-as-processes correspondence between intuitionistic linear logic and the session-typed $\pi$-calculus~\cite{caires_pfenning_2010:_session_types_intuit_linear_propos}.
A session type $A$ prescribes a protocol for communicating over a channel.
A process $P$ \defin{provides} a service $A$ on a channel $c$ and \defin{uses} (is a \defin{client} of) zero or more services $A_i$ on channels $c_i$.
The used services form a linear context $\Delta = c_1 : A_1, \dotsc, c_n : A_n$.
The process $P$ can use values from the functional layer.
These are abstracted by a structural context $\Psi$ of functional variables.
These data are captured by the judgment $\jtypem{\Psi}{\Delta}{P}{c}{A}$.
It is inductively defined by the rules of \cref{sec:rules-proc-form}.

Session types are polarized: they can be partitioned as \defin{positive} or \defin{negative}.
When looking at a judgment $\jtypem{\Psi}{\Delta}{P}{c}{A}$, we can imagine ``positive information'' as flowing left-to-right and ``negative information'' as flowing right-to-left.
For example, when $A$ is positive, communication on $c$ is sent by $P$; when $A$ is negative, it is received by $P$.
As $P$ executes, the type of a channel $b : B$ in $\Delta, c : A$ evolves, sometimes becoming a positive subphrase of $B$, sometimes a negative subphrase of $B$.
It is this evolution that allows for bidirectional communication.
We write $\jisst[+]{B}$ to mean $B$ is positive and $\jisst[-]{B}$ to mean $B$ is negative.
Most session types have a polar-dual session type, where the direction of the communication is reversed.
For brevity, we consider only positive session types below.
The semantics of negative session types follow by symmetry and are given in the appendices.

The functional layer is the simply-typed $\lambda$-calculus with a fixed-point operator and a call-by-value evaluation semantics.
A judgment $\jtypef{\Psi}{M}{\tau}$ means the functional term $M$ has functional type $\tau$ under the structural context $\Psi$ of functional variables $x_i : \tau_i$.
This judgment's definition is standard and it is inductively defined by the rules of \cref{sec:rules-term-formation}.
We use the judgment $\jisft{\tau}$ to mean that $\tau$ is a functional type.
New is the base type $\Tproc{a:A}{\overline{a_i:A_i}}$ of quoted processes, where we abbreviate ordered lists using an overline.
Quoted processes are introduced and eliminated by the rules:
\[
  \begin{adjustbox}{max width=\textwidth}
    \infer[\rn{I-\{\}}]{
      \jtypef{\Psi}{\tProc{a}{P}{\overline{a_i}}}{\Tproc{a:A}{\overline{a_i:A_i}}}
    }{
      \jtypem{\Psi}{\overline{a_i:A_i}}{P}{a}{A}
    }
    \quad
    \infer[\rn{E-\{\}}]{
      \jtypem{\Psi}{\overline{a_i:A_i}}{\tProc{a}{M}{\overline a_i}}{a}{A}
    }{
      \jtypef{\Psi}{M}{\Tproc{a:A}{\overline{a_i:A_i}}}
    }
  \end{adjustbox}
\]

Functional values can be sent over channels of type $\Tand{\tau}{A}$.
It specifies that the sent value has type $\tau$ and that subsequent communication has type $A$.
The process $\tSendV{a}{M}{P}$ evaluates $M$ to a value, sends it over the channel $a$, and continues as $P$.
The process $\tRecvV{x}{a}{Q}$ receives a value on $a$, binds it to the variable $x$, and~continues~as~$Q$.
\[
  \begin{adjustbox}{max width=\textwidth}
    \infer[\rn{$\Tand{}{} R$}]{
      \jtypem{\Psi}{\Delta}{\tSendV{a}{M}{P}}{a}{\Tand{\tau}{A}}
    }{
      \jtypef{\Psi}{M}{\tau}
      &
      \jtypem{\Psi}{\Delta}{P}{a}{A}
    }
    \quad
    \infer[\rn{$\Tand{}{} L$}]{
      \jtypem{\Psi}{\Delta, a:\Tand{\tau}{A}}{\tRecvV{x}{a}{Q}}{c}{C}
    }{
      \jtypem{\Psi,x:\tau}{\Delta, a:A}{Q}{c}{C}
    }
  \end{adjustbox}
\]

The process $\tFwd{b}{a}$ forwards all messages between channels $a$ and $b$ of the same type.
Process composition $\tCut{a}{P}{Q}$ captures Milner's ``parallel composition plus hiding'' operation~\cite[pp.~20f.]{milner_1980:_calcul_commun_system}.
It spawns processes $P$ and $Q$ that communicate over a private channel~$a$.
\[
  \infer[\rn{Fwd}]{
    \jtypem{\Psi}{a:A}{\tFwd{b}{a}}{b}{A}
  }{}
  \quad
  \infer[\rn{Cut}]{
    \jtypem{\Psi}{\Delta_1,\Delta_2}{\tCut{a}{P}{Q}}{c}{C}
  }{
    \jtypem{\Psi}{\Delta_1}{P}{a}{A}
    &
    \jtypem{\Psi}{a:A,\Delta_2}{Q}{c}{C}
  }
\]

Processes can close channels of type $\Tu$.
To do so, the process $\tClose a$ sends a ``close message'' over the channel $a$ and terminates.
The process $\tWait{a}{P}$ blocks on $a$ until it receives the close message and then continues as $P$.
\[
  \infer[\rn{$\Tu R$}]{
    \jtypem{\Psi}{\cdot}{\tClose a}{a}{\Tu}
  }{}
  \quad
  \infer[\rn{$\Tu L$}]{
    \jtypem{\Psi}{\Delta, a : \Tu}{\tWait{a}{P}}{c}{C}
  }{
    \jtypem{\Psi}{\Delta}{P}{c}{C}
  }
\]

Processes can communicate channels over channels.
The protocol $B \Tot A$ prescribes transmitting a channel of type $B$ followed by communication of type $A$.
The process $\tSendC{a}{b}{P}$ sends the channel $b$ over the channel $a$ and continues as $P$.
The process $\tRecvC{b}{a}{P}$ receives a channel over $a$, binds it to the name $b$, and continues as $P$.
\[
  \begin{adjustbox}{max width=\textwidth}
    \infer[\rn{$\Tot R^*$}]{
      \jtypem{\Psi}{\Delta, b : B}{\tSendC{a}{b}{P}}{a}{B \Tot A}
    }{
      \jtypem{\Psi}{\Delta}{P}{a}{A}
    }
    \quad
    \infer[\rn{$\Tot L$}]{
      \jtypem{\Psi}{\Delta, a : B \Tot A}{\tRecvC{b}{a}{P}}{c}{C}
    }{
      \jtypem{\Psi}{\Delta, a : A, b : B}{P}{c}{C}
    }
  \end{adjustbox}
\]

Process communication is asynchronous.
Synchronization is encoded using ``polarity shifts''~\cite{pfenning_griffith_2015:_polar_subst_session_types}.
The positive protocol $\Tds A$ prescribes a synchronization (the shift) followed by communication satisfying the negative type $A$.
The process $\tSendS{a}{P}$ signals that it is ready to receive on $a$ by sending a ``shift message'' on $a$, and continues as $P$.
The process $\tRecvS{a}{Q}$ blocks until it receives the shift message and continues as $Q$.
\[
  \begin{adjustbox}{max width=\textwidth}
    \infer[\rn{$\Tds{} R$}]{
      \jtypem{\Psi}{\Delta}{\tSendS{a}{P}}{a}{\Tds A}
    }{
      \jtypem{\Psi}{\Delta}{P}{a}{A}
    }
    \quad
    \infer[\rn{$\Tds{} L$}]{
      \jtypem{\Psi}{\Delta,a : \Tds A}{\tRecvS{a}{P}}{c}{C}
    }{
      \jtypem{\Psi}{\Delta,a : A}{P}{c}{C}
    }
  \end{adjustbox}
\]

Processes can choose between services.
An internal choice type  $\Tplus \{ l : A_l \}_{l \in L}$ prescribes a choice between session types $\{ A_l \}_{l \in L}$ ($L$ finite).
The process $\tSendL{a}{k}{P}$ chooses to provide the service $A_k$ by sending the label $k$ on $a$, and then continues as $P$.
The process $\tCase{a}{\left\{l \Rightarrow P_l\right\}_{l \in L}}$ blocks until it receives a label $k$ on $a$ and then continues as $P_k$.
\[
  \begin{adjustbox}{max width=\textwidth}
    \infer[\rn{$\Tplus R_k$}]{
      \jtypem{\Psi}{\Delta}{\tSendL{a}{k}{P}}{a}{{\Tplus\{l:A_l\}}_{l \in L}}
    }{
      \jtypem{\Psi}{\Delta}{P}{a}{A_k}\quad(k \in L)
    }
    \quad
    \infer[\rn{$\Tplus L$}]{
      \jtypem{\Psi}{\Delta,a:{\Tplus\{l : A_l\}}_{l \in L}}{\tCase{a}{\left\{l_l \Rightarrow P_l\right\}_{i\in I}}}{c}{C}
    }{
      \jtypem{\Psi}{\Delta,a:A_l}{P_l}{c}{C}\quad(\forall l \in L)
    }
  \end{adjustbox}
\]

Open session types are given by the judgment $\jstype[p]{\Xi}{A}$, where $\Xi$ is a structural context of polarized type variables $\jisst[p_i]{\alpha_i}$ and $p, p_i \in \{{-},{+}\}$.
We abbreviate the judgment as $\Xi \vdash \jisst{A}$ when no ambiguity arises.
It is inductively defined by the rules of \cref{sec:rules-type-formation}.

The recursive type $\Trec{\alpha}{A}$ prescribes an ``unfold'' message followed by communication of type $[\Trec{\alpha}{A}/\alpha]A$.
The process $\tSendU{a}{P}$ sends an unfold message and continues as $P$.
The process $\tRecvU{a}{P}$ receives an unfold message and continues as $P$.
\[
  \begin{adjustbox}{max width=\textwidth}
    \infer[\rn{$\rho^+R$}]{
      \jtypem{\Psi}{\Delta}{\tSendU{a}{P}}{a}{\Trec{\alpha}{A}}
    }{
      \jtypem{\Psi}{\Delta}{P}{a}{[\Trec{\alpha}{A}/\alpha]A}
      &
      \cdot \vdash \jisst[+]{\Trec{\alpha}{A}}
    }
    \,
    \infer[\rn{$\rho^+L$}]{
      \jtypem{\Psi}{\Delta, a : \Trec{\alpha}{A}}{\tRecvU{a}{P}}{c}{C}
    }{
      \jtypem{\Psi}{\Delta, a : [\Trec{\alpha}{A}/\alpha]A}{P}{c}{C}
      &
      \cdot \vdash \jisst[+]{\Trec{\alpha}{A}}
    }
  \end{adjustbox}
\]


\section{Overview of the Semantics}
\label{sec:overview-semantics}

The semantics of the functional layer is standard~\cite{crole_1993:_categ_types,gunter_1992:_seman_progr_languag,reynolds_2009:_theor_progr_languag,tennent_1995:_denot_seman}.
A functional type $\tau$ denotes a Scott domain (an $\omega$-algebraic bounded-complete pointed dcpo) $\sembr{\tau}$.
A structural context $\Psi$ of functional variables $x_i : \tau_i$ is interpreted as the $\Psi$-indexed product $\sembr{\Psi} = \prod_{x_i \in \Psi} \sembr{\tau_i}$.
A functional term $\jtypef{\Psi}{M}{\tau}$ denotes a continuous function $\sembr{\jtypef{\Psi}{M}{\tau}} : \sembr{\Psi} \to \sembr{\tau}$ in the category $\moabc$ of Scott domains and continuous functions.

We cannot interpret processes $\jtypem{\Psi}{\Delta}{P}{c}{A}$ in the same way as functional terms, that is, as functions $\sembr{\Psi} \times \sembr{\Delta} \to \sembr{A}$.
This is due to the fundamental difference between \textit{variables} and \textit{channel names}.
A variable $x : \tau$ in the context $\Psi$ stands for a value of type $\tau$.
A channel name in $\Delta,c:A$ stands not for a value, but for a channel of typed bidirectional communications.
Informally and novelly, we interpret processes as continuous~functions
\begin{equation}
  \label{eq:fossacs:4}
  \sembr{\jtypem{\Psi}{\Delta}{P}{c}{A}} : \sembr{\Psi} \to [\text{``}\ms{inputOn}(\Delta, c:A)\text{''} \to \text{``}\ms{outputOn}(\Delta, c:A)\text{''}]
\end{equation}
where we use the notation $[ {\cdot} \to {\cdot} ]$ for function spaces.

We use polarity to split a bidirectional communication into a pair of unidirectional communications.
Given a protocol $B$, its \textit{positive aspect} prescribes the left-to-right communications, while its \textit{negative aspect} prescribes the right-to-left communications.
These aspects determine Scott domains $\sembr{B}^-$ and $\sembr{B}^+$ that respectively contain the ``negative'' and ``positive'' portions of the bidirectional communications.
Their bottom elements capture the absence of communication, \eg, due to a stuck process.
The informal interpretation \eqref{eq:fossacs:4} is then made precise as the continuous function
\begin{equation}
  \label{eq:fossacs:7}
  \sembr{\jtypem{\Psi}{\Delta}{P}{c}{A}} : \sembr{\Psi} \to [\sembr{\Delta}^+ \times \sembr{c : A}^- \to \sembr{\Delta}^- \times \sembr{c : A}^+]
\end{equation}
in $\moabc$, where $\sembr{c_1 : A_1, \dotsc, c_n : A_n}^p$ is the $\{c_1^p,\dotsc,c_n^p\}$-indexed product $\prod_{c_i^p} \sembr{A_i}^p$.

Operationally, the process composition $\jtypem{\Psi}{\Delta_1,\Delta_2}{\tCut{a}{P}{Q}}{c}{C}$ spawns two processes $\jtypem{\Psi}{\Delta_1}{P}{a}{A}$ and $\jtypem{\Psi}{a : A,\Delta_2}{Q}{c}{C}$ communicating along the private channel $a$ of type $A$.
We follow prior work on the semantics of dataflow~\cite{kahn_1974:_seman_simpl_languag_paral_progr} and on the geometry of interactions~\cite{abramsky_jagadeesan_1994:_new_found_geomet_inter,abramsky_2002:_geomet_inter_linear_combin_algeb}, and capture this feedback using a least fixed point.
Given an environment $u \in \sembr{\Psi}$, and $(\delta_1^+, \delta_2^+, c^-) \in \sembr{\Delta_1,\Delta_2}^+ \times \sembr{c : C}^-$, we define
\begin{equation}
  \label[intn]{eq:fossacs:1}
  \sembr{\jtypem{\Psi}{\Delta_1,\Delta_2}{\tCut{a}{P}{Q}}{c}{C}}u(\delta_1^+, \delta_2^+, c^-) = (\delta_1^-, \delta_2^-, c^+)
\end{equation}
where $\delta_1^-$, $\delta_2^-$, $a^-$, $a^+$, and $c^+$ form the least solution to the equations
\begin{align*}
  (\delta_1^-, a^+) &= \sembr{\jtypem{\Psi}{\Delta_1}{P}{a}{A}}u(\delta_1^+, a^-),\\
  (\delta_2^-, a^-, c^+) &= \sembr{\jtypem{\Psi}{a : A,\Delta_2}{Q}{c}{C}}u(\delta_2^+, a^+, c^-).
\end{align*}
We can compute this least solution as the least upper bound $(\delta_1^-, \delta_2^-, a^-, a^+, c^+)$ of the following inductively-defined ascending chain $(\delta_{1,n}^-, \delta_{2,n}^-, a_n^-,a_n^+,c_n^+)_{n \in \N}$:
\begin{align*}
    (\delta_{1,0}^-, a_0^+) &= \sembr{\jtypem{\Psi}{\Delta_1}{P}{a}{A}}u(\delta_1^+, \bot_{\sembr{A}^-}),\\
    (\delta_{2,0}^-, a_0^-, c_0^+) &= \sembr{\jtypem{\Psi}{a : A,\Delta_2}{Q}{c}{C}}u(\delta_2^+, \bot_{\sembr{A}^+}, c^-),\\
  (\delta_{1,n+1}^-, a_{n+1}^+) &= \sembr{\jtypem{\Psi}{\Delta_1}{P}{a}{A}}u(\delta_1^+, a_n^-),\\
  (\delta_{2,n+1}^-, a_{n+1}^-, c_{n+1}^+) &= \sembr{\jtypem{\Psi}{a : A,\Delta_2}{Q}{c}{C}}u(\delta_2^+, a_n^+, c^-).
\end{align*}
This chain captures the iterative nature of the feedback loop created by $P$ and $Q$ communicating on $a : A$.

This fixed-point operation is an instance of a ``trace operator'' on the traced monoidal category $\moabc$ \cite{abramsky_2002:_geomet_inter_linear_combin_algeb,cazanescu_stefanescu_1990:_towar_new_algeb,joyal_1996:_traced_monoid_categ}.
A trace operator is a family of functions $\Trop^X_{A,B} : [A \times X \to B \times X] \to [A \to B]$ natural in $A$ and $B$, dinatural in $X$, and satisfying a collection of axioms.
It captures Milner's ``parallel composition plus hiding'' operation and it is used by Abramsky~et~al.~\cite{abramsky_1996:_inter_categ_found} to give semantics to concurrent computation.
Trace operators have a rich theory that we can leverage to reason about process composition.

We build on standard domain-theoretic techniques to interpret session types and functional types.
Open session types $\jstype{\Xi}{A}$ denote locally continuous functors on a category of domains.
Let $\moabcs$ be the subcategory of $\moabc$ whose morphisms are strict functions.
Let $\sembr{\Xi}$ be the $\Xi$-indexed product $\prod_{\alpha \in \Xi} \moabcs$.
The positive and negative aspects $\sembr{\jstype{\Xi}{A}}^-$ and  $\sembr{\jstype{\Xi}{A}}^+$ are then locally continuous functors from $\sembr{\Xi}$ to $\moabcs$.
Recursive types are interpreted as solutions to domain equations.

Our semantics does not use $\omega$-algebraicity or bounded-completeness, but we mention this additional structure in case it is useful for future work.
All occurrences of $\moabc$ and $\moabcs$ could respectively be replaced by $\mb{DCPO}_\bot$ and $\strc{\mb{DCPO}}$.


\section{Background and Notation}
\label{sec:background}

Given a locally small category $\mb{C}$, we write $\mb{C}({-}, {-}) : \op{\mb{C}} \times \mb{C} \to \mb{Set}$ for the hom functor.
If $\mb{C}$ has an internal hom, we write $\mb{C}[{-} \to {-}] : \op{\mb{C}} \times \mb{C} \to \mb{C}$ or just $[{-} \to {-}]$ for it.

We write $(d_1 : D_1) \times \dotsb \times (d_n : D_n)$ for the product of the $D_i$ indexed by the $d_i$.
Given $\delta_i \in D_i$, we write $(d_1 : \delta_1,\dotsc,d_n : \delta_n)$ for an element of this product.
Given an indexed product $\prod_{i \in I} D_i$ and a subset $J \subseteq I$, we write $\pi^I_J$ or $\pi_J$ for the projection $\prod_{i \in I} D_i \to \prod_{j \in J} D_j$.

The category $\moabc$ has a continuous \defin{least-fixed-point operator} $\fix : [D \to D] \to D$ for each object $D$.
It also has a continuous fixed-point operator $\sfix{(\cdot)} : [A \times X \to X] \to [A \to X]$ given by $\sfix{f}(a) = \fix(\lambda x.f(a,x))$.
It satisfies the fixed-point identity of \cite{bloom_esik_1996:_fixed_point_operat}: $\sfix f = f \circ \langle \ms{id}, \sfix f \rangle$.

The category $\moabc$ is also equipped with a \defin{trace operator}~\cite{abramsky_2002:_geomet_inter_linear_combin_algeb,cazanescu_stefanescu_1990:_towar_new_algeb,joyal_1996:_traced_monoid_categ}.
It fixes the $X$ component of a morphism $f : A \times X \to B \times X$ to produce a morphism $\Tr{(f)}{X} : A \to B$.
It is given by $\Tr{(f)}{X} = \pi^{B \times X}_B \circ f \circ \langle \ms{id}_A, \sfix{(\pi^{B\times X}_X \circ f)} \rangle$.
It has the following Knaster-Tarski-style formulation: $\Tr{(f)}{X}(a) = \pi_B(\bigsqcap\{(b,x) \mid f(a,x) \sqsubseteq (b,x) \})$.

The \defin{lifting functor} $({-})_\bot : \moabc \to \moabcs$ is left-adjoint~\cite[p.~44]{abramsky_jung_1995:_domain_theor} to the inclusion $\moabcs \hookrightarrow \moabc$.
The domain $D_\bot$ is obtained by adjoining a new bottom element to $D$.
The unit $\up : \ms{id} \to ({-})_\bot$ witnesses the inclusion of $D$ in $D_\bot$.
We write $\upim{d}$ for $\up_D(d) \in D_\bot$ and $\down$ for the counit.
The morphism $f_\bot : D_\bot \to E_\bot$ is given by $f_\bot(\bot) = \bot$ and $f_\bot(\upim{d}) = \upim{f(d)}$ for $d \in D$.
Given a functor $F$, we abbreviate $({-})_\bot \circ F$ as $F_\bot$.

To make a function $f : \prod_{i \in I} A_i \to B$ strict in a component $j \in I$, we use the continuous function $\strictfn_j : [\prod_{i \in I} A_i \to B] \to [\prod_{i \in I} A_i \to B]$:
\begin{equation}
  \label{eq:46}
  \strictfn_j(f)\left(\left(a_i\right)_{i \in I}\right) =
  \begin{cases}
    \bot_B & \text{if } a_j = \bot_{A_j}\\
    f\left(\left(a_i\right)_{i \in I}\right) & \text{otherwise}.
  \end{cases}
\end{equation}


\section{Semantic Clauses}
\label{sec:semantic-clauses}

We define the denotations of judgments by induction on their derivation.
To illustrate our semantics, we show that it validates expected equivalences.

\begin{definition}[Semantic equivalence]%
  \label{def:3}
  Processes $\jtypem{\Psi}{\Delta}{P}{c}{A}$ and $\jtypem{\Psi}{\Delta}{Q}{c}{A}$ are equivalent, $P \equiv Q$, if $\sembr{\jtypem{\Psi}{\Delta}{P}{c}{A}} = \sembr{\jtypem{\Psi}{\Delta}{Q}{c}{A}}$.
  Functional terms $\jtypef{\Psi}{M}{\tau}$ and $\jtypef{\Psi}{N}{\tau}$ are equivalent, $M \equiv N$, if $\sembr{\jtypef{\Psi}{M}{\tau}} = \sembr{\jtypef{\Psi}{N}{\tau}}$.
\end{definition}

\subsection{Functional Programming and Value Transmission}
\label{sec:funct-progr-value}

The functional layer is the simply-typed $\lambda$-calculus with a call-by-value semantics and a fixed-point operator.
Arrow types are formed by the rule \rn{T$\to$}.
They are interpreted as strict function spaces in $\moabcs$ to enforce a call-by-value semantics:
\begin{equation}
  \label[intn]{eq:73}
  \sembr{\jisft{\tau \to \sigma}} = \moabcs[\sembr{\jisft\tau} \to \sembr{\jisft{\sigma}}].
\end{equation}
The typing rules \rn{F-Var}, \rn{F-Fun}, \rn{F-App}, and \rn{F-Fix} for the functional layer are standard.
The call-by-value semantics is as in \cite{stoy_1977:_denot_seman}.
We let $u$ range over $\sembr{\Psi}$.
The environment $\upd{u}{x \mapsto v} \in \sembr{\Psi, x : \tau}$ maps $x$ to $v$ and $y$ to $u(y)$ for all $y \in \Psi$.
The fixed-point operator \rn{F-Fix} is interpreted using the fixed-point operator defined in \cref{sec:background}.
\begin{align}
  \sembr{\jtypef{\Psi,x:\tau}{x}{\tau}}u &= \pi^{\Psi,x}_{x}u\label[intn]{eq:39}\\
  \sembr{\jtypef{\Psi}{\lambda x: \tau.M}{\tau \to \sigma}}u &= \strictfn\left(\lambda v \in \sembr{\tau}.\sembr{\jtypef{\Psi,x:\tau}{M}{\sigma}}\upd{u}{x \mapsto v}\right)\label[intn]{eq:41}\\
  \sembr{\jtypef{\Psi}{MN}{\sigma}}u &= \sembr{\jtypef{\Psi}{M}{\tau \to \sigma}}u(\sembr{\jtypef{\Psi}{N}{\tau}}u)\label[intn]{eq:42}\\
  \sembr{\jtypef{\Psi}{\tFix{x}{M}}{\tau}}u &= \sfix{\sembr{\jtypef{\Psi,x:\tau}{M}{\tau}}}u  \label[intn]{eq:28}
\end{align}

\begin{proposition}[restate=recsubst,name={}]
  \label{cor:122}
  For all $\jtypef{\Psi,x : \tau}{M}{\tau}$, we have $[\tFix{x}{M}/x]M \equiv \tFix{x}{M}$.
\end{proposition}

For simplicity, the only base types are those of quoted processes.
They are formed by \rn{T\{\}}, and its interpretation is:
\begin{equation}
  \label[intn]{eq:72}
  \sembr{\jisft{\Tproc{a:A}{\overline{a_i:A_i}}}} = \sembr{\overline{a_i:A_i} \vdash a : A}_\bot
\end{equation}
where we abbreviate $\moabc\left[\sembr{\Delta}^+ \times \sembr{a:A}^- \to \sembr{\Delta}^- \times \sembr{a:A}^+\right]$ as $\sembr{\Delta \vdash a : A}$.
The distinction between the two ``bottom'' elements in $\sembr{\overline{a_i:A_i} \vdash a : A}_\bot$ is semantically meaningful.
The genuine bottom element denotes the absence of a value of type $\Tproc{a:A}{\overline{a_i:A_i}}$.
The lifted bottom element $\upim{\lambda x.\bot}$ corresponds to a stuck process that produces no output.

The \rn{I-\{\}} introduction rule quotes processes.
Its denotation is:
\begin{equation}
  \label[intn]{eq:191919}
  \sembr{\jtypef{\Psi}{\tProc{a}{P}{\overline{a_i}}}{\Tproc{a:A}{\overline{a_i:A_i}}}} = \up \circ \sembr{\jtypem{\Psi}{\overline{a_i:A_i}}{P}{a}{A}}.
\end{equation}
Because the unit $\up$ is not strict, quoting respects the semantic distinction between quoted stuck processes $\upim{\lambda x.\bot}$ and the absence $\bot$ of a value of type $\Tproc{a:A}{\overline{a_i:A_i}}$.

The \rn{E-\{\}} elimination rule spawns quoted processes.
Its denotation is:
\begin{equation}
  \label[intn]{eq:115}
  \sembr{\jtypem{\Psi}{\overline{a_i:A_i}}{\tProc{a}{M}{\overline a_i}}{a}{A}} = \down \circ \sembr{\jtypef{\Psi}{M}{\Tproc{a:A}{\overline{a_i:A_i}}}}.
\end{equation}
Two cases are possible when unquoting $M$.
If $\sembr{\jtypef{\Psi}{M}{\Tproc{c:A}{\overline{a_i:A_i}}}}u$ is $\bot$, then $\down(\bot)$ is the constant function $\lambda x.\bot$.
This represents the process that never produces output.
If $\sembr{\jtypef{\Psi}{M}{\Tproc{c:A}{\overline{a_i:A_i}}}}u = \upim{p}$, then $\down(\upim{p}) = p$ is the denotation of some quoted process.

\Cref{eq:191919,eq:115} satisfy the following $\eta$-like property:

\begin{proposition}[restate=quoteunquote,name={}]
  \label{prop:fscd:1}
  For all $\jtypem{\Psi}{\overline{a_i:A_i}}{P}{a}{A}$, we have $P \equiv \tProc{a}{\tProc{a}{P}{\overline a_i}}{\overline a_i}$.
\end{proposition}

A communication of type $\Tand{\tau}{A}$ is a value $v \in \sembr{\tau}$ followed by a communication satisfying $A$.
We use lifting to account for the potential lack of a communication: the bottom element represents the absence of communication.
The value travels in the positive direction, so it only appears in the positive aspect.
\begin{align}
  \sembr{\Xi\vdash\jisst{\Tand{\tau}{A}}}^- &= \sembr{\Xi\vdash \jisst{A}}^-\label[intn]{eq:15104}\\
  \sembr{\Xi\vdash\jisst{\Tand{\tau}{A}}}^+ &= \left(\sembr{\tau}\times\sembr{\Xi\vdash \jisst{A}}^+\right)_\bot\label[intn]{eq:15106}
\end{align}

The process $\tSendV{a}{M}{P}$ sends a functional value on $a$ and continues as $P$.
To send the term $M$ on $a$, we evaluate it under the current environment $u$ to get an element $\sembr{\jtypef{\Psi}{M}{\tau}}u \in \sembr{\tau}$.
Divergence is represented by $\bot_{\sembr{\tau}}$; the other elements represent values of type $\tau$.
If $\sembr{\jtypef{\Psi}{M}{\tau}}u$ represents a value, then we pair it with the output of the continuation process $P$ on $a^+$.
Otherwise, the process transmits nothing.
This gives the clause:
\begin{equation}
  \label[intn]{eq:1005}
  \begin{aligned}
    &\sembr{\jtypem{\Psi}{\Delta}{\tSendV{a}{M}{P}}{a}{\Tand{\tau}{A}}}u(\delta^+,a^-)\\
    &= \begin{cases}
      \bot & \text{if }\sembr{\jtypef{\Psi}{M}{\tau}}u = \bot\\
      \left(\delta^-,\upim{\left(v,a^+\right)}\right) & \text{if }\sembr{\jtypef{\Psi}{M}{\tau}}u = v \neq \bot
    \end{cases}\\
    &\text{where }\sembr{\jtypem{\Psi}{\Delta}{P}{a}{A}}u(\delta^+,a^-) = (\delta^-, a^+).
  \end{aligned}
\end{equation}

The process $\tRecvV{x}{a}{P}$ blocks until it receives a communication on the channel $a$.
If a communication $\upim{(v,\alpha^+)}$ arrives on $a^+$, then the process binds $v$ to $x$ in the environment and continues as $P$ with the remaining communication $\alpha^+$ on $a^+$.
If it receives no message, then it should produce no output, \ie, it should produce $\bot$ on all channels.
This means that its denotation should be strict in the $a^+$ component.
\begin{equation}
  \label[intn]{eq:1006}
  \begin{aligned}
    &\sembr{\jtypem{\Psi}{\Delta, a:\Tand{\tau}{A}}{\tRecvV{x}{a}{P}}{c}{C}}u\\
    &= \strictfn_{a^+}\left( \lambda \left(\delta^+,a^+ : \upim{\left(v, \alpha^+\right)},c^-\right) . \right.\\
    &\qquad\quad\qquad\qquad \left. \sembr{\jtypem{\Psi,x:\tau}{\Delta, a:A}{P}{c}{C}}\upd{u}{x \mapsto v}(\delta^+, \alpha^+, c^-) \right)
  \end{aligned}
\end{equation}
We abuse notation to pattern match on the component $a^+$.
By strictness, we know that it will be an element of the form $\upim{(v,\alpha^+)}$.

\Cref{eq:1006} illustrates a general principle in our semantics.
The bottom element $\bot$ represents the absence of communication.
When a process waits for input on a channel $a$ it uses (it provides), its denotation is strict in $a^+$ (resp.,~$a^-$).

The $\eta$-property for value transmission is subtle because the functional term $M$ might diverge.
The equivalence $\tCut{a}{P}{[M/x]Q} \; \equiv \; \tCut{a}{{(\tSendV{a}{M}{P})}}{{(\tRecvV{x}{a}{Q})}}$ does not hold in general.
If $x$ does not appear free in $Q$, the substitution on the left has no effect and $\tCut{a}{P}{Q}$ runs as usual.
However, if $M$ diverges, then the process on the right is stuck.
Indeed, $\tRecvV{x}{a}{Q}$ waits on $a$ but $\tSendV{a}{M}{P}$ gets stuck evaluating $M$ and sends nothing.
The two processes in the equivalence have equal denotations whenever $M$ converges.
Process equivalence requires the processes to have equal denotations under all environments $u$.
For the equivalence to hold, $M$ must then converge under every environment $u \in \sembr{\Psi}$.
This justifies the statement of \cref{prop:29}.
Its proof uses the substitution property (\cref{prop:fscd:19}) and the Knaster-Tarski-style formulation of the trace.

\begin{proposition}[restate=etaquote,name={}]%
  \label{prop:29}
  For all $\jtypem{\Psi}{\Delta_1}{P}{a}{A}$, all $\jtypem{\Psi,x : \tau}{a : A, \Delta_2}{Q}{c}{C}$, and all  $\jtypef{\Psi}{M}{\tau}$, we have $\tCut{a}{P}{[M/x]Q} \equiv \tCut{a}{(\tSendV{a}{M}{P})}{(\tRecvV{x}{a}{Q})}$ whenever $\sembr{\jtypef{\Psi}{M}{\tau}}u \neq \bot$ for all $u \in \sembr{\Psi}$.
\end{proposition}

\subsection{Manipulating channels}
\label{sec:manip-chann}

Forwarding denotes the function that copies data from $a^+$ to $b^+$ and from~$b^-$~to~$a^-$:
\begin{equation}
  \label[intn]{eq:10}
  \begin{aligned}
    &\sembr{\jtypem{\Psi}{a:A}{\tFwd{b}{a}}{b}{A}}u(a^+ : \alpha, b^- : \beta) = (a^- : \beta, b^+ : \alpha).
  \end{aligned}
\end{equation}

Process composition ``connects'' the common channel $a$ of two communicating processes.
As motivated in \cref{sec:overview-semantics}, we use the trace operator to fix $a$'s positive and negative aspects:
\begin{equation}
  \label[intn]{eq:11}
  \begin{aligned}
    &\sembr{\jtypem{\Psi}{\Delta_1, \Delta_2}{\tCut{a}{P}{Q}}{c}{C}}u\\
    &= \Tr{\left(\sembr{\jtypem{\Psi}{\Delta_1}{P}{a}{A}}u \times \sembr{\jtypem{\Psi}{a : A,\Delta_2}{Q}{c}{C}}u\right)}{a^- \times a^+}.
  \end{aligned}
\end{equation}

Associativity of composition follows from the trace operator axioms:

\begin{proposition}[restate=assoccut,name={}]
  \label{prop:1}
  For all ${\jtypem{\Psi}{\Delta_1}{P_1}{c_1}{C_1}}$, all ${\jtypem{\Psi}{c_1:C_1,\Delta_2}{P_2}{c_2}{C_2}}$, and all ${\jtypem{\Psi}{c_2:C_2,\Delta_3}{P_3}{c_3}{C_3}}$, we have $\tCut{c_1}{P_1}{(\tCut{c_2}{P_2}{P_3})} \equiv {\tCut{c_2}{(\tCut{c_1}{P_1}{P_2})}{P_3}}$.
\end{proposition}

Processes can close channels of type $\Tu$.
The close message is the only communication possible on a channel of type $\Tu$.
The positive aspect of \rn{C$\Tu$} is the constant functor onto the two-element pointed domain $\{\ast\}_\bot = \{ \bot \sqsubseteq \ast \}$.
The element $\ast$ represents the close message, while $\bot$ represents the absence of communication.
All communication on a channel of type $\Tu$ is positive, so the negative aspect is the constant functor onto the terminal object $\{\bot\}$.
\begin{align}
  \sembr{\Xi \vdash \jisst{\Tu}}^- &= \lambda \xi.\{\bot\}\label[intn]{eq:1502}\\
  \sembr{\Xi \vdash \jisst{\Tu}}^+ &= \lambda \xi.\{\ast\}_\bot\label[intn]{eq:1501}
\end{align}

In our asynchronous setting, $\tClose a$ does not wait for a client before sending the close message.
We interpret \rn{$\Tu R$} as the constant function that sends the close message $\ast$ on $a^+$:
\begin{equation}
  \label[intn]{eq:1509}
  \sembr{\jtypem{\Psi}{\cdot}{\tClose a}{a}{\Tu}}u(a^- : \bot) = (a^+ : \ast).
\end{equation}

The process $\tWait{a}{P}$ blocks until it receives the close message, so its denotation is strict in the component $a^+$.
All other communication is handled by $P$.
We interpret \rn{$\Tu L$} by:
\begin{equation}
  \label[intn]{eq:21}
  \begin{aligned}
    &\sembr{\jtypem{\Psi}{\Delta,a : \Tu}{\tWait{a}{P}}{c}{C}}u = \strictfn_{a^+}\left(\lambda (\delta^+,a^+,c^-).(\delta^-,\bot,c^+)\right))\\
    &\text{where } \sembr{\jtypem{\Psi}{\Delta}{P}{c}{C}}u(\delta^+,c^-) = (\delta^-, c^+).
  \end{aligned}
\end{equation}

\begin{proposition}[restate=etaunit,name={}]
  \label{prop:9}
  For all $\jtypem{\Psi}{\Delta}{P}{c}{C}$, we have $P \equiv \tCut{a}{\tClose a}{(\tWait{a}{P})}$.
\end{proposition}

Processes can communicate channels.
We cannot directly observe a channel, only the communications it carries.
For this reason, we treat communications of type $A \Tot B$ as a pair of communications: one for the sent channel and one for the continuation channel.
This is analogous to the denotation of $A \Tot B$ given by Atkey~\citeN{atkey_2017:_obser_commun_seman_class_proces}.
We account for the potential absence of communication by lifting.
\begin{align}
  \sembr{\Xi\vdash\jisst{A \Tot B}}^- &= \sembr{\Xi\vdash\jisst{A}}^- \times \sembr{\Xi\vdash\jisst{B}}^-\label[intn]{eq:13}\\
  \sembr{\Xi\vdash\jisst{A \Tot B}}^+ &= \left(\sembr{\Xi\vdash\jisst{A}}^+ \times \sembr{\Xi\vdash\jisst{B}}^+\right)_\bot\label[intn]{eq:14}
\end{align}

To send the channel $b$ over $a$, the process $\tSendC{a}{b}{P}$ relays the positive communication from $b^+$ to the $\sembr{B}^+$-component of $\sembr{a : B \Tot A}^+$, and the negative communication on the $\sembr{B}^-$-component of $\sembr{a : B \Tot A}^-$ to $b^-$.
The continuation $P$ handles all other communication.
%
%
\begin{equation}
  \label[intn]{eq:17}
  \begin{aligned}
    &\sembr{\jtypem{\Psi}{\Delta, b : B}{\tSendC{a}{b}{P}}{a}{B \Tot A}}u(\delta^+,b^+,(a^-_B,a^-_A)) = \left(\delta^-,a_B^-,\upim{\left(b^+,a_A^+\right)}\right)\\
    &\text{ where }\sembr{\jtypem{\Psi}{\Delta}{P}{a}{A}}u(\delta^+,a_A^-) = (\delta^-,a_A^+).
  \end{aligned}
\end{equation}

The client $\tRecvC{b}{a}{Q}$ blocks until it receives a channel on $a$.
When it receives a communication $\upim{(a_B^+, a_A^+)}$ on $\sembr{a : B \Tot A}^+$, it unpacks it into the two positive channels $\sembr{a : A, b : B}^+$ expected by $Q$.
It then repacks the negative communication $Q$ produces on $\sembr{a : A, b : B}^-$ and relays it over~$\sembr{a : B \Tot A}^-$.
\begin{equation}
  \label[intn]{eq:26}
  \begin{aligned}
    &\sembr{\jtypem{\Psi}{\Delta, a : B \Tot A}{\tRecvC{b}{a}{Q}}{c}{C}}u(\delta^+,a^+,c^-)\\
    &= \strictfn_{a^+}\left(\lambda (\delta^+,a^+ : \upim{(a_B^+, a_A^+)},c^-) . (\delta^-,(b^-,a^-),c^+) \right)\\
    &\text{where }\sembr{\jtypem{\Psi}{\Delta, a : A, b : B}{Q}{c}{C}}u(\delta^+,a_A^+,a_B^+,c^-) = (\delta^-,a^-,b^-,c^+).
  \end{aligned}
\end{equation}

\begin{proposition}[restate=etatensor,name={}]
  \label{prop:522}
  For all $\jtypem{\Psi}{\Delta_1}{P}{a}{A}$ and $\jtypem{\Psi}{a : A, \Delta_2, b : B}{Q}{c}{C}$, we have $\tCut{a}{P}{Q} \equiv \tCut{b}{\left(\tSendC{a}{b}{P}\right)}{\left(\tRecvC{b}{a}{Q}\right)}$.
\end{proposition}

\begin{example}
  \label{ex:fscd:11}
  The process below blocks until it receives a channel $a$ of type $\Tu$ over the channel $b$, at which point the type of $b$ becomes $\Tu$.
  Then, the process waits for the close messages on $a$ and $b$ before closing $c$.
  The element $\upim{(\ast,\ast)} \in \sembr{\Tu \Tot \Tu}^+ = (\sembr{\Tu}^+ \times \sembr{\Tu}^+)_\bot$ corresponds to receiving the channel $a$, the close message on $a$, and the close message on $b$.
  The element $\upim{(\bot,\bot)}$ corresponds to receiving $a$ but no close messages, while the elements $\upim{(\ast,\bot)}$ and $\upim{(\bot,\ast)}$ correspond to receiving $a$ and one close message.
  The element $\bot$ means that $a$ is never received.
  We see that the process only closes $c$ in the first case:
  \begin{align*}
    &\sembr{\jtypem{\cdot}{b : \Tu \Tot \Tu}{\tRecvC{a}{b}{\tWait{a}{\tWait{b}{\tClose{c}}}}}{c}{\Tu}}\bot(b^+ : \beta, c^- : \bot)\\
    &=
    \begin{cases}
      (b^- : (\bot,\bot), c^+ : \ast) & \text{if }\beta = \upim{(\ast, \ast)}\\
      (b^- : (\bot,\bot), c^+ : \bot) & \text{otherwise.}
    \end{cases}
  \end{align*}
\end{example}

\subsection{Shifts in Polarity}
\label{sec:shifts-polarity}

Synchronization is encoded using ``polarity shifts''.
A communication of type $\Tds A$ is a synchronization message (the ``shift'' message) followed by a communication of type $A$.
``Downshifting'' $A$ to $\Tds A$ introduces only \textit{positive} communication (the ``shift'' message), so the negative aspect of $\Tds A$ is the same as the negative aspect of $A$.
We interpret $\sembr{\Tds A}^+$ by lifting $\sembr{A}^+$.
The element $\bot_{\sembr{\Tds A}^+}$ captures the absence of communication.
The elements $\upim{a}$ for $a \in \sembr{A}^+$ capture a $\ms{shift}$ message followed by a communication satisfying $A$.
\begin{align}
  \sembr{\Xi\vdash\jisst{\Tds A}}^- &= \sembr{\Xi\vdash\jisst{A}}^-\label[intn]{eq:1506}\\
  \sembr{\Xi\vdash\jisst{\Tds A}}^+ &= \sembr{\Xi\vdash\jisst{A}}^+_\bot\label[intn]{eq:1507}
\end{align}

In our asynchronous setting, the shift message is always sent.
This corresponds to lifting the output of $P$ on the $a^+$ component.
We interpret \rn{$\Tds{} R$} as:
\begin{equation}
  \label[intn]{eq:59}
  \begin{aligned}
    &\sembr{\jtypem{\Psi}{\Delta}{\tSendS{a}{P}}{a}{\Tds A}}u\\
    &= \left(\ms{id} \times \left(a^+ : \up\right)\right) \circ \sembr{\jtypem{\Psi}{\Delta}{P}{a}{A}}u.
  \end{aligned}
\end{equation}

The client blocks until it receives the shift message on $a^+$.
We lower $\sembr{A}^+_\bot$ to $\sembr{A}^+$ to extract the positive communication expected by $P$.
\begin{equation}
  \label[intn]{eq:56}
  \begin{aligned}
    &\sembr{\jtypem{\Psi}{\Delta,a : \Tds A}{\tRecvS{a}{P}}{c}{C}}u\\
    &= \strictfn_{a^+}\left(\sembr{\jtypem{\Psi}{\Delta, a : A}{P}{c}{C}}u \circ \left(\ms{id} \times \left(a^+ : \down\right)\right)\right)
  \end{aligned}
\end{equation}

\begin{proposition}[restate=etadshift,name={}]
  \label{prop:7}
  For all $\jtypem{\Psi}{\Delta_1}{P}{a}{A}$ and $\jtypem{\Psi}{a : A, \Delta_2}{Q}{c}{C}$, we have $\tCut{a}{P}{Q} \equiv \tCut{a}{(\tSendS{a}{P})}{(\tRecvS{a}{Q})}$.
\end{proposition}

\begin{example}
  \label{ex:fscd:3}
  Upshifts are the polar duals of downshifts.
  The following process waits for its client to synchronize with it before closing the channel.
  The protocol $\Tus\Tu$ has denotations $\sembr{\Tus\Tu}^- = \sembr{\Tu}^-_{\bot} = \{\bot\}_\bot$ and $\sembr{\Tus\Tu}^+ = \sembr{\Tu}^+ = \{\ast\}_\bot$.
  The element $\upim{\bot} \in \sembr{\Tus\Tu}^-$ captures the synchronizing shift message.
  The process closes $a$ if and only if it receives the shift message:
  \[
    \sembr{\jtypem{\cdot}{\cdot}{\tRecvS{a}{\tClose{a}}}{a}{\Tus{\Tu}}}\bot(a^- : \alpha) =
    \begin{cases}
      (a^+ : \bot) & \text{if }\alpha = \bot\\
      (a^+ : \ast) & \text{if }\alpha = \upim{\bot}.
    \end{cases}
  \]
\end{example}

\subsection{Making Choices}
\label{sec:making-choices}

The internal choice type $\Tplus \{ l : A_l \}_{l \in L}$ prescribes a choice between session types $\{ A_l \}_{l \in L}$ ($L$ finite).
A communication of type $\Tplus \{ l : A_l \}_{l \in L}$ is a label $k \in L$ sent in the positive direction followed by a communication satisfying $A_k$.
Denotationally, this corresponds to tagging a communication $a_k \in \sembr{A_k}^+$ with the label $k$.
Tagged communications $(k, a_k)$ are the elements of the disjoint union $\biguplus_{l \in L} \sembr{A_l}^+$.
To account for the potential lack of communication, we lift this disjoint union.
This lifted disjoint union is isomorphic to the coalesced sum $\bigoplus_{l \in L} \sembr{A_l}^+_\bot$.
Coalesced sums are coproducts in $\moabcs$, and we define the interpretation using a coalesced sum to make this structure evident.
Explicitly, its elements are $\bot$ and $(k,\upim{a_k})$ for $k \in L$ and $a_k \in \sembr{A_k}^+$.
The client does not know a priori which branch it will take: it must be ready to send negative information for every possible branch.
This justifies \cref{eq:202020}.
\begin{align}
  \sembr{\Xi\vdash\jisst{\Tplus\{l:A_l\}_{l \in L}}}^- &= \prod_{l \in L} \sembr{\Xi\vdash\jisst{A_l}}^-\label[intn]{eq:202020}\\
  \sembr{\Xi\vdash\jisst{\Tplus\{l:A_l\}_{l \in L}}}^+ &= \bigoplus_{l \in L} \sembr{\Xi\vdash\jisst{A_l}}^+_\bot\label[intn]{eq:222222}
\end{align}

\begin{samepage}%
To interpret \rn{$\Tplus R_k$}, we extract from $a^-$ the negative information $a_k^-$ required by the continuation process $P$.
Output on $a^+$ is the label $k$ followed by the output of $P$ on $a^+$:
\begin{equation}
  \label[intn]{eq:52}
  \begin{aligned}
    &\sembr{\jtypem{\Psi}{\Delta}{\tSendL{a}{k}{P}}{a}{\Tplus\{l:A_l\}_{l \in L}}}u\left(\delta^+, \left(a_l^-\right)_{l \in L}\right) = \left(\delta^-, \left(k, \upim{a_k^+}\right)\right)\\
    &\text{where }\sembr{\jtypem{\Psi}{\Delta}{P}{a}{A_k}}u\left(\delta^+, a_k^-\right) = \left(\delta^-, a_k^+\right).
  \end{aligned}
\end{equation}

The client $\tCase{a}{\left\{l \Rightarrow P_l\right\}_{l \in L}}$ blocks until it receives a communication $(k, \upim{a_k^+})$ on $a^+$, and then it takes the branch $P_k$.
\end{samepage}%
To transmit $P_k$'s output $a_k^-$ on $a^-$ back to the provider, we place $a_k^-$ in the $k$ component of the product sent on $a^-$.
The other branches were not taken and produced no communication, so their respective components in $a^-$ are filled with $\bot$.
\begin{equation}
  \label[intn]{eq:53}
  \begin{aligned}
    &\sembr{\jtypem{\Psi}{\Delta,a:\Tplus\{l:A_l\}_{l \in L}}{\tCase{a}{\left\{l \Rightarrow P_l\right\}_{l \in L}}}{c}{C}}u\\
    &= \strictfn_{a^+}\left(\lambda \left(\delta^+, a^+ : \left(k, \upim{a_k^+}\right), c^-\right).\left(\delta^-, a^- : \left(k : a_k^-, l \neq k : \bot\right)_{l \in L}, c^+\right)\right)\\
    &\text{where }\sembr{\jtypem{\Psi}{\Delta,a:A_k}{P_k}{c}{C}}u(\delta^+, a_k^+, c^-) = (\delta^-, a_k^-, c^+)
  \end{aligned}
\end{equation}

\begin{proposition}[restate=etaplus,name={}]
  \label{prop:1253}
  Let $k \in L$.
  If ${\jtypem{\Psi}{\Delta_1}{P}{a}{A_k}}$, and ${\jtypem{\Psi}{a : A_l, \Delta_2}{Q_l}{c}{C}}$ for all $l \in L$, then
  $\tCut{a}{P}{Q_k} \equiv \tCut{a}{\left(\tSendL{a}{k}{P}\right)}{\left(\tCase{a}{\left\{l \Rightarrow Q_l\right\}_{l \in L}}\right)}$.
\end{proposition}

\begin{example}
  We build on \cref{ex:fscd:3}.
  External choices $\Tamp\{l : A_l\}_{l \in L}$ are the polar duals of internal choices.
  Let $A = \Tamp\{\mt{j} : \Tus\Tu, \mt{k} : \Tus\Tu\}$.
  A provider of $A$ receives a label and a synchronizing shift before closing the channel.
  The elements $(l,\upim{\upim{\bot}}) \in \sembr{A}^-$ correspond to receiving the label $l$ over $a$ followed by a shift, while the elements $(l,\upim{\bot})$ correspond to receiving $l$ but no shift.
  The communication on $a^+$ depends on the label $l$ received: the close message is in the $l$ component of the output on $a^+$.
  \begin{align*}
    &\sembr{\jtypem{\cdot}{\cdot}{\tCase{a}{\left\{l \Rightarrow \tRecvS{a}{\tClose{a}}\right\}_{l \in \{\mt{j},\mt{k}\}}}}{a}{A}}\bot(a^- : \alpha)\\
    &=
      \begin{cases}
        (a^+ : (\mt{j} : \ast, \mt{k} : \bot)) & \text{if } \alpha = (\mt{j},\upim{\upim{\bot}})\\
        (a^+ : (\mt{j} : \bot, \mt{k} : \ast)) & \text{if } \alpha = (\mt{k},\upim{\upim{\bot}})\\
        (a^+ : (\mt{j} : \bot, \mt{k} : \bot)) & \text{if } \alpha = (l,\upim{\bot}) \text{ for }l \in \{\mt{j},\mt{k}\}\text{ or if } \alpha = \bot
      \end{cases}
  \end{align*}
\end{example}

\subsection{Recursive Types}
\label{sec:recursive-types}

The substitution property (\cref{prop:11}) forces the denotation of the variable rule \rn{CVar}:
\begin{align}
  \sembr{\Xi,\jisst[p]{\alpha}\vdash\jisst[p]{\alpha}}^q &= \pi^{\Xi,\alpha}_\alpha \quad (q \in \{{-},{+}\})\label[intn]{eq:fscd:7}
\end{align}

We interpret recursive types by parametrized solutions of recursive domain equations.
Every locally continuous functor $G : \moabcs \to \moabcs$ has a canonical fixed point $\FIX(G)$ in $\moabcs$.
Given a locally continuous functor $F : \sembr{\Xi} \times \moabcs \to \moabcs$, the mapping $D \mapsto \FIX(F(D, {-}))$ extends to a locally continuous functor $\sfix{F} : \sembr{\Xi} \to \moabcs$ \cite[Proposition~5.2.7]{abramsky_jung_1995:_domain_theor}.
The fixed-point property $F(D, \sfix{F} D) \cong \sfix{F} D$ is witnessed by a canonical natural isomorphism $\ms{Fold} : F \circ \langle \ms{id}_{\sembr{\Xi}}, \sfix F \rangle \nto \sfix{F}$ with inverse $\ms{Unfold}$.
The rule \rn{C$\rho$} denotes:
\begin{align}
  \sembr{\Xi\vdash\jisst{\Trec{\alpha}{A}}}^p &= \sfix{\left(\sembr{\Xi,\jisst{\alpha}\vdash \jisst{A}}^p\right)}\quad (p \in \{{-},{+}\})\label[intn]{eq:20051}
\end{align}
Let $\ms{Fold}^p \,{:}\, \sembr{\Xi,\jisst{\alpha}\vdash \jisst{A}}^p \circ \langle \ms{id}_{\sembr{\Xi}}, \!\sfix{\left(\sembr{\Xi,\jisst{\alpha}\vdash \jisst{A}}^p\right)} \rangle \nto \sembr{\jstype{\Xi}{\Trec{\alpha}{A}}}^p$ be the canonical natural isomorphism and let $\ms{Unfold}^p$ be its inverse.
These isomorphisms capture the semantic folding and unfolding of recursive types.
Indeed, by \cref{prop:11,eq:fscd:7,eq:20051}, the domain of $\ms{Fold}^p$ is:
\[
  \sembr{\Xi,\jisst{\alpha}\vdash \jisst{A}}^p \circ \left\langle \ms{id}_{\sembr{\Xi}}, \sfix{\left(\sembr{\jstype{\Xi}{\Trec{\alpha}{A}}}^p\right)} \right\rangle = \sembr{\Xi\vdash [\Trec{\alpha}{A}/\alpha]\jisst{A}}^p.
\]

\begin{samepage}
Processes unfold recursive types by transmitting unfold messages.
Semantically, this denotes pre- and post-composition with $\Fold^p$ and $\Unfold^p$.
We interpret \rn{$\rho^+R$} and \rn{$\rho^+L$}~by:
\begin{align}
  \begin{split}
    &\sembr{\jtypem{\Psi}{\Delta}{\tSendU{a}{P}}{a}{\Trec{\alpha}{A}}}u\\
    &= \left(\ms{id} \times \left(a^+ : \ms{Fold}^+\right)\right)
    \circ \sembr{\jtypem{\Psi}{\Delta}{P}{a}{[\Trec{\alpha}{A}/\alpha]A}}u
    \circ \left(\ms{id} \times \left(a^- : \ms{Unfold}^-\right)\right),\label[intn]{eq:fossacs:2}
  \end{split}\\
  \begin{split}
    &\sembr{\jtypem{\Psi}{\Delta, a : \Trec{\alpha}{A}}{\tRecvU{a}{P}}{c}{C}}u\\
    &= \left(\ms{id} \!\times \!\left(a^- : \ms{Fold}^-\right)\right)
    \circ \sembr{\jtypem{\Psi}{\Delta, a : [\Trec{\alpha}{A}/\alpha]A}{P}{c}{C}}u
    \circ \left(\ms{id} \!\times \!\left(a^+ : \ms{Unfold}^+\right)\right)\!.\label[intn]{eq:fossacs:3}
  \end{split}
\end{align}

\begin{proposition}[restate=etafunfold,name={}]
  \label{prop:fscd:18}
  If $\jtypem{\Psi}{\Delta_1}{P}{a}{[\Trec{\alpha}{A}/\alpha]A}$ and ${\jtypem{\Psi}{a : [\Trec{\alpha}{A}/\alpha]A, \Delta_2}{Q}{c}{C}}$, then $\tCut{a}{P}{Q} \equiv \tCut{a}{(\tSendU{a}{P})}{(\tRecvU{a}{Q})}$.
\end{proposition}
\end{samepage}

\subsection{Structural Properties}
\label{sec:struct-prop}

Our semantics respects the exchange rule because we interpret structural contexts as indexed products.
It also respects weakening (\cref{prop:extended:8,prop:extended:9}) and substitution (\cref{prop:11,prop:fscd:19}).
These propositions follow by induction on the derivations.

\begin{proposition}[restate=cohsestyp,name=Coherence of Session Types]
  \label{prop:extended:8}
  If $\Xi \vdash \jisst[q]{A}$, then $\sembr{\Xi,\Theta\vdash \jisst[q]{A}} = \sembr{\Xi \vdash \jisst[q]{A}}^p\pi^{\Xi,\Theta}_\Xi$ for all $p \in \{{-},{+}\}$.
\end{proposition}

\begin{proposition}[restate=cohterpro,name=Coherence of Terms and Processes]
  \label{prop:extended:9}
  If $\jtypef{\Psi}{M}{\tau}$, then $\sembr{\jtypef{\Phi,\Psi}{M}{\tau}} = \sembr{\jtypef{\Psi}{M}{\tau}} \circ \pi^{\Phi,\Psi}_\Psi$.
  If ${\jtypem{\Psi\!}{\!\Delta\!}{\!P\!}{\!a\!}{\!A}}$, then $\sembr{\jtypem{\Phi,\Psi\!}{\!\Delta\!}{\!P\!}{\!a\!}{\!A}}\,{=}\,\sembr{\jtypem{\Psi\!}{\!\Delta\!}{\!P\!}{\!a\!}{\!A}} \circ \pi^{\Phi,\Psi}_\Psi$.
\end{proposition}

\begin{proposition}[restate=ssst,name=Semantic Substitution of Session Types]\label{prop:11}
  Let $\Xi = \jisst[p_1]{\alpha_1},\dotsc,\jisst[p_n]{\alpha_n}$.
  For all $p \in \{{-},{+}\}$ and all choices of types $\Theta \vdash \jisst[p_i]{A_i}$ ($1 \leq i \leq n$), if $\Xi \vdash \jisst[q]{B}$, then
  \[
    \sembr{\Theta \vdash \jisst[q]{[\vec A/\vec \alpha]B}}^p = \sembr{\Xi \vdash \jisst[q]{B}}^p \circ \langle \sembr{\Theta \vdash  \jisst[p_i]{A_i}}^p \mid 1 \leq i \leq n \rangle.
  \]
\end{proposition}

\begin{proposition}[restate=ssft,name=Semantic Substitution of Functional Terms]
  \label{prop:fscd:19}
  Let $\Psi = {x_1:\tau_1},\dotsc,{x_n:\tau_n}$.
  For all choices of terms $\jtypef{\Phi}{M_i}{\tau_i}$ ($1 \leq i \leq n$),
  if $\jtypef{\Psi}{N}{\tau}$ and $\jtypem{\Psi}{\Delta}{P}{c}{C}$, then
  \begin{align*}
    \sembr{\jtypef{\Phi}{[\vec{M}/\vec{x}]N}{\tau}} &= \sembr{\jtypef{\Psi}{N}{\tau}} \circ \langle \sembr{\jtypef{\Phi}{M_i}{\tau_i}} \mid 1 \leq i \leq n \rangle,\\
    \sembr{\jtypem{\Phi}{\Delta}{[\vec{M}/\vec{x}]P}{c}{C}} &= \sembr{\jtypem{\Psi}{\Delta}{P}{c}{C}} \circ \langle \sembr{\jtypef{\Phi}{M_i}{\tau_i}} \mid 1 \leq i \leq n \rangle.
  \end{align*}
\end{proposition}




\section{Illustrative Example: Flipping Bit Streams}
\label{sec:bit-streams}

We illustrate our semantics by studying the recursive bit-flipping process from the introduction.
We will show that flipping the bits in a stream twice is equivalent to forwarding it.
The bit stream protocol is specified by the session type $\mt{bits} = \Trec{\beta}{\Tplus \{ \mt{0} : \beta, \mt{1} : \beta \}}$.
It denotes the domains $\sembr{\mt{bits}}^+ = \FIX\left(X \mapsto \left(\left(\mt 0 : X_\bot\right) \oplus \left(\mt 1 : X_\bot \right)\right)\right)$ and $\sembr{\mt{bits}}^- = \{\bot\}$.
Its unfolding is $\mt{BITS} = \Tplus\{\mt{0} : \mt{bits}, \mt{1} : \mt{bits}\}$.
There are canonical isomorphisms
\begin{align*}
  &\ms{Unfold}^+ : \sembr{\mt{bits}}^+ \to \left(\left(\mt{0}:\sembr{\mt{bits}}^+_\bot\right) \oplus \left(\mt{1}:\sembr{\mt{bits}}^+_\bot\right)\right)\\
  &\ms{Unfold}^- : \{\bot\} \to \{(\mt{0} : \bot, \mt{1} : \bot)\}
\end{align*}
with respective inverses $\ms{Fold}^+$ and $\ms{Fold}^-$.
Write $\cons{0}{\alpha}$ and $\cons{1}{\alpha}$ for $\Fold^+((\mt{0},\upim{\alpha}))$ and $\Fold^+((\mt{1},\upim{\alpha}))$, respectively.

We desugar the bit-flipping process.
Write $\tilde l$ for the complement of $l \in \{\mt{0},\mt{1}\}$.
Desugaring the quoted process gives a term $\jtypef{\cdot}{\mt{flip}}{\Tproc{f:\mt{bits}}{b:\mt{bits}}}$ where $\mt{flip}$ is:
\[
  \tFix{F}{
    \tProc{f}{\tSendU{f}{\tRecvU{b}{\tCase{b}{\{l \Rightarrow \tSendL{f}{\tilde l}{\tProc{f}{F}{b}}\}}_{l \in \{\mt{0},\mt{1}\}}}}}{b}
  }.
\]
Its denotation is $\sembr{\jtypef{\cdot}{\mt{flip}}{\Tproc{a : \mt{bits}}{b : \mt{bits}}}}\bot = \up(\fix(\Phi))$ where the function $\Phi : \sembr{b : \mt{bits} \vdash a : \mt{bits}} \to \sembr{b : \mt{bits} \vdash a : \mt{bits}}$ is given by:
\[
  \Phi(r) = \strictfn_{b^+} \left(\begin{aligned}
      \lambda (b^+,\bot).&
      \begin{cases}
        \left(\bot, \cons{1}{a_{\mt 0}^+}\right) & \text{if } b^+ = \cons{0}{b_{\mt 0}^+}\\
        \left(\bot,  \cons{0}{a_{\mt 1}^+}\right) & \text{if } b^+ = \cons{1}{b_{\mt 1}^+}
      \end{cases}\\
      &\quad\text{where }r(b_l^+, \bot) = (\bot, a_l^+) \text{ for } l \in L
    \end{aligned}
  \right)
\]

Consider the composition $\jtypem{\cdot}{a : \mt{bits}}{\tProc{t}{\mt{flip}}{a}; \tProc{b}{\mt{flip}}{t}}{b}{\mt{bits}}$ of two bit-flipping processes.
By the Knaster-Tarski-style formulation of the trace operator,
$\sembr{\jtypem{\cdot}{a : \mt{bits}}{\tProc{t}{\mt{flip}}{a}; \tProc{b}{\mt{flip}}{t}}{b}{\mt{bits}}}\bot = \fix(\Phi) \circ \sigma \circ \fix(\Phi)$ where $\sigma$ is the obvious relabelling.
We want to show that the composition is equivalent to forwarding.
By \cref{eq:10}, this means showing for all $(a^+,\bot) \in \sembr{a : \mt{bits}}^+ \times \sembr{b : \mt{bits}}^-$ that:
\begin{equation}
  \sembr{\jtypem{\cdot}{a : \mt{bits}}{\tProc{t}{\mt{flip}}{a}; \tProc{b}{\mt{flip}}{t}}{b}{\mt{bits}}}\bot(a^+,\bot) = (\bot, a^+).\label{eq:fscd:10}
\end{equation}

Given an $a \in \sembr{\mt{bits}}^+$, let $\fd a$ be given by $\left(\fix(\Phi) \circ \sigma \circ \fix(\Phi)\right)(a, \bot) = (\bot, \fd a)$.
It is easy to check that $\fd \bot = \bot$, $\fd{(\cons{0}{\alpha})} = \cons{0}{\fd \alpha}$, and $\fd{(\cons{1}{\alpha})} = \cons{1}{\fd \alpha}$ for all $\alpha \in \sembr{\mt{bits}}^+$.

To show \cref{eq:fscd:10}, we must show that $a = \fd a$ for all $a \in \sembr{\mt{bits}}^+$.
We use a co\-induction principle by Pitts~\cite{pitts_1994:_co_induc_princ}.
Given a relation $\R$ on $\sembr{\mt{bits}}^+_\bot$, let $\Rh$ be the relation on $\sembr{\mt{BITS}}^+_\bot$ where $a \Rh b$ if and~only~if for all $\alpha \in \sembr{\mt{bits}}^+$:
\begin{itemize}
\item if $a = \upim{(\mt{0}, \upim{\alpha})}$, then $b = \upim{(\mt{0}, \upim{\beta})}$ for some $\beta \in \sembr{\mt{bits}}^+$ with $\upim{\alpha} \R \upim{\beta}$;
\item if $a = \upim{(\mt{1}, \upim{\alpha})}$, then $b = \upim{(\mt{1}, \upim{\beta})}$ for some $\beta \in \sembr{\mt{bits}}^+$ with $\upim{\alpha} \R \upim{\beta}$.
\end{itemize}
The relation $\R$ is a \defin{simulation} if $a \R b$ implies $\ms{Unfold}^+_\bot(a) \Rh \ms{Unfold}^+_\bot(b)$.
By \cite[Theorem~2.5]{pitts_1994:_co_induc_princ}, $a \sqsubseteq b$ in $\sembr{\mt{bits}}^+_\bot$ if and only if there exists a simulation $\R$ such that $a \R b$.

Let ${\R} = \left\{ (\upim{\bot},\upim{\bot}), (\upim{\cons{0}{\alpha}}, \upim{\cons{0}{\fd \alpha}}), (\upim{\cons{1}{\alpha}}, \upim{\cons{1}{\fd \alpha}}) \mid \alpha \in \sembr{\mt{bits}}^+ \right\} \subseteq \sembr{\mt{bits}}^+_\bot \times \sembr{\mt{bits}}^+_\bot$.
A case analysis shows that $\R$ is a simulation.
To show that $a \sqsubseteq \fd a$ for $a \in \sembr{\mt{bits}}^+$, it is sufficient to show that $\upim{a} \R \upim{\fd a}$.
If $a = \bot$, then $\fd a = \bot$ by strictness.
If $a = \cons{0}{\alpha}$, then $\fd a = \cons{0}{\fd \alpha}$.
If $a = \cons{1}{\alpha}$, then $\fd a = \cons{1}{\fd \alpha}$.
In all cases, $\upim{a} \R \upim{\fd a}$ by definition of $\R$.
It follows that $a \sqsubseteq \fd a$.
The relation $\R$ is clearly a simulation if and only if $\op\R$ is a simulation, so $\fd a \mathrel{\op\R} a$ and $\fd a \sqsubseteq a$.
We conclude $a = \fd a$ as desired.
It follows that flipping a stream twice is equivalent to forwarding (the identity process).


\section{Related and Future Work}
\label{sec:related-work}

Honda~\citeN{honda_1993:_types_dyadic_inter} and Takeuchi et~al.~\citeN{takeuchi_1994:_inter_based_languag} introduced session types to describe sessions of interaction.
Caires and Pfenning~\citeN{caires_pfenning_2010:_session_types_intuit_linear_propos} observed a proofs-as-programs correspondence between the session-typed $\pi$-calculus and intuitionistic linear logic, where the \rn{Cut} rule captures process communication.
Toninho et~al.~\citeN{toninho_2013:_higher_order_proces_funct_session} built on this correspondence and introduced SILL's monadic integration between functional and synchronous message-passing programming.
They specified SILL's operational behaviour using a substructural operational semantics (SSOS).
Gay and Vasconcelos~\citeN{gay_vasconcelos_2009:_linear_type_theor} introduced asynchronous communication for session-typed languages.
They used an operational semantics and buffers to model asynchronicity.
Pfenning and Griffith~\citeN{pfenning_griffith_2015:_polar_subst_session_types} observed that the polarity of a type determines the direction of communication along a channel.
They observed that synchronous communication can be encoded in an asynchronous setting using explicit shift operators.
Polarized SILL's operational behaviour is specified by a SSOS~\citeN{toninho_2013:_higher_order_proces_funct_session,pfenning_griffith_2015:_polar_subst_session_types}.
In ongoing work, we seek to show that our denotational semantics agrees with this SSOS.
Concretely, we extend this SSOS to give polarized SILL an operational notion of observation.
Then, we seek to show that processes are denotationally equivalent if and only if they are observationally equivalent.

Wadler~\citeN{wadler_2014:_propos_as_session} introduced ``Classical Processes'' (CP), a proofs-as-programs interpretation of classical linear logic that builds on the ideas of Caires and Pfenning~\citeN{caires_pfenning_2010:_session_types_intuit_linear_propos}.
CP supports replication but not recursion.
Though CP does not natively support functional programming, Wadler gives a translation for GV, a linear functional language with pairs but no recursion, into CP.
In contrast, polarized SILL uniformly integrates functional programming and message-passing concurrency.
CP has a synchronous communication semantics and does not have an explicit treatment of polarities.
Polarized SILL has an asynchronous communication semantics, and synchronous communication is encoded using polarity shifts.

Atkey~\citeN{atkey_2017:_obser_commun_seman_class_proces} gave a denotational semantics for CP, where types are interpreted as sets and processes are interpreted as relations over these.
Because processes in CP are proof terms for classical linear logic, the interpretation of processes is identical to the relational semantics of proofs in classical linear logic~\cite{barr_1991:_auton_categ_linear_logic}.
Our jump from sets and relations to domains and continuous functions was motivated by two factors.
First, domains provide a natural setting for studying recursion.
Second, we believe that monotonicity and continuity are essential properties for a semantics of processes with infinite data, and it is unclear how to capture these properties in a relational setting.
Atkey interpreted process composition as relational composition.
Our use of traces is more complex, but we believe that known trace identities make it tractable.
We believe that the extra complexity is justified by SILL's more complex behavioural phenomena.

Our semantics generalizes Kahn's stream-based semantics for deterministic networks~\cite{kahn_1974:_seman_simpl_languag_paral_progr}.
A deterministic network is graph whose nodes are deterministic processes, and whose edges are unidirectional channels.
Each channel carries values of a single fixed simple type, e.g., integers or booleans.
Semantically, channels denote domains of sequences of values, and processes denote continuous functions from input channels to output channels.
Our semantics generalizes this to allow for bidirectional, session-typed communication channels.
Satisfactorily generalizing Kahn-style semantics to handle non-determinism is difficult~\cite{broy_1988:_nondet_data_flow_progr,keller_panangaden_1985:_seman_networ_contain_indet_operat,panangaden_1985:_abstr_inter_indet,panangaden_shanbhogue_1992:_expres_power_indet_dataf_primit,stark_1987:_concur_trans_system,stark_1990:_simpl_gener_kahns}, partly due to the Keller~\cite{keller_1977:_denot_model_paral} and Brock-Ackerman~anomalies~\cite{brock_ackerman_1981:_scenar}.

Our process interpretations exist within a ``wave''-style~\cite{abramsky_1996:_retrac_some_paths_proces_algeb} geometry of interaction (GoI) construction~\cite[Definition~2.6]{abramsky_2002:_geomet_inter_linear_combin_algeb}.
Indeed, the objects of the GoI construction $\mathcal{G}(\moabc)$ are pairs $(A^+, A^-)$ of objects $A^+$ and $A^-$ of $\moabc$.
Morphisms $f : (A^+, A^-) \to (B^+, B^-)$ of $\mathcal{G}(\moabc)$ are morphisms $\hat f : A^+ \times B^- \to A^- \times B^+$ of $\moabc$.
Given a morphism $g : (B^+, B^-) \to (C^+, C^-)$, the composition $g \circ f$ is defined by $\Trop^{B^- \times B^+}_{A^+ \times C^-,A^- \times C^+}(\hat g \times \hat f)$.
\Cref{eq:11} is exactly the composition $\sembr{\jtypem{\Psi}{a : A,\Delta_2}{Q}{c}{C}}u \circ \sembr{\jtypem{\Psi}{\Delta_1}{P}{a}{A}}u$ in $\mathcal{G}(\moabc)$.
Though the categorical setting is identical, our goals and approach are different from prior work using the geometry of interaction.
For example, Abramsky and Jagadeesan~\cite{abramsky_jagadeesan_1994:_new_found_geomet_inter} give a type-free interpretation of classical linear logic where all types denote the same ``universal domain''~\cite[p.~69]{abramsky_jagadeesan_1994:_new_found_geomet_inter}.
Abramsky et~al.~\cite{abramsky_2002:_geomet_inter_linear_combin_algeb} use GoI constructions to give an algebraic framework for Girard's Geometry of Interactions~\cite{girard_1989:_geomet_inter,girard_1990:_geomet_inter,girard_1995:_geomet_inter_iii}.
We instead give a semantics that captures the computational aspects of a programming language with recursion.

Castellan and Yoshida~\citeN{castellan_yoshida_2019:_two_sides_same_coin} gave a game semantics interpretation of the session $\pi$-calculus with recursion.
It is fully abstract relative to a barbed congruence notion of behavioural equivalence.
Session types denote event structures that encode games and that are endowed with an $\omega$-cpo structure.
Open types denote continuous maps between these and recursive types denote least fixed points.
Open processes are interpreted as continuous maps that describe strategies.
We conjecture that our semantics could be related via barbed congruence.

Kokke et~al.~\citeN{kokke_2019:_better_late_than_never} introduced ``hypersequent classical processes'' (HCP).
HCP is a revised proofs-as-processes interpretation between classical linear logic and the $\pi$-calculus.
Building on Atkey's~\cite{atkey_2017:_obser_commun_seman_class_proces} semantics for CP, they gave HCP a denotational semantics using Brzozowski derivatives~\cite{brzozowski_1964:_deriv_regul_expres}.
Their semantics is fully abstract relative to their notions of bisimilarity and barbed congruence.
It is unclear how to extend their semantics~to~handle~recursion.




\section*{Acknowledgments}

This work is funded in part by a Natural Sciences and Engineering Research Council of Canada Postgraduate Scholarship.
The author thanks anonymous referees, Robert Atkey, Stephen Brookes, and Frank Pfenning for their comments.

%
%
\bibliographystyle{plainurl}
\bibliography{fscd}

\appendix

\section{Rules for Polarized SILL}
\label{sec:rules-polarized-sill}

For ease of reference, we collect all of the rules for polarized SILL in this appendix.

\subsection{Rules for Term Formation}
\label{sec:rules-term-formation}

\[
  \infer[\rn{I-\{\}}]{
    \jtypef{\Psi}{\tProc{a}{P}{\overline{a_i}}}{\Tproc{a:A}{\overline{a_i:A_i}}}
  }{
    \jtypem{\Psi}{\overline{a_i:A_i}}{P}{a}{A}
  }
\]
\[
  \infer[\rn{F-Var}]{\jtypef{\Psi, x: \tau}{x}{\tau}}{\mathstrut}
  \quad
  \infer[\rn{F-Fix}]{
    \jtypef{\Psi}{\tFix{x}{M}}{\tau}
  }{
    \jtypef{\Psi,x:\tau}{M}{\tau}
  }
\]
\[
  \infer[\rn{F-Fun}]{
    \jtypef{\Psi}{\lambda x : \tau.M}{\tau \to \sigma}
  }{
    \jtypef{\Psi, x:\tau}{M}{\sigma}
  }
  \quad
  \infer[\rn{F-App}]{
    \jtypef{\Psi}{MN}{\sigma}
  }{
    \jtypef{\Psi}{M}{\tau \to \sigma}
    &
    \jtypef{\Psi}{N}{\tau}
  }
\]

\subsection{Rules for Process Formation}
\label{sec:rules-proc-form}

\[
  \begin{adjustbox}{max width=\textwidth}
    \infer[\rn{Fwd}]{
      \jtypem{\Psi}{a:A}{\tFwd{b}{a}}{b}{A}
    }{}
    \quad
    \infer[\rn{Cut}]{
      \jtypem{\Psi}{\Delta_1,\Delta_2}{\tCut{a}{P}{Q}}{c}{C}
    }{
      \jtypem{\Psi}{\Delta_1}{P}{a}{A}
      &
      \jtypem{\Psi}{a:A,\Delta_2}{Q}{c}{C}
    }
  \end{adjustbox}
\]
\[
  \begin{adjustbox}{max width=\textwidth}
    \infer[\rn{$\Tu R$}]{
      \jtypem{\Psi}{\cdot}{\tClose a}{a}{\Tu}
    }{}
    \quad
    \infer[\rn{$\Tu L$}]{
      \jtypem{\Psi}{\Delta, a : \Tu}{\tWait{a}{P}}{c}{C}
    }{
      \jtypem{\Psi}{\Delta}{P}{c}{C}
    }
  \end{adjustbox}
\]
\[
  \begin{adjustbox}{max width=\textwidth}
    \infer[\rn{$\Tds{} R$}]{
      \jtypem{\Psi}{\Delta}{\tSendS{a}{P}}{a}{\Tds A}
    }{
      \jtypem{\Psi}{\Delta}{P}{a}{A}
    }
    \quad
    \infer[\rn{$\Tds{} L$}]{
      \jtypem{\Psi}{\Delta,a : \Tds A}{\tRecvS{a}{P}}{c}{C}
    }{
      \jtypem{\Psi}{\Delta,a : A}{P}{c}{C}
    }
  \end{adjustbox}
\]
\[
  \begin{adjustbox}{max width=\textwidth}
    \infer[\rn{$\Tus R$}]{
      \jtypem{\Psi}{\Delta}{\tRecvS{a}{P}}{a}{\Tus{A}}
    }{
      \jtypem{\Psi}{\Delta}{P}{a}{}
    }
    \quad
    \infer[\rn{$\Tus L$}]{
      \jtypem{\Psi}{\Delta,a : \Tus A}{\tSendS{a}{P}}{c}{C}
    }{
      \jtypem{\Psi}{\Delta,a : A}{P}{c}{C}
    }
  \end{adjustbox}
\]
\[
  \begin{adjustbox}{max width=\textwidth}
    \infer[\rn{$\Tplus R_k$}]{
      \jtypem{\Psi}{\Delta}{\tSendL{a}{k}{P}}{a}{{\Tplus\{l:A_l\}}_{l \in L}}
    }{
      \jtypem{\Psi}{\Delta}{P}{a}{A_k}\quad(k \in L)
    }
    \quad
    \infer[\rn{$\Tplus L$}]{
      \jtypem{\Psi}{\Delta,a:{\Tplus\{l : A_l\}}_{l \in L}}{\tCase{a}{\left\{l_l \Rightarrow P_l\right\}_{i\in I}}}{c}{C}
    }{
      \jtypem{\Psi}{\Delta,a:A_l}{P_l}{c}{C}\quad(\forall l \in L)
    }
  \end{adjustbox}
\]
\[
  \begin{adjustbox}{max width=\textwidth}
    \infer[\rn{$\Tamp R$}]{
      \jtypem{\Psi}{\Delta}{\tCase{a}{\left\{l \Rightarrow P_l\right\}_{l \in L}}}{a}{{\Tamp\{l :A_l \}}_{l \in L}}
    }{
      \jtypem{\Psi}{\Delta}{P_l}{a}{A_l}\quad(\forall l \in L)
    }
    \quad
    \infer[\rn{$\Tamp L_k$}]{
      \jtypem{\Psi}{\Delta,a:{\Tamp\{l : A_l\}}_{l \in L}}{\tSendL{a}{k}{P}}{c}{C}
    }{
      \jtypem{\Psi}{\Delta,a:A_k}{P}{c}{C}
      &
      (k \in L)
    }
  \end{adjustbox}
\]
\[
  \begin{adjustbox}{max width=\textwidth}
    \infer[\rn{$\Tot R^*$}]{
      \jtypem{\Psi}{\Delta, b : B}{\tSendC{a}{b}{P}}{a}{B \Tot A}
    }{
      \jtypem{\Psi}{\Delta}{P}{a}{A}
    }
    \quad
    \infer[\rn{$\Tot L$}]{
      \jtypem{\Psi}{\Delta, a : B \Tot A}{\tRecvC{b}{a}{P}}{c}{C}
    }{
      \jtypem{\Psi}{\Delta, a : A, b : B}{P}{c}{C}
    }
  \end{adjustbox}
\]
\[
  \begin{adjustbox}{max width=\textwidth}
    \infer[\rn{${\Tlolly}R$}]{
      \jtypem{\Psi}{\Delta}{\tRecvC{b}{a}{P}}{a}{B \Tlolly A}
    }{
      \jtypem{\Psi}{\Delta, b : B}{P}{a}{A}
    }
    \quad
    \infer[\rn{${\Tlolly}L$}]{
      \jtypem{\Psi}{\Delta, b : B, a : B \Tlolly A}{\tSendC{a}{b}{P}}{c}{C}
    }{
      \jtypem{\Psi}{\Delta,a : A}{P}{c}{C}
    }
  \end{adjustbox}
\]
\[
  \begin{adjustbox}{max width=\textwidth}
    \infer[\rn{$\Tand{}{} R$}]{
      \jtypem{\Psi}{\Delta}{\tSendV{a}{M}{P}}{a}{\Tand{\tau}{A}}
    }{
      \jtypef{\Psi}{M}{\tau}
      &
      \jtypem{\Psi}{\Delta}{P}{a}{A}
    }
    \quad
    \infer[\rn{$\Tand{}{} L$}]{
      \jtypem{\Psi}{\Delta, a:\Tand{\tau}{A}}{\tRecvV{x}{a}{P}}{c}{C}
    }{
      \jtypem{\Psi,x:\tau}{\Delta, a:A}{P}{c}{C}
    }
  \end{adjustbox}
\]
\[
  \begin{adjustbox}{max width=\textwidth}
    \infer[\rn{${\Timp{}{}} R$}]{
      \jtypem{\Psi}{\Delta}{\tRecvV{x}{a}{P}}{a}{\Timp{\tau}{A}}
    }{
      \jtypem{\Psi,x:\tau}{\Delta}{P}{a}{A}
    }
    \quad
    \infer[\rn{${\Timp{}{}} L$}]{
      \jtypem{\Psi}{\Delta,a : \Timp{\tau}{A}}{\tSendV{a}{M}{P}}{c}{C}
    }{
      \jtypef{\Psi}{M}{\tau}
      &
      \jtypem{\Psi}{\Delta, a : A}{P}{c}{C}
    }
  \end{adjustbox}
\]
\[
  \begin{adjustbox}{max width=\textwidth}
    \infer[\rn{$\rho^+R$}]{
      \jtypem{\Psi}{\Delta}{\tSendU{a}{P}}{a}{\Trec{\alpha}{A}}
    }{
      \jtypem{\Psi}{\Delta}{P}{a}{[\Trec{\alpha}{A}/\alpha]A}
      &
      \cdot \vdash \jisst[+]{\Trec{\alpha}{A}}
    }
    \,
    \infer[\rn{$\rho^+L$}]{
      \jtypem{\Psi}{\Delta, a : \Trec{\alpha}{A}}{\tRecvU{a}{P}}{c}{C}
    }{
      \jtypem{\Psi}{\Delta, a : [\Trec{\alpha}{A}/\alpha]A}{P}{c}{C}
      &
      \cdot \vdash \jisst[+]{\Trec{\alpha}{A}}
    }
  \end{adjustbox}
\]
\[
  \begin{adjustbox}{max width=\textwidth}
    \infer[\rn{$\rho^-R$}]{
      \jtypem{\Psi}{\Delta}{\tRecvU{a}{P}}{a}{\Trec{\alpha}{A}}
    }{
      \jtypem{\Psi}{\Delta}{P}{a}{[\Trec{\alpha}{A}/\alpha]A}
      &
      \cdot \vdash \jisst[-]{\Trec{\alpha}{A}}
    }
    \,
    \infer[\rn{$\rho^-L$}]{
      \jtypem{\Psi}{\Delta, a : \Trec{\alpha}{A}}{\tSendU{a}{P}}{c}{C}
    }{
      \jtypem{\Psi}{\Delta, a : [\Trec{\alpha}{A}/\alpha]A}{P}{c}{C}
      &
      \cdot \vdash \jisst[-]{\Trec{\alpha}{A}}
    }
  \end{adjustbox}
\]
\[
  \infer[\rn{E-\{\}}]{
    \jtypem{\Psi}{\overline{a_i:A_i}}{\tProc{a}{M}{\overline a_i}}{a}{A}
  }{
    \jtypef{\Psi}{M}{\Tproc{a:A}{\overline{a_i:A_i}}}
  }
\]

\subsection{Rules for Type Formation}
\label{sec:rules-type-formation}

\[
  \infer[\rn{C$\Tu$}]{\Xi\vdash\jisst[+]{\Tu}}{}
  \quad
  \infer[\rn{CVar}]{\Xi,\jisst[p]{\alpha}\vdash\jisst[p]{\alpha}}{}
  \quad
  \infer[\rn{C$\rho$}]{
    \Xi \vdash \jisst[p]{\Trec{\alpha}{A}}
  }{
    \Xi, \jisst[p]{\alpha} \vdash \jisst[p]{A}
  }
\]
\[
  \infer[\rn{C$\Tds{}$}]{
    \Xi\vdash\jisst[+]{\Tds A}
  }{
    \Xi\vdash\jisst[-]{A}
  }
  \quad
  \infer[\rn{C$\Tus{}$}]{
    \Xi\vdash\jisst[-]{\Tus A}
  }{
    \Xi\vdash\jisst[+]{A}
  }
\]
\[
  \infer[\rn{C$\Tplus$}]{
    \Xi\vdash\jisst[+]{{\Tplus\{l : A_l\}}_{l \in L}}
  }{
    \Xi\vdash\jisst[+]{A_l}\quad(\forall l \in L)
  }
  \quad
  \infer[\rn{C$\Tamp$}]{
    \Xi\vdash\jisst[-]{{\Tamp\{l :A_l \}}_{l \in L}}
  }{
    \Xi\vdash\jisst[-]{A_l}\quad(\forall l \in L)
  }
\]
\[
  \infer[\rn{C$\Tot$}]{
    \Xi\vdash\jisst[+]{A \Tot B}
  }{
    \Xi\vdash\jisst[+]{A}
    &
    \Xi\vdash\jisst[+]{B}
  }
  \quad
  \infer[\rn{C$\Tlolly$}]{
    \Xi\vdash\jisst[-]{A \Tlolly B}
  }{
    \Xi\vdash\jisst[+]{A}
    &
    \Xi\vdash\jisst[-]{B}
  }
\]
\[
  \infer[\rn{C$\Tand{}{}$}]{
    \Xi\vdash\jisst[+]{\Tand{\tau}{A}}
  }{
    \jisft{\tau}
    &
    \Xi\vdash\jisst[+]{A}
  }
  \quad
  \infer[\rn{C$\Timp{}{}$}]{
    \Xi\vdash\jisst[-]{\Timp{\tau}{A}}
  }{
    \jisft{\tau}
    &
    \Xi\vdash\jisst[-]{A}
  }
\]
\[
  \infer[\rn{T\{\}}]{
    \jisft{\Tproc{a_0:A_0}{a_1:A_1,\dotsc,a_n:A_n}}
  }{
    \cdot\vdash\jisst{A_i} \quad(0\leq i \leq n)
  }
  \quad
  \infer[\rn{T$\to$}]{
    \jisft{\tau \to \sigma}
  }{
    \jisft{\tau}
    &
    \jisft{\sigma}
  }
\]


\section{Summary of Interpretations}
\label{sec:summ-interpr}

For ease of reference, we give all of the semantic clauses (including omitted clauses).

\subsection{Clauses for Term Formation (\cref{sec:rules-term-formation})}
\label{sec:claus-term-form}

\begin{description}
\item[Rule \rn{I-\{\}}.] \Cref{eq:191919}:
  \begin{equation*}
    \sembr{\jtypef{\Psi}{\tProc{a}{P}{\overline{a_i}}}{\Tproc{a:A}{\overline{a_i:A_i}}}} = \up \circ \sembr{\jtypem{\Psi}{\overline{a_i:A_i}}{P}{a}{A}}
  \end{equation*}
\item[Rule \rn{F-Var}.] \Cref{eq:39}:
  \begin{align*}
    \sembr{\jtypef{\Psi,x:\tau}{x}{\tau}}u &= \pi^{\Psi,x}_{x}u
  \end{align*}
\item[Rule \rn{F-Fix}.] \Cref{eq:28}:
  \begin{align*}
    \sembr{\jtypef{\Psi}{\tFix{x}{M}}{\tau}}u &= \sfix{\sembr{\jtypef{\Psi,x:\tau}{M}{\tau}}}u
  \end{align*}
\item[Rule \rn{F-Fun}.] \Cref{eq:41}:
  \begin{align*}
    \sembr{\jtypef{\Psi}{\lambda x: \tau.M}{\tau \to \sigma}}u &= \strictfn\left(\lambda v \in \sembr{\tau}.\sembr{\jtypef{\Psi,x:\tau}{M}{\sigma}}\upd{u}{x \mapsto v}\right)
  \end{align*}
\item[Rule \rn{F-App}.] \Cref{eq:42}:
  \begin{align*}
    \sembr{\jtypef{\Psi}{MN}{\sigma}}u &= \sembr{\jtypef{\Psi}{M}{\tau \to \sigma}}u(\sembr{\jtypef{\Psi}{N}{\tau}}u)
  \end{align*}
\end{description}

\subsection{Clauses for Process Formation (\cref{sec:rules-proc-form})}
\label{sec:claus-proc-form}

\begin{description}
\item[Rule \rn{Fwd}.] \Cref{eq:10}:
  \[
    \sembr{\jtypem{\Psi}{a:A}{\tFwd{b}{a}}{b}{A}}u(a^+ : \alpha, b^- : \beta) = (a^- : \beta, b^+ : \alpha)
  \]
\item[Rule \rn{Cut}.] \Cref{eq:11}:
  \begin{equation*}
    \begin{aligned}
      &\sembr{\jtypem{\Psi}{\Delta_1, \Delta_2}{\tCut{a}{P}{Q}}{c}{C}}u\\
      &= \Tr{\left(\sembr{\jtypem{\Psi}{\Delta_1}{P}{a}{A}}u \times \sembr{\jtypem{\Psi}{a : A,\Delta_2}{Q}{c}{C}}u\right)}{a^- \times a^+}
    \end{aligned}
  \end{equation*}
\item[Rule \rn{$\Tu R$}.] \Cref{eq:1509}:
  \[
    \sembr{\jtypem{\Psi}{\cdot}{\tClose a}{a}{\Tu}}u(a^- : \bot) = (a^+ : \ast)
  \]
\item[Rule \rn{$\Tu L$}.] \Cref{eq:21}:
  \begin{equation*}
    \begin{aligned}
      &\sembr{\jtypem{\Psi}{\Delta,a : \Tu}{\tWait{a}{P}}{c}{C}}u = \strictfn_{a^+}\left(\lambda (\delta^+,a^+,c^-).(\delta^-,\bot,c^+)\right)\\
      &\text{where }\sembr{\jtypem{\Psi}{\Delta}{P}{c}{C}}u(\delta^+,c^-) = (\delta^-, c^+)
    \end{aligned}
  \end{equation*}
\item[Rule \rn{$\Tds{} R$}.] \Cref{eq:59}:
  \begin{equation*}
    \begin{aligned}
      &\sembr{\jtypem{\Psi}{\Delta}{\tSendS{a}{P}}{a}{\Tds A}}u\\
      &= \left(\ms{id} \times \left(a^+ : \up\right)\right) \circ \sembr{\jtypem{\Psi}{\Delta}{P}{a}{A}}u
    \end{aligned}
  \end{equation*}
\item[Rule \rn{$\Tds{} L$}.] \Cref{eq:56}:
  \begin{equation*}
    \begin{aligned}
      &\sembr{\jtypem{\Psi}{\Delta,a : \Tds A}{\tRecvS{a}{P}}{c}{C}}u\\
      &= \strictfn_{a^+}\left(\sembr{\jtypem{\Psi}{\Delta, a : A}{P}{c}{C}}u \circ \left(\ms{id} \times \left(a^+ : \down\right)\right)\right)
    \end{aligned}
  \end{equation*}
\item[Rule \rn{$\Tus{} R$}.] Omitted.
  \begin{equation}
    \label[intn]{eq:fscd:3}
    \begin{aligned}
      &\sembr{\jtypem{\Psi}{\Delta}{\tRecvS{a}{P}}{a}{\Tus{A}}}u\\
      &= \strictfn_{a^-}\left(\sembr{\jtypem{\Psi}{\Delta}{P}{a}{A}}u \circ \left(\ms{id} \times \left(a^- : \down\right)\right)\right)
    \end{aligned}
  \end{equation}
\item[Rule \rn{$\Tus{} L$}.] Omitted.
  \begin{equation}
    \label[intn]{eq:fscd:4}
    \begin{aligned}
      &\sembr{\jtypem{\Psi}{\Delta,a : \Tus A}{\tSendS{a}{P}}{c}{C}}u\\
      &= \left(\ms{id} \times \left(a^- : \up\right)\right) \circ \sembr{\jtypem{\Psi}{\Delta,a : A}{P}{c}{C}}u
    \end{aligned}
  \end{equation}
\item[Rule \rn{$\Tplus R_k$}.] \Cref{eq:52}:
  \begin{equation*}
    \begin{aligned}
      &\sembr{\jtypem{\Psi}{\Delta}{\tSendL{a}{k}{P}}{a}{\Tplus\{l:A_l\}_{l \in L}}}u\left(\delta^+, \left(a_l^-\right)_{l \in L}\right) = \left(\delta^-, \left(k, \upim{a_k^+}\right)\right)\\
      &\text{where }\sembr{\jtypem{\Psi}{\Delta}{P}{a}{A_k}}u\left(\delta^+, a_k^-\right) = \left(\delta^-, a_k^+\right)
    \end{aligned}
  \end{equation*}
\item[Rule \rn{$\Tplus L$}.] \Cref{eq:53}:
  \begin{equation*}
    \begin{aligned}
      &\sembr{\jtypem{\Psi}{\Delta,a:\Tplus\{l:A_l\}_{l \in L}}{\tCase{a}{\left\{l \Rightarrow P_l\right\}_{l \in L}}}{c}{C}}u\\
      &= \strictfn_{a^+}\left(\lambda \left(\delta^+, a^+ : \left(k, \upim{a_k^+}\right), c^-\right).\left(\delta^-, a^- : \left(k : a_k^-, l \neq k : \bot\right)_{l \in L}, c^+\right)\right)\\
      &\text{where }\sembr{\jtypem{\Psi}{\Delta,a:A_k}{P_k}{c}{C}}u(\delta^+, a_k^+, c^-) = (\delta^-, a_k^-, c^+)
    \end{aligned}
  \end{equation*}
\item[Rule \rn{$\Tamp R$}.] Omitted.
  \begin{equation}
    \label[intn]{eq:fscd:5}
    \begin{aligned}
      &\sembr{\jtypem{\Psi}{\Delta}{\tCase{a}{\left\{l \Rightarrow P_l\right\}_{l \in L}}}{a}{{\Tamp\{l :A_l \}}_{l \in L}}}u\\
      &=  \strictfn_{a^-}\left(\lambda \left(\delta^+, a^- : \left(k, \upim{a_k^-}\right)\right).\left(\delta^-, a^+ : \left(k : a_k^+, l \neq k : \bot\right)_{l \in L}\right)\right)\\
      &\text{where }\sembr{\jtypem{\Psi}{\Delta}{P_k}{a}{A_k}}u(\delta^+, a_k^-) = (\delta^-, a_k^+)
    \end{aligned}
  \end{equation}
\item[Rule \rn{$\Tamp L_k$}.] Omitted.
  \begin{equation}
    \label[intn]{eq:fscd:6}
    \begin{aligned}
      &\sembr{\jtypem{\Psi}{\Delta,a:{\Tamp\{l : A_l\}}_{l \in L}}{\tSendL{a}{k}{P}}{c}{C}}u\left(\delta^+, \left(a_l^+\right)_{l \in L}, c^-\right) = \left(\delta^-, \left(k, \upim{a_k^-}\right), c^+\right)\\
      &\text{where }\sembr{\jtypem{\Psi}{\Delta,a:A_k}{P}{c}{C}}u\left(\delta^+, a_k^+, c^-\right) = \left(\delta^-, a_k^-, c^+\right)
    \end{aligned}
  \end{equation}
\item[Rule \rn{$\Tot R^*$}.] \Cref{eq:17}:
  \begin{equation*}
    \begin{aligned}
      &\sembr{\jtypem{\Psi}{\Delta, b : B}{\tSendC{a}{b}{P}}{a}{B \Tot A}}u(\delta^+,b^+,(a^-_B,a^-_A)) = \left(\delta^-,a_B^-,\upim{\left(b^+,a_A^+\right)}\right)\\
      &\text{where }\sembr{\jtypem{\Psi}{\Delta}{P}{a}{A}}u(\delta^+,a_A^-) = (\delta^-,a_A^+)
    \end{aligned}
  \end{equation*}
\item[Rule \rn{$\Tot L$}.] \Cref{eq:26}:
  \begin{equation*}
    \begin{aligned}
      &\sembr{\jtypem{\Psi}{\Delta, a : B \Tot A}{\tRecvC{b}{a}{Q}}{c}{C}}u(\delta^+,a^+,c^-)\\
      &= \strictfn_{a^+}\left(\lambda (\delta^+,a^+ : \upim{(a_B^+, a_A^+)},c^-) . (\delta^-,(b^-,a^-),c^+) \right)\\
      &\text{where }\sembr{\jtypem{\Psi}{\Delta, a : A, b : B}{Q}{c}{C}}u(\delta^+,a_A^+,a_B^+,c^-) = (\delta^-,a^-,b^-,c^+)
    \end{aligned}
  \end{equation*}
\item[Rule \rn{${\Tlolly}R$}.] Omitted.
  \begin{equation}
    \label[intn]{eq:fscd:12}
    \begin{aligned}
      &\sembr{\jtypem{\Psi}{\Delta}{\tRecvC{b}{a}{P}}{a}{B \Tlolly A}}u(\delta^+,a^-)\\
      &= \strictfn_{a^-}\left(\lambda (\delta^+,a^- : \upim{(a_B^+, a_A^-)}) . (\delta^+,(b^-,a^+)) \right)\\
      &\text{where }\sembr{\jtypem{\Psi}{\Delta, b : B}{P}{a}{A}}u(\delta^+,a_B^+,a_A^-) = (\delta^-,b^-,a^+)
    \end{aligned}
  \end{equation}
\item[Rule \rn{${\Tlolly}L$}.] Omitted.
  \begin{equation}
    \label[intn]{eq:fscd:28}
    \begin{aligned}
      &\sembr{\jtypem{\Psi}{\Delta, b : B, a : B \Tlolly A}{\tSendC{a}{b}{P}}{c}{C}}u(\delta^+,b^+,(a^-_B,a^+_A),c^-)\\
      &= \left(\delta^-,a_B^-,\upim{\left(b^+,a_A^-\right)},c^+\right)\\
      &\text{where }\sembr{\jtypem{\Psi}{\Delta,a : A}{P}{c}{C}}u(\delta^+,a_A^+,c^-) = (\delta^-,a_A^-,c^+)
    \end{aligned}
  \end{equation}
\item[Rule \rn{$\Tand{}{} R$}.] \Cref{eq:1005}:
  \begin{equation*}
    \begin{aligned}
      &\sembr{\jtypem{\Psi}{\Delta}{\tSendV{a}{M}{P}}{a}{\Tand{\tau}{A}}}u(\delta^+,a^-)\\
      &= \begin{cases}
        \bot & \text{if }\sembr{\jtypef{\Psi}{M}{\tau}}u = \bot\\
        \left(\delta^-,\upim{\left(v,a^+\right)}\right) & \text{if }\sembr{\jtypef{\Psi}{M}{\tau}}u = v \neq \bot
      \end{cases}\\
      &\text{where }\sembr{\jtypem{\Psi}{\Delta}{P}{a}{A}}u(\delta^+,a^-) = (\delta^-, a^+)
    \end{aligned}
  \end{equation*}
\item[Rule \rn{$\Tand{}{} L$}.] \Cref{eq:1006}:
  \begin{equation*}
    \begin{aligned}
      &\sembr{\jtypem{\Psi}{\Delta, a:\Tand{\tau}{A}}{\tRecvV{x}{a}{P}}{c}{C}}u\\
      &= \strictfn_{a^+}\left( \lambda \left(\delta^+,a^+ : \upim{\left(v, \alpha^+\right)},c^-\right) . \right.\\
      &\qquad\quad\qquad\qquad \left. \sembr{\jtypem{\Psi,x:\tau}{\Delta, a:A}{P}{c}{C}}\upd{u}{x \mapsto v}(\delta^+, \alpha^+, c^-) \right)
    \end{aligned}
  \end{equation*}
\item[Rule \rn{${\Timp{}{}} R$}.] Omitted.
  \begin{equation}
    \label[intn]{eq:fscd:30}
    \begin{aligned}
      &\sembr{\jtypem{\Psi}{\Delta}{\tRecvV{x}{a}{P}}{a}{\Timp{\tau}{A}}}u\\
      &= \strictfn_{a^-}\left( \lambda \left(\delta^+,a^- : \upim{\left(v, \alpha^-\right)}\right) .  \sembr{\jtypem{\Psi,x:\tau}{\Delta}{P}{a}{A}}\upd{u}{x \mapsto v}(\delta^+, \alpha^-) \right)
    \end{aligned}
  \end{equation}
\item[Rule \rn{${\Timp{}{}} L$}.] Omitted.
  \begin{equation}
    \label[intn]{eq:fscd:31}
    \begin{aligned}
      &\sembr{\jtypem{\Psi}{\Delta,a : \Timp{\tau}{A}}{\tSendV{a}{M}{P}}{c}{C}}u(\delta^+,a^+,c^-)\\
      &= \begin{cases}
        \bot & \text{if }\sembr{\jtypef{\Psi}{M}{\tau}}u = \bot\\
        \left(\delta^-,\upim{\left(v,a^-\right)}, c^+\right) & \text{if }\sembr{\jtypef{\Psi}{M}{\tau}}u = v \neq \bot
      \end{cases}\\
      &\text{where }\sembr{\jtypem{\Psi}{\Delta, a : A}{P}{c}{C}}u(\delta^+,a^+,c^-) = (\delta^-, a^-,c^+)
    \end{aligned}
  \end{equation}
\item[Rule \rn{$\rho^+R$}.] \Cref{eq:fossacs:2}:
  \begin{equation*}
    \begin{aligned}
      &\sembr{\jtypem{\Psi}{\Delta}{\tSendU{a}{P}}{a}{\Trec{\alpha}{A}}}u\\
      &= \left(\ms{id} \times \left(a^+ : \ms{Fold}\right)\right)
      \circ \sembr{\jtypem{\Psi}{\Delta}{P}{a}{[\Trec{\alpha}{A}/\alpha]A}}u
      \circ \left(\ms{id} \times \left(a^- : \ms{Unfold}\right)\right)
    \end{aligned}
  \end{equation*}
\item[Rule \rn{$\rho^+L$}.] \Cref{eq:fossacs:3}:
  \begin{equation*}
    \begin{aligned}
      &\sembr{\jtypem{\Psi}{\Delta, a : \Trec{\alpha}{A}}{\tRecvU{a}{P}}{c}{C}}u\\
      &= \left(\ms{id} \times \left(a^- : \ms{Fold}\right)\right)
      \circ \sembr{\jtypem{\Psi}{\Delta, a : [\Trec{\alpha}{A}/\alpha]A}{P}{c}{C}}u
      \circ \left(\ms{id} \times \left(a^+ : \ms{Unfold}\right)\right)
    \end{aligned}
  \end{equation*}
\item[Rule \rn{$\rho^-R$}.] Omitted.
  \begin{equation}
    \label[intn]{eq:fscd:50}
    \begin{aligned}
      &\sembr{\jtypem{\Psi}{\Delta}{\tRecvU{a}{P}}{a}{\Trec{\alpha}{A}}}u\\
      &= \left(\ms{id} \times \left(a^+ : \ms{Fold}\right)\right)
      \circ \sembr{\jtypem{\Psi}{\Delta}{P}{a}{[\Trec{\alpha}{A}/\alpha]A}}u
      \circ \left(\ms{id} \times \left(a^- : \ms{Unfold}\right)\right)
    \end{aligned}
  \end{equation}
\item[Rule \rn{$\rho^-L$}.] Omitted.
  \begin{equation}
    \label[intn]{eq:fscd:51}
    \begin{aligned}
      &\sembr{\jtypem{\Psi}{\Delta, a : \Trec{\alpha}{A}}{\tSendU{a}{P}}{c}{C}}u\\
      &= \left(\ms{id} \times \left(a^- : \ms{Fold}\right)\right)
      \circ \sembr{\jtypem{\Psi}{\Delta, a : [\Trec{\alpha}{A}/\alpha]A}{P}{c}{C}}u
      \circ \left(\ms{id} \times \left(a^+ : \ms{Unfold}\right)\right)
    \end{aligned}
  \end{equation}
\item[Rule \rn{E-\{\}}] \Cref{eq:115}:
  \begin{equation*}
    \begin{aligned}
      &\sembr{\jtypem{\Psi}{\overline{a_i:A_i}}{\tProc{a}{M}{\overline a_i}}{a}{A}} = \down \circ \sembr{\jtypef{\Psi}{M}{\Tproc{a:A}{\overline{a_i:A_i}}}}
    \end{aligned}
  \end{equation*}
\end{description}

\subsection{Clauses for Type Formation (\cref{sec:rules-type-formation})}
\label{sec:claus-type-form}

\begin{description}
\item[Rule \rn{C$\Tu$}.] \Cref{eq:1502,eq:1501}:
  \begin{align*}
    \sembr{\Xi \vdash \jisst{\Tu}}^- &= \lambda \xi.\{\bot\}\\
    \sembr{\Xi \vdash \jisst{\Tu}}^+ &= \lambda \xi.\{\ast\}_\bot
  \end{align*}
\item[Rule \rn{CVar}.] \Cref{eq:fscd:7}:
  \begin{align*}
    \sembr{\Xi,\jisst[p]{\alpha}\vdash\jisst[p]{\alpha}}^q &= \pi^{\Xi,\alpha}_\alpha \quad (q \in \{{-},{+}\})
  \end{align*}
\item[Rule \rn{C$\rho$}.] \Cref{eq:20051}:
  \begin{align*}
    \sembr{\Xi\vdash\jisst{\Trec{\alpha}{A}}}^p &= \sfix{\left(\sembr{\Xi,\jisst{\alpha}\vdash \jisst{A}}^p\right)}\quad (p \in \{{-},{+}\})
  \end{align*}
\item[Rule \rn{C$\Tds{}$}.] \Cref{eq:1506,eq:1507}:
  \begin{align*}
    \sembr{\Xi\vdash\jisst{\Tds A}}^- &= \sembr{\Xi\vdash\jisst{A}}^-\\
    \sembr{\Xi\vdash\jisst{\Tds A}}^+ &= \sembr{\Xi\vdash\jisst{A}}^+_\bot
  \end{align*}
\item[Rule \rn{C$\Tus{}$}.] Omitted.
  \begin{align}
    \sembr{\Xi\vdash\jisst{\Tus A}}^- &= \sembr{\Xi\vdash\jisst{A}}^-_\bot\label[intn]{eq:fscd:14}\\
    \sembr{\Xi\vdash\jisst{\Tus A}}^+ &= \sembr{\Xi\vdash\jisst{A}}^+\label[intn]{eq:fscd:15}
  \end{align}
\item[Rule \rn{C$\Tplus$}.] \Cref{eq:202020,eq:222222}:
  \begin{align*}
    \sembr{\Xi\vdash\jisst{\Tplus\{l:A_l\}_{l \in L}}}^- &= \prod_{l \in L} \sembr{\Xi\vdash\jisst{A_l}}^-\\
    \sembr{\Xi\vdash\jisst{\Tplus\{l:A_l\}_{l \in L}}}^+ &= \bigoplus_{l \in L} \sembr{\Xi\vdash\jisst{A_l}}^+_\bot
  \end{align*}
\item[Rule \rn{C$\Tamp$}.] Omitted.
  \begin{align}
    \sembr{\Xi\vdash\jisst{\Tamp\{l:A_l\}_{l \in L}}}^- &= \bigoplus_{l \in L} \sembr{\Xi\vdash\jisst{A_l}}^-_\bot\label[intn]{eq:fscd:19}\\
    \sembr{\Xi\vdash\jisst{\Tamp\{l:A_l\}_{l \in L}}}^+ &= \prod_{l \in L} \sembr{\Xi\vdash\jisst{A_l}}^+\label[intn]{eq:fscd:20}
  \end{align}
\item[Rule \rn{C$\Tot$}.] \Cref{eq:13,eq:14}:
  \begin{align*}
    \sembr{\Xi\vdash\jisst{A \Tot B}}^- &= \sembr{\Xi\vdash\jisst{A}}^- \times \sembr{\Xi\vdash\jisst{B}}^-\\
    \sembr{\Xi\vdash\jisst{A \Tot B}}^+ &= \left(\sembr{\Xi\vdash\jisst{A}}^+ \times \sembr{\Xi\vdash\jisst{B}}^+\right)_\bot
  \end{align*}
\item[Rule \rn{C$\Tlolly$}.] Omitted.
  \begin{align}
    \sembr{\Xi\vdash\jisst{A \Tlolly B}}^- &= \left(\sembr{\Xi\vdash\jisst{A}}^+ \times \sembr{\Xi\vdash\jisst{B}}^-\right)_\bot\label[intn]{eq:1001}\\
    \sembr{\Xi\vdash\jisst{A \Tlolly B}}^+ &= \sembr{\Xi\vdash\jisst{A}}^- \times \sembr{\Xi\vdash\jisst{B}}^+\label[intn]{eq:1002}
  \end{align}
\item[Rule \rn{C$\Tand{}{}$}.] \Cref{eq:15104,eq:15106}:
  \begin{align*}
    \sembr{\Xi\vdash\jisst{\Tand{\tau}{A}}}^- &= \sembr{\Xi\vdash \jisst{A}}^-\\
    \sembr{\Xi\vdash\jisst{\Tand{\tau}{A}}}^+ &= \left(\sembr{\tau}\times\sembr{\Xi\vdash \jisst{A}}^+\right)_\bot
  \end{align*}
\item[Rule \rn{C$\Timp{}{}$}.] Omitted.
  \begin{align}
    \sembr{\Xi\vdash\jisst{\Timp{\tau}{A}}}^- &=  \left(\sembr{\tau}\times\sembr{\Xi\vdash \jisst{A}}^-\right)_\bot\label[intn]{eq:fscd:24}\\
    \sembr{\Xi\vdash\jisst{\Timp{\tau}{A}}}^+ &=  \sembr{\Xi\vdash \jisst{A}}^+\label[intn]{eq:fscd:25}
  \end{align}
\item[Rule \rn{T\{\}}.] \Cref{eq:72}:
  \begin{align*}
    \sembr{\jisft{\Tproc{a:A}{\overline{a_i:A_i}}}} &= \sembr{\overline{a_i:A_i} \vdash a : A}_\bot
  \end{align*}
\item[Rule \rn{T$\to$}.] \Cref{eq:73}:
  \begin{align*}
    \sembr{\jisft{\tau \to \sigma}} &= [\sembr{\jisft\tau} \sto \sembr{\jisft{\sigma}}]
  \end{align*}
\end{description}


\section{Properties of Trace and Parametrized Fixed-Point Operators}
\label{sec:prop-trace-oper}

Traces were first discovered Căzănescu and {\c{S}}tefănescu~\cite[\S~4.3]{cazanescu_stefanescu_1990:_towar_new_algeb} and then independently rediscovered by Joyal et~al.~\cite{joyal_1996:_traced_monoid_categ} in the setting of balanced monoidal categories.
In this section, we define traces in the setting of symmetric monoidal categories.
We related traces to the continuous parametrized fixed-point operator $\sfix{(\cdot)}$ given in \cref{sec:background}.
We then prove various identities that are useful for computing semantic equivalences.

\begin{definition}
A \defin{monoidal category} is a sextuple $(\mb{M}, \otimes, I, \lambda, \rho, \alpha)$ satisfying the pentagon axiom~\cite[diagram~(2.2)]{etingof_2015:_tensor_categ} and the triangle axiom~\cite[diagram (2.10)]{etingof_2015:_tensor_categ}, where
\begin{itemize}
\item $\mb{M}$ is a category
\item $\otimes : \mb{M} \times \mb{M} \to \mb{M}$ is a tensor product on $\mb{M}$
\item $I$ is the unit of the tensor
\item $\lambda : I \times A \nto A$ is a natural isomorphism witnessing that $I$ is the left unit
\item $\rho : A \otimes I \nto A$ is a natural isomorphism witnessing that $I$ is the right unit
\item $\alpha : (A \otimes B) \otimes C \nto A \otimes (B \otimes C)$ is a natural isomorphism witnessing the associativity of the tensor $\otimes$.
\end{itemize}
A monoidal category is \defin{symmetric} if it is additionally equipped by a natural isomorphism $\sigma : A \otimes B \to B \otimes A$ satisfying various other axioms given on \cite[p.~404]{barr_wells_1999:_categ_theor_comput_scien}.\label{def:fscd:1}
\end{definition}

\begin{example}
  The category $\moabc$ is a traced symmetric monoidal category.
  Its tensor product is given by the Cartesian product $\times$ and the terminal object $\top$ is the unit.
\end{example}

\begin{definition}[{\cite[Definition~2.4]{benton_hyland_2003:_traced_premon_categ}}]
  \label{def:proposal:1}
  A \defin{trace} on a symmetric monoidal category $(\mb{M}, \otimes, I, \lambda, \rho, \\\alpha, \sigma)$ is a family of functions
  \[
    \Trop^U_{A,B} : \mb{M}(A \otimes U, B \otimes U) \to \mb{M}(A, B)
  \]
  satisfying the following conditions:
  \begin{enumerate}
  \item Naturality in $A$ (left tightening): if $f : A' \otimes U \to B \otimes U$ and $g : A \to A'$, then
    \[
      \Trop^U_{A,B}\left(f \circ \left(g \otimes \ms{id}_U\right)\right) = \Trop^U_{A',B}(f) \circ g : A \to B.
    \]
  \item Naturality in $B$ (right tightening): if $f : A \otimes U \to B' \otimes U$ and $g : B' \to B$, then
    \[
      \Trop^U_{A,B}\left(\left(g \otimes \ms{id}_U\right) \circ f\right) = g \circ \Trop^U_{A,B'}(f) : A \to B.
    \]
  \item Dinaturality (sliding): if $f : A \otimes U \to B \otimes V$ and $g : V \to U$, then
    \[
      \Trop^U_{A,B}\left(\left(\ms{id}_B \otimes g\right) \circ f\right) = \Trop^V_{A,B}\left(f \circ \left(\ms{id}_A \otimes g\right)\right) : A \to B.
    \]
  \item Action (vanishing): if $f : A \to B$, then
    \[
      \Trop^I_{A,B}\left(\rho^{-1} \circ f \circ \rho\right) = f : A \to B,
    \]
    and if $f : A \otimes (U \otimes V) \to B \otimes (U \otimes V)$, then
    \[
      \Trop^{U \otimes V}_{A,B}(f) = \Trop^U_{A,B}\left(\Trop^V_{A \otimes U, B \otimes U}\left(\alpha^{-1} \circ f \circ \alpha\right)\right).
    \]
  \item Superposing: if $f : A \otimes U \to B \otimes U$, then
    \[
      \Trop^{U}_{C \otimes A, C \otimes B}\left(\alpha^{-1} \circ \left(\ms{id}_C \otimes f\right) \circ \alpha\right) = \ms{id}_C \otimes \Trop^{U}_{A, B}(f) : C \otimes A \to C \otimes B.
    \]
  \item Yanking: for all $U$,
    \[
      \Trop^U_{U,U}\left(\sigma_{U,U}\right) = \ms{id}_U : U \to U.
    \]
  \end{enumerate}
\end{definition}

\begin{example}
  The following defines a trace on the category $\moabc$~\cite[\S~5.6]{abramsky_2002:_geomet_inter_linear_combin_algeb}:
  \[
    \Trop^X_{A,B}(f) = \pi^{B\times X}_B \circ f \circ \left\langle \ms{id}_A, \sfix{\left( \pi^{B\times X}_X \circ f \right)} \right\rangle.
  \]
\end{example}

This trace is defined in terms of the parametrized fixed-point operator of \cref{sec:background}.
Formally, a parametrized fixed-point operator is:

\begin{definition}[{\cite[Definitions~2.2 and~2.4]{simpson_plotkin_2000:_compl_axiom_categ}}]
  \label{def:fscd:2}
  A \defin{parametrized fixed-point operator} on a category $\mb{M}$ is a family of morphisms $\sfix{({\cdot})} : \mb{M}(X \times A \to A) \nto \mb{M}(X \to A)$ satisfying:
  \begin{enumerate}
  \item Naturality: for any $g : X \to Y$ and $f : Y \times A \to A$,
    \[
      \sfix{f} \circ g = \sfix{(f \circ (g \times \ms{id}_A))} : X \to A.
    \]
  \item The parametrized fixed-point property: for any $f : X \times A \to A$,
    \[
      f \circ \langle \ms{id}_X, \sfix{f} \rangle = \sfix{f} : X \to A.
    \]
  \end{enumerate}
  It is a \defin{Conway operator} if it additionally satisfies:
  \begin{enumerate}[resume]
  \item Parameterized dinaturality: for any $f : X \times B \to A$ and $g : X \times A \to B$,
    \[
      f \circ \langle \ms{id}_X, \sfix{\left(g \circ \langle \pi_X, f \rangle \right)} \rangle = \sfix{(f \circ \langle \pi_1, g \rangle)} : X \to A.
    \]
  \item The diagonal property: for any $f : X \times A \times A \to A$,
    \[
      \sfix{(f \circ (\ms{id}_X \times \Delta))} = \sfix{\left(\sfix{f}\right)} : X \to A,
    \]
    where $\Delta : A \to A \times A$ is the diagonal map.
  \end{enumerate}
\end{definition}

\begin{example}
  Recall that $\moabc$ is equipped with a continuous least-fixed-point operator $\fix : [D \to D] \to D$ for each object $D$.
  The operator $\sfix{(\cdot)} : [X \times A \to A] \to [X \to A]$ from \cref{sec:background} given by $\sfix{f}(x) = \fix(\lambda a.f(x,a))$ is a Conway operator.
\end{example}

Our goal is to relate this Conway operator to the trace operator on $\moabc$.
To do so, we will use explicit definitions of $\fix$ and $\sfix{(\cdot)}$.
Recall that given a continuous $g : X \to X$, we can compute $\fix(g)$ in two ways:
\begin{itemize}
\item using the Kleene fixed-point theorem, with $\fix(g) = \dirsup_{n \in \N} g^n(\bot_X)$;
\item using a variant of the Knaster-Tarski theorem, with $\fix(g) = \bigsqcap \{ x \in X \mid g(x) \sqsubseteq x \}$.
\end{itemize}
Given a continuous $g : X \times A \to A$, we can then compute $\sfix{g}$ in two ways:
\begin{itemize}
\item using the Kleene fixed-point theorem, with $\sfix{g}(x) = \dirsup_{n \in \N} (\lambda a \in A.g(x,a))^n(\bot_A)$;
\item using a variant of the Knaster-Tarski theorem, with $\sfix{g}(x) = \bigsqcap \{ a \in A \mid g(x,a) \sqsubseteq a \}$.
\end{itemize}

Hasegawa~\cite[Theorem~7.1]{hasegawa_1999:_recur_cyclic_sharin} and Hyland independently discovered~\cite[p.~281]{benton_hyland_2003:_traced_premon_categ} that a cartesian category has a trace if and only if it has a Conway operator.
The following is a special case of the proof of the Hasegawa-Hyland theorem.

\begin{proposition}
  \label{prop:2}
  For all $f : A \times X \to B \times X$, we have $f \circ \langle \ms{id}_A, \sfix{(\pi^{B\times X}_X \circ f)} \rangle = \sfix{(f \circ \pi^{B\times X\times A}_{A \times X})}$.
  Consequently, $\Tr{(f)}{X} = \pi^{B\times X}_B \circ \sfix{(f \circ \pi^{B\times X\times A}_{A \times X})}$.
\end{proposition}

As a consequence of \cref{prop:2}, we can compute the trace $\Tr{(f)}{X}$ of a morphism $f : A \times X \to B \times X$ in two ways:
\begin{itemize}
\item using the Kleene fixed-point theorem, with
  \[
    \Tr{(f)}{X}(a) = \pi^{B \times X}_B \left(\dirsup_{n \in \N} \left(\lambda \left(b,x\right).f\left(a,x\right)\right)^n\left(\bot_B,\bot_X\right)\right);
  \]
\item using a variant of the Knaster-Tarski theorem, with
  \begin{equation}\label{eq:808}
    \Tr{(f)}{X}(a) = \pi^{B \times X}_B \left(\bigsqcap \left\{(b,x) \in B \times X \mid f(a,x) \sqsubseteq (b, x) \right\}\right).
  \end{equation}
\end{itemize}

\Cref{lemma:fscd:1} further relates the Conway operator and the trace on $\moabc$.

\begin{lemma}
  \label{lemma:fscd:1}
  Let $f : A \times X \to B \times X$ be continuous and let $a \in A$ be arbitrary.
  Take $P = \{(b, x) \in B \times X \mid f(a, x) \sqsubseteq (b, x) \}$.
  Then $\bigsqcap P = \sfix{\left(f \circ \pi^{B \times A \times X}_{A \times X}\right)}(a)$ and $\bigsqcap P \in P$.
\end{lemma}

\begin{proof}
  Let $(\hat b, \hat x) = \sfix{\left(f \circ \pi^{B \times A \times X}_{A \times X}\right)}(a)$.
  We begin by showing that $(\hat b, \hat x) \in P$.
  By the parametrized fixed-point property (\cref{def:fscd:2}),
  \begin{align*}
    f(a, \hat x) &= \left(f \circ \pi^{B \times A \times X}_{A \times X}\right)(\hat b, a, \hat x)\\
                 &= \sfix{\left(f \circ \pi^{B \times A \times X}_{A \times X}\right)}(a)\\
                 &= (\hat b, \hat x).
  \end{align*}
  So $(\hat b, \hat x) \in P$.
  Next, we show that $(\hat b, \hat x) = \bigsqcap P$.
  But this is immediate by the Knaster-Tarski formulation of $\sfix{\left(f \circ \pi^{B \times A \times X}_{A \times X}\right)}$.
  Indeed,
  \begin{align*}
    (\hat b, \hat x) &=\sfix{\left(f \circ \pi^{B \times A \times X}_{A \times X}\right)}(a)\\
                     &= \bigsqcap \left\{ (b, x) \in B \times X \mid \left(f \circ \pi^{B \times A \times X}_{A \times X}\right)(b, a, x) \sqsubseteq (b, x) \right\}\\
                     &= \bigsqcap \left\{ (b, x) \in B \times X \mid f(a, x) \sqsubseteq (b, x) \right\}\\
                     &= \bigsqcap P.\qedhere
  \end{align*}
\end{proof}






\Cref{prop:fscd:12} allows us to drop section-retraction pairs from fixed expressions.
We call a morphism $s : Y \to X$ a \defin{section} if there exists a morphism $r : X \to Y$ with $r \circ s = \ms{id}_Y$.
In this case, we call $r$ the \defin{retract}.
We frequently encounter the section-retraction pair $(\up, \down)$ when computing semantic equivalences.

\begin{proposition}
  \label{prop:fscd:12}
  Let $s : Y \to X$ be a section with retract $r : X \to Y$ in a traced category.
  For all $f : A \otimes Y \to B \otimes Y$, we have $\Trop^X_{A,B}\left(\left(\ms{id}_B \otimes s\right) \circ f \circ \left(\ms{id}_A \otimes r \right)\right) = \Trop^Y_{A,B}\left(f\right)$.
\end{proposition}

\begin{proof}
  The result follows by associativity of composition, the sliding axiom, and the fact that $r \circ s = \ms{id}_Y$:
  \begin{align*}
    &\Trop^X_{A,B}\left(\left(\ms{id}_B \otimes s\right) \circ f \circ \left(\ms{id}_A \otimes r \right)\right)\\
    &= \Trop^X_{A,B}\left(\left(\ms{id}_B \otimes s\right) \circ \left(f \circ \left(\ms{id}_A \otimes r \right)\right)\right)\\
    &= \Trop^Y_{A,B}\left(\left(f \circ \left(\ms{id}_A \otimes r \right)\right) \circ \left(\ms{id}_A \otimes s\right) \right)\\
    &= \Trop^Y_{A,B}\left(f \circ \left(\ms{id}_A \otimes \left(r \circ s\right) \right)\right)\\
    &= \Trop^Y_{A,B}\left(f \circ \left(\ms{id}_A \otimes \ms{id}_Y \right)\right)\\
    &= \Trop^Y_{A,B}\left(f\right).\qedhere
  \end{align*}
\end{proof}

\Cref{prop:fscd:13} gives an associativity property for traces.

\begin{proposition}
  \label{prop:fscd:13}
  Consider the following morphisms in a traced monoidal category where the associator and symmetry natural isomorphisms are the identity:
  \begin{align*}
    f_1 &: A_1 \otimes X_1 \to B_1 \otimes Y_1\\
    f_2 &: A_2 \otimes Y_1 \otimes X_2 \to B_2 \otimes X_1 \otimes Y_2\\
    f_2 &: A_3 \otimes Y_2 \to B_3 \otimes X_2.
  \end{align*}
  Then
  \begin{align*}
    &\Trop^{X_1 \otimes Y_1 \otimes X_2 \otimes Y_2}_{A_1 \otimes A_2 \otimes A_3, B_1 \otimes B_2 \otimes B_3}(f_1 \otimes f_2 \otimes f_3)\\
    &= \Trop^{X_1 \otimes Y_1}_{A_1 \otimes A_2 \otimes A_3, B_1 \otimes B_2 \otimes B_3}\left(f_1 \otimes \Trop^{X_2 \otimes Y_2}_{A_2 \otimes A_3, B_2 \otimes B_3}\left(f_2 \otimes f_3\right)\right)\\
    &= \Trop^{X_2 \otimes Y_2}_{A_1 \otimes A_2 \otimes A_3, B_1 \otimes B_2 \otimes B_3}\left(\Trop^{X_1 \otimes Y_1}_{A_1 \otimes A_2, B_1 \otimes B_2}\left(f_1 \otimes f_2\right) \otimes f_3\right).
  \end{align*}
\end{proposition}

\begin{proof}
  \begin{align*}
    &\Trop^{X_1 \otimes Y_1}_{A_1 \otimes A_2 \otimes A_3, B_1 \otimes B_2 \otimes B_3}\left(f_1 \otimes \Trop^{X_2 \otimes Y_2}_{A_2 \otimes A_3, B_2 \otimes B_3}\left(f_2 \otimes f_3\right)\right)\\
    &= \Trop^{X_1 \otimes Y_1}_{A_1 \otimes A_2 \otimes A_3, B_1 \otimes B_2 \otimes B_3}\left(
      \left(f_1 \otimes \ms{id}_{B_2 \otimes B_3}\right)
      \circ
      \left(
      \ms{id}_{A_1 \otimes X_1}
      \otimes
      \Trop^{X_2 \otimes Y_2}_{A_2 \otimes A_3, B_2 \otimes B_3}\left(f_2 \otimes f_3\right)
      \right)
      \right)\\
    \shortintertext{by superposing:}
    &= \Trop^{X_1 \otimes Y_1}_{A_1 \otimes A_2 \otimes A_3, B_1 \otimes B_2 \otimes B_3}\left(
      \left(f_1 \otimes \ms{id}_{B_2 \otimes B_3}\right)
      \circ
      \Trop^{X_2 \otimes Y_2}_{A_1 \otimes A_2 \otimes A_3, B_1 \otimes B_2 \otimes B_3}\left(\ms{id}_{A_1 \otimes X_1} \otimes f_2 \otimes f_3\right)
      \right)\\
    \shortintertext{by right tightening:}
    &= \Trop^{X_1 \otimes Y_1}_{A_1 \otimes A_2 \otimes A_3, B_1 \otimes B_2 \otimes B_3}\left(
      \Trop^{X_2 \otimes Y_2}_{A_1 \otimes A_2 \otimes A_3, B_1 \otimes B_2 \otimes B_3}\left(f_1 \otimes f_2 \otimes f_3\right)
      \right)\\
    \shortintertext{by action:}
    &= \Trop^{X_1 \otimes Y_1 \otimes X_2 \otimes Y_2}_{A_1 \otimes A_2 \otimes A_3, B_1 \otimes B_2 \otimes B_3}\left(
      f_1 \otimes f_2 \otimes f_3
      \right)\\
    \shortintertext{by assumption on the symmetry isomorphism:}
    &= \Trop^{X_2 \otimes Y_2 \otimes X_1 \otimes Y_1}_{A_1 \otimes A_2 \otimes A_3, B_1 \otimes B_2 \otimes B_3}\left(
      f_1 \otimes f_2 \otimes f_3
      \right)\\
    \shortintertext{by the above argument in reverse:}
    &= \Trop^{X_2 \otimes Y_2}_{A_1 \otimes A_2 \otimes A_3, B_1 \otimes B_2 \otimes B_3}\left(\Trop^{X_1 \otimes Y_1}_{A_1 \otimes A_2, B_1 \otimes B_2}\left(f_1 \otimes f_2\right) \otimes f_3\right).\qedhere
  \end{align*}
\end{proof}

\Cref{prop:4503} lets us drop the strictness operator under certain circumstances.
This is useful when computing semantic equivalences because we often need to take the trace of functions that are strict in the component we are fixing.

\begin{proposition}
  \label{prop:4503}
  Let $f : A \times X \to C \times Y$ and $g : B \times Y \to D \times X$.
  The equality
  \[
    \Tr{\left(\strictfn_X(f)\times g\right)}{X \times Y} = \Tr{\left(f\times g\right)}{X\times Y} : A\times B \to C\times D
  \]
  holds whenever $\pi^{D\times X}_X\circ g$ is not strict or $f = \strictfn_X(f)$.
\end{proposition}

\begin{proof}
  If $f$ is strict in $X$, then the result immediate.
  Assume now that $\pi^{D\times X}_X\circ g$ is not strict.
  We use the Knaster-Tarski formulation \eqref{eq:808} of the trace operator.
  Let $(a,b) \in A \times B$ be arbitrary.
  We calculate that
  \begin{align*}
    \Tr{\left(\strictfn_X(f)\times g\right)}{X \times Y}(a,b) &= \pi_{C \times D}\left(\bigsqcap L\right)\\
    \Tr{\left(f\times g\right)}{X\times Y}(a,b) &= \pi_{C \times D}\left(\bigsqcap R\right)
  \end{align*}
  where
  \begin{align*}
    l &= \strictfn_X(f) \times g,\\
    L &= \{ (c,y,d,x) \in C \times Y \times D \times X \mid l(a,x,b,y) \sqsubseteq (c,y,d,x) \},\\
    r &= f \times g,\\
    R &= \{ (c,y,d,x) \in C \times Y \times D \times X \mid r(a,x,b,y) \sqsubseteq (c,y,d,x) \}.
  \end{align*}
  To show
  \[
    \Tr{\left(\strictfn_X(f)\times g\right)}{X \times Y}(a,b) = \Tr{\left(f\times g\right)}{X\times Y}(a,b),
  \]
  it is sufficient to show that $\bigsqcap L = \bigsqcap R$.


  We begin by showing that $\bigsqcap L \sqsubseteq \bigsqcap R$.
  To do so, we show that $R \subseteq L$.
  Let $(c_r,y_r,d_r,x_r) \in R$ be arbitrary.
  By definition of $R$, $r(a,x_r,b,y_r) \sqsubseteq (c_r,y_r,d_r,x_r)$.
  We have $l \sqsubseteq r$ because $\strictfn_X(f) \sqsubseteq f$ and products are monotone.
  Then
  \[
    l(a,x_r,b,y_r) \sqsubseteq r(a,x_r,b,y_r) \sqsubseteq (c_r,y_r,d_r,x_r),
  \]
  so $(c_r,y_r,d_r,x_r) \in L$.
  We conclude $R \subseteq L$ and $\bigsqcap L \sqsubseteq \bigsqcap R$.

  We now show that $\bigsqcap R \sqsubseteq \bigsqcap L$.
  To do so, we show that $\bigsqcap L \in R$.
  Let $(c_l,y_l,d_l,x_l) = \bigsqcap L$.
  We must show that
  \[
    r(a,x_l,d,y_l) \sqsubseteq (c_l,y_l,d_l,x_l).
  \]
  Because $\pi^{D \times X}_X \circ g$ is not strict, we know that $x_l \neq \bot_X$.
  Indeed, by monotonicity we have
  \[
    \bot_X \sqsubsetneq \left(\pi^{D \times X}_X \circ g\right)(\bot_B, \bot_Y) \sqsubseteq \left(\pi^{D \times X}_X \circ g\right)(b, y_l) = x_l
  \]
  Observe that whenever $x \neq \bot_X$, we have for all $y$ that $r(a,x,b,y) = l(a,x,b,y)$.
  So $r(a,x_l,d,x_l) = l(a,x_l,d,x_l)$.
  By \cref{lemma:fscd:1}, $\bigsqcap L \in L$, so $l(a,x_l,b,y_l) \sqsubseteq (c_l,y_l,d_l,x_l)$.
  These facts imply
  \[
    r(a,x_l,d,x_l)  = l(a,x_l,d,x_l) \sqsubseteq  (c_l,y_l,d_l,x_l),
  \]
  \ie, that $\bigsqcap L \in R$.
  We conclude that $\bigsqcap R \sqsubseteq \bigsqcap L$.

  Because $(a,b)$ was arbitrary in the domain, we conclude that the functions are equal.\qedhere
\end{proof}

\Cref{prop:fscd:16} tells us that the trace in $\moabc$ interacts well with currying.

\begin{proposition}
  \label{prop:fscd:16}
  Let $f : A \times B \times X \to C \times X$.
  Let $\Lambda f : A \to [B \times X \to C \times X]$ be given by Currying.
  Then $\Trop^{X}_{A \times B,C}(f) = \lambda (a, b) \in A \times B.\Trop^X_{B,C}\left(\Lambda f(a)\right)(b)$.
\end{proposition}

\begin{proof}
  We use the explicit definition of the trace and Conway operators.
  Let $(a, b) \in A \times B$ be arbitrary and write $f_a$ for $\Lambda f (a)$:
  \begin{align*}
    &\Trop^{X}_{A \times B,C}(f)(a,b)\\
    &= \left(\pi_C \circ f \circ \langle \ms{id}_{A \times B}, \sfix{(\pi_X \circ f)} \rangle\right)(a,b)\\
    &= \pi_C\left(f\left(a, b, \sfix{(\pi_X \circ f)}(a,b) \right)\right)\\
    &= \pi_C\left(f\left(a, b, \fix{\left(\lambda x \in X . \pi_X(f(a,b,x))\right)} \right)\right)\\
    &= \pi_C\left(f_a\left(b, \fix{\left(\lambda x \in X . \pi_X(f_a(b,x))\right)} \right)\right)\\
    &= \pi_C\left(f_a\left(b, \sfix{(\pi_X \circ f_a)}(b) \right)\right)\\
    &= \left(\pi_C \circ f_a \circ \langle \ms{id}_{B}, \sfix{(\pi_X \circ f_a)} \rangle\right)(b)\\
    &= \Trop^{X}_{B,C}(f_a)(b)\\
    &= \left(\lambda (a, b) \in A \times B.\Trop^X_{B,C}\left(\Lambda f(a)\right)(b)\right)(a,b).
  \end{align*}
  Because $(a,b)$ was arbitrary, we conclude that the two functions are equal.
\end{proof}

\Cref{lemma:fscd:2} will be used by the proof of \cref{prop:fscd:15}.

\begin{lemma}
  \label{lemma:fscd:2}
  Let $X$ and $Y$ be partial orders with bottom elements.
  Let $s : X \to Y$ and $r : Y \to X$ monotone morphisms.
  If $(s, r)$ is a section-retraction pair, then $r$ is strict.
\end{lemma}

\begin{proof}
  We must show that $r(\bot_Y) = \bot_X$.
  Observe that $\bot_Y \sqsubseteq s(\bot_X)$.
  By monotonicity of $r$, $r(\bot_Y) \sqsubseteq r(s(\bot_X)) = \ms{id}_X(\bot_X) = \bot_X$.
\end{proof}

\Cref{prop:fscd:15} is useful when computing fixed points of quoted processes.

\begin{proposition}
  \label{prop:fscd:15}
  Let $X$ and $Y$ be dcpos.
  Let $s : X \to Y$ and $r : Y \to X$ be a continuous section-retraction pair and let $f : X \to X$ be continuous.
  Then $\fix(s \circ f \circ r) = s(\fix(f))$.
\end{proposition}

\begin{proof}
  Let $l = s \circ f \circ r$.
  We claim for all $n \geq 1$ that $l^n = s \circ f^n \circ r$.
  We proceed by induction on $n$.
  The result is immediate for $n = 1$.
  Assuming the result for some $n$, we get:
  \[
    l^{n+1} = l \circ l^n = (s \circ f \circ r) \circ (s \circ f^n \circ r) = s \circ f \circ f^n \circ r = s \circ f^{n+1} \circ r.
  \]
  We then compute:
  \begin{align*}
    \fix(s \circ f \circ r) &= \dirsup_{n = 0}^\infty l^n(\bot_Y)\\
                            &= \dirsup_{n = 1}^\infty l^n(\bot_Y)\\
    \shortintertext{by the claim:}
                            &= \dirsup_{n = 1}^\infty (s \circ f^n \circ r)(\bot_Y)\\
    \shortintertext{by \cref{lemma:fscd:2}:}
                            &= \dirsup_{n = 1}^\infty (s \circ f^n)(\bot_X)\\
    \shortintertext{by continuity:}
                            &= s\left(\dirsup_{n = 1}^\infty f^n(\bot_X)\right)\\
                            &= s(\fix(f)).\qedhere
  \end{align*}
\end{proof}

\Cref{lemma:fscd:4} is useful when reasoning about stream transducers.

\begin{lemma}
  \label{lemma:fscd:4}
  Consider continuous functions
  \begin{align*}
    &f : (c_1^+ : A) \times (c_2^- : \top) \to (c_1^- : \top) \times (c_2^+ : A),\\
    &g : (c_2^+ : A) \times (c_3^- : \top) \to (c_2^- : \top) \times (c_3^+ : A).
  \end{align*}
  Let
  \begin{align*}
    \rho(c_1^+ : a, c_3^- : \bot) &= (c_1^+ : a, c_2^- : \bot) : (c_1^+ : A) \times (c_3^- : \top) \to (c_1^+ : A) \times (c_2^- : \top)\\
    \sigma(c_1^- : \bot, c_2^+ : a) &= (c_2^+ : a, c_3^- : \bot) : (c_1^- : \top) \times (c_2^+ : A) \to (c_2^+ : A) \times (c_3^- : \top)\\
    \tau(c_2^- : \bot, c_3^+ : a) &= (c_1^- : \bot, c_3^+ : a) : (c_2^- : \top) \times (c_3^+ : A) \to (c_1^- : \top) \times (c_3^+ : A)
  \end{align*}
  be relabelling isomorphisms.
  Then $\Trop^{c_2^- \times c_2^+}(f \times g) = \tau \circ g \circ \sigma \circ f \circ \rho$.
\end{lemma}

\begin{proof}
  Let $(a,\bot) \in (c_1^+ : A) \times (c_3^- : \top)$ be arbitrary.
  Let
  \begin{multline*}
    F = \left\{(\bot,a_2,\bot,a_3) \in (c_1^- : \top) \times (c_2^+ : A) \times (c_2^- : \top) \times (c_3^+ : A) \mid {} \right.\\
    \left. {} \mid (f \times g)(a,\bot,a_2,\bot) \sqsubseteq (\bot,a_2,\bot,a_3) \right\}.
  \end{multline*}
  By the Knaster-Tarski formulation of the trace (\eqref{eq:808}):
  \[
    \Trop^{c_2^- \times c_2^+}(f \times g)(a,\bot) = \pi_{c_1^-,c_3^+}\left(\bigsqcap F\right).
  \]
  To prove the proposition, it is sufficient to show that:
  \[
    \pi_{c_1^-,c_3^+}\left(\bigsqcap F\right) = (\tau \circ g \circ \sigma \circ f \circ \rho)(a,\bot).
  \]

  We begin by showing $(\tau \circ g \circ \sigma \circ f \circ \rho)(a,\bot) \sqsubseteq \pi_{c_1^-,c_3^+}\left(\bigsqcap F\right)$.
  Let $(\bot,a_2,\bot,a_3) = \bigsqcap F$.
  By \cref{lemma:fscd:1}, $\bigsqcap F \in F$, so $(f \times g)(a,\bot,a_2,\bot) \sqsubseteq (\bot,a_2,\bot,a_3)$, \ie,
  \begin{align}
    f(c_1^+ : a, c_2^- : \bot) &\sqsubseteq (c_1^- : \bot, c_2^+ : a_2)\label{eq:fscd:54}\\
    g(c_2^+ : a_2, c_3^- : \bot) &\sqsubseteq (c_2^- : \bot,c_3^+ : a_3).\label{eq:fscd:55}
  \end{align}
  We generate a sequence of inequalities using monotonicity.
  First, by definition of $\rho$:
  \begin{align*}
    \rho(c_1^+ : a, c_3^- : \bot) &= (c_1^+ : a, c_2^- : \bot)\\
    \shortintertext{applying $f$ to both sides:}
    (f \circ \rho)(c_1^+ : a, c_3^- : \bot) &= f(c_1^+ : a, c_2^- : \bot)\\
    \shortintertext{by transitivity with \eqref{eq:fscd:54}:}
    (f \circ \rho)(c_1^+ : a, c_3^- : \bot) &\sqsubseteq (c_1^- : \bot, c_2^+ : a_2)\\
    \shortintertext{applying $\sigma$ to both sides and computing on the right:}
    (\sigma \circ f \circ \rho)(c_1^+ : a, c_3^- : \bot) &\sqsubseteq (c_3^- : \bot, c_2^+ : a_2)\\
    \shortintertext{applying $f$ to both sides:}
    (g \circ \sigma \circ f \circ \rho)(c_1^+ : a, c_3^- : \bot) &\sqsubseteq g(c_3^- : \bot, c_2^+ : a_2)\\
    \shortintertext{by transitivity with \eqref{eq:fscd:55}:}
    (g \circ \sigma \circ f \circ \rho)(c_1^+ : a, c_3^- : \bot) &\sqsubseteq (c_2^- : \bot,c_3^+ : a_3)\\
    \shortintertext{applying $\tau$ to both sides and computing on the right:}
    (\tau \circ g \circ \sigma \circ f \circ \rho)(c_1^+ : a, c_3^- : \bot) &\sqsubseteq (c_1^- : \bot,c_3^+ : a_3)\\
    \shortintertext{but $(c_1^- : \bot,c_3^+ : a_3) = \pi_{c_1^-,c_3^+}\left(\bigsqcap F\right)$, so:}
    (\tau \circ g \circ \sigma \circ f \circ \rho)(c_1^+ : a, c_3^- : \bot) &\sqsubseteq \pi_{c_1^-,c_3^+}\left(\bigsqcap F\right).
  \end{align*}

  Now we show that $\pi_{c_1^-,c_3^+}\left(\bigsqcap F\right) \sqsubseteq (\tau \circ g \circ \sigma \circ f \circ \rho)(a,\bot)$.
  Let $a_2, a_3 \in A$ be such that
  \begin{align*}
    f(c_1^+ : a, c_2^- : \bot) &= (c_1^- : \bot, c_2^+ : a_2),\\
    g(c_2^+ : a_2, c_3^- : \bot) &= (c_2^- : \bot, c_3^+ : a_3).
  \end{align*}
  Then $(\tau \circ g \circ \sigma \circ f \circ \rho)(c_1^+ : a, c_3^- : \bot) = (c_1^- : \bot, c_3^+ : a_3)$.
  But the above two equations imply that
  \[
    (f \times g)(c_1^+ : a,  c_2^- : \bot, c_2^+ : a_2, c_3^- : \bot) = (c_1^- : \bot, c_2^+ : a_2, c_2^- : \bot,c_3^+ : a_3),
  \]
  so $(c_1^- : \bot, c_2^+ : a_2, c_2^- : \bot,c_3^+ : a_3) \in F$.
  It follows that
  \[
    \pi_{c_1^-,c_3^+}\left(\bigsqcap F\right) \sqsubseteq (c_1^- : \bot, c_3^+ : a_3) = (\tau \circ g \circ \sigma \circ f \circ \rho)(c_1^+ : a, c_3^- : \bot). \qedhere
  \]
\end{proof}


\section{Semantic Results for Terms and Processes}
\label{sec:semant-results-terms}

We show various supporting results for terms and processes.
To simplify the proofs of these results, we adopt the following slogan:
\begin{quote}
  The sound categorical interpretation of notion the of term formation amounts to requiring that certain naturality conditions hold in the categorical model.~\cite[p.~165]{crole_1993:_categ_types}:
\end{quote}
To make this explicit, assume that judgments $\jtypef{\Psi}{M}{\tau}$ are interpreted as morphisms $\sembr{\Psi} \to \sembr{\tau}$ in some category $\mb{C}$.
Consider a term-forming rule
\[
  \adjustbox{width=\linewidth,keepaspectratio}\bgroup
  \infer{
    \jtypef{\Psi}{F_{\mathcal{J}_1,\dotsc,\mathcal{J}_l}(P_1,\dotsc,P_n,M_1,\dotsc,M_m)}{\tau}
  }{
    \jtypem{\Psi}{\Delta_1}{P_1}{c_1}{C_1}
    &
    \cdots
    &
    \jtypem{\Psi}{\Delta_n}{P_n}{c_n}{C_n}
    &
    \jtypef{\Psi}{M_1}{\tau_1}
    &
    \cdots
    &
    \jtypef{\Psi}{M_m}{\tau_m}
    &
    \mathcal{J}_1
    &
    \cdots
    &
    \mathcal{J}_l
  }
  \egroup
\]
Assume its interpretation is given by
\begin{equation}
  \label[intn]{eq:fscd:37}
  \begin{aligned}
    &\sembr{\jtypem{\Psi}{\Delta}{F_{\mathcal{J}_1,\dotsc,\mathcal{J}_l}(P_1,\dotsc,P_n,M_1,\dotsc,M_m)}{c}{C}}\\
    &= \sembr{F_{\mathcal{J}_1,\dotsc,\mathcal{J}_l}}_{\sembr{\Psi}}\left(
      \sembr{\jtypem{\Psi}{\Delta_1}{P_1}{c_1}{C_1}},
      \dotsc,
      \sembr{\jtypem{\Psi}{\Delta_n}{P_n}{c_n}{C_n}},\right.\\
    &\qquad\qquad\left.
      \sembr{\jtypef{\Psi}{M_1}{\tau_1}},
      \dotsc,
      \sembr{\jtypef{\Psi}{M_m}{\tau_m}}
    \right),
  \end{aligned}
\end{equation}
where $\sembr{F_{\mathcal{J}_1,\dotsc,\mathcal{J}_l}}$ is a family of morphisms
\begin{equation}
  \label{eq:fscd:2}
  \sembr{F_{\mathcal{J}_1,\dotsc,\mathcal{J}_l}}_{\sembr{\Psi}} :
  \left(\prod_{i = 1}^n \mb{C}\left( \sembr{\Psi}, \sembr{\Delta_i \vdash c_i : C_i} \right)\right)
  \times
  \left(\prod_{i = 1}^m \mb{C}\left( \sembr{\Psi}, \sembr{\tau_i} \right) \right)
  \to
  \mb{C}\left( \sembr{\Psi}, \sembr{\tau} \right).
\end{equation}
We say that \cref{eq:fscd:37} is \defin{natural in its environment} if the family \eqref{eq:fscd:2} is natural in $\sembr{\Psi}$.
In this case, we call $\sembr{F_{\mathcal{J}_1,\dotsc,\mathcal{J}_l}}$ a \defin{natural interpretation} of the rule.
The definition for process-forming rules is analogous.

\begin{proposition}
  \label{prop:fscd:2}
  If $\jtypef{\Psi}{M}{\tau}$, then the interpretation $\sembr{\jtypef{\Psi}{M}{\tau}}$ is natural in its environment.
  If $\jtypem{\Psi}{\Delta}{P}{a}{A}$, then the interpretation $\sembr{\jtypem{\Psi}{\Delta}{P}{a}{A}}$ is natural in its environment.
\end{proposition}

\begin{proof}
  By case analysis on the last rule in the derivation of $\jtypef{\Psi}{M}{\tau}$ and $\jtypem{\Psi}{\Delta}{P}{a}{A}$.
  We explain notation and styles of arguments in the first case in which they are used.
  Because many of the cases are similar, we adopt a concise style in subsequent cases.

  \begin{description}
  \item[Case \rn{I-\{\}}.]
    \[
      \infer[\rn{I-\{\}}]{
        \jtypef{\Psi}{\tProc{a}{P}{\overline{a_i}}}{\Tproc{a:A}{\overline{a_i:A_i}}}
      }{
        \jtypem{\Psi}{\overline{a_i:A_i}}{P}{a}{A}
      }
    \]
    Recall \cref{eq:191919}:
    \[
      \sembr{\jtypef{\Psi}{\tProc{a}{P}{\overline{a_i}}}{\Tproc{a:A}{\overline{a_i:A_i}}}} = \up \circ \sembr{\jtypem{\Psi}{\overline{a_i:A_i}}{P}{a}{A}}.
    \]
    The corresponding natural interpretation is:
    \[
      \moabc\left( {-} ,  \up \right) : \moabc\left( {-} ,  \sembr{\tProc{a}{P}{\overline{a_i}} \vdash a : A} \right) \nto \moabc\left( {-} ,  \sembr{\Tproc{a:A}{\overline{a_i:A_i}}} \right)
    \]
    where the component at $\sembr{\Psi}$ is given by
    \[
      \moabc\left( \sembr{\Psi} ,  \up \right)(p) = \up \circ p.
    \]

  \item[Case \rn{F-Var}.]
    \[
      \infer[\rn{F-Var}]{\jtypef{\Psi, x: \tau}{x}{\tau}}{\mathstrut}
    \]
    Recall \cref{eq:39}:
    \[
      \sembr{\jtypef{\Psi,x:\tau}{x}{\tau}}u = \pi^{\Psi,x}_{x}u.
    \]
    The corresponding natural interpretation is:
    \[
      \left(\lambda \bot . \lambda u \in \sembr{\Psi,x:\tau}.\pi^{\Psi,x}_{x}u\right)_{\sembr{\Psi}} : \{ \bot \} \nto \moabc\left( \sembr{\Psi,x:\tau} ,  \sembr{\tau} \right)
    \]

  \item[Case \rn{F-Fix}.]
    \[
      \infer[\rn{F-Fix}]{
        \jtypef{\Psi}{\tFix{x}{M}}{\tau}
      }{
        \jtypef{\Psi,x:\tau}{M}{\tau}
      }
    \]
    Recall \cref{eq:28}:
    \[
      \sembr{\jtypef{\Psi}{\tFix{x}{M}}{\tau}}u = \sfix{\sembr{\jtypef{\Psi,x:\tau}{M}{\tau}}}u.
    \]
    This is equivalent to
    \[
      \sembr{\jtypef{\Psi}{\tFix{x}{M}}{\tau}}u = \fix\left(\lambda v \in \sembr{\tau}.\sembr{\jtypef{\Psi,x:\tau}{M}{\tau}}\upd{u}{x \mapsto v}\right),
    \]
    which is itself equivalent to
    \[
      \sembr{\jtypef{\Psi}{\tFix{x}{M}}{\tau}}u = \fix\left(\Lambda\left(\sembr{\jtypef{\Psi,x:\tau}{M}{\tau}}\right)\right),
    \]
    where
    \[
      \Lambda : \moabc\left( {-} \times \sembr{x : \tau} ,  \sembr{\tau} \right) \nto \moabc\left( {-}, \moabc\left( \sembr{\tau} ,  \sembr{\tau} \right) \right)
    \]
    is the currying natural isomorphism given by the adjunction for the exponential.
    The corresponding natural interpretation for \rn{F-Fix} is then:
    \[
      \moabc\left( {-} ,  \fix \right)  \circ \Lambda : \moabc\left( {-} \times \sembr{x : \tau} ,  \sembr{\tau} \right) \nto \moabc\left( {-} ,  \sembr{\tau} \right).
    \]

  \item[Case \rn{F-Fun}.]
    \[
      \infer[\rn{F-Fun}]{
        \jtypef{\Psi}{\lambda x : \tau.M}{\tau \to \sigma}
      }{
        \jtypef{\Psi, x:\tau}{M}{\sigma}
      }
    \]
    Recall \cref{eq:41}:
    \[
      \sembr{\jtypef{\Psi}{\lambda x: \tau.M}{\tau \to \sigma}}u = \strictfn\left(\lambda v \in \sembr{\tau}.\sembr{\jtypef{\Psi,x:\tau}{M}{\sigma}}\upd{u}{x \mapsto v}\right).
    \]
    This is equivalent to
    \[
      \sembr{\jtypef{\Psi}{\lambda x: \tau.M}{\tau \to \sigma}}u = \strictfn(\Lambda(\sembr{\jtypef{\Psi,x:\tau}{M}{\sigma}})).
    \]
    The corresponding natural interpretation is:
    \[
      \moabc\left({-} ,  \strictfn\right) \circ \Lambda : \moabc\left( {-} \times \sembr{x : \tau} ,  \sembr{\sigma} \right) \nto \moabc\left( {-} ,  \sembr{\sigma} \right).
    \]

  \item[Case \rn{F-App}.]
    \[
      \infer[\rn{F-App}]{
        \jtypef{\Psi}{MN}{\sigma}
      }{
        \jtypef{\Psi}{M}{\tau \to \sigma}
        &
        \jtypef{\Psi}{N}{\tau}
      }
    \]
    Recall \cref{eq:42}:
    \[
      \sembr{\jtypef{\Psi}{MN}{\sigma}}u = \sembr{\jtypef{\Psi}{M}{\tau \to \sigma}}u(\sembr{\jtypef{\Psi}{N}{\tau}}u).
    \]
    We must show that this is the image $(\eta_{\sembr{\Psi}}(\sembr{\jtypef{\Psi}{M}{\tau \to \sigma}}, \sembr{\jtypef{\Psi}{N}{\tau}}))u$ of some natural transformation
    \[
      \eta : \moabc\left( {-}, \sembr{\tau \to \sigma} \right) \times \moabc\left({-} ,  \sembr{\tau} \right) \nto \moabc\left({-} ,  \sembr{\sigma} \right).
    \]
    There exists a canonical natural isomorphism
    \[
      \alpha : \moabc\left( {-}, \sembr{\tau \to \sigma} \right) \times \moabc\left({-} ,  \sembr{\tau} \right) \nto \moabc\left( {-}, \sembr{\tau \to \sigma} \times \sembr{\tau} \right)
    \]
    whose $D$-component is $\alpha_D(m,n)(u) = (mu, nu)$.
    The counit $\ms{ev}$ of the exponential adjunction is a natural transformation whose $\sembr{\tau},\sembr{\sigma}$ component is
    \[
      \ms{ev}_{\sembr{\tau},\sembr{\sigma}} : \sembr{\tau \to \sigma} \times \sembr{\tau} \to \sembr{\sigma}.
    \]
    sending $(f,v)$ to $f(v)$.
    This morphism induces a natural transformation
    \[
      \moabc\left( {-} ,  \ms{ev}_{\sembr{\tau},\sembr{\sigma}} \right) : \moabc\left( {-}, \sembr{\tau \to \sigma} \times \sembr{\tau} \right) \nto \moabc\left( {-} ,  \sembr{\sigma} \right).
    \]
    The corresponding natural interpretation for \rn{F-App} is then $\eta = \moabc\left( {-} ,  \ms{ev}_{\sembr{\tau},\sembr{\sigma}} \right) \circ \alpha$.

  \item[Case \rn{Fwd}.]
    \[
      \infer[\rn{Fwd}]{
        \jtypem{\Psi}{a:A}{\tFwd{b}{a}}{b}{A}
      }{}
    \]
    Recall \cref{eq:10}:
    \[
      \sembr{\jtypem{\Psi}{a:A}{\tFwd{b}{a}}{b}{A}}u(a^+,b^-) = (b^-,a^+)
    \]
    Let $f = \lambda (a^+,b^-) . (b^-,a^+)$.
    The corresponding natural interpretation is
    \[
      \left(\lambda \bot . \lambda u \in \sembr{\Psi} . f \right)_{\sembr{\Psi}} : \{\bot\} \nto \moabc\left( \sembr{\Psi} ,  \sembr{a : A \vdash b : A} \right).
    \]
    This family is natural is because it is a constant family.

  \item[Case \rn{Cut}.]
    \[
      \infer[\rn{Cut}]{
        \jtypem{\Psi}{\Delta_1,\Delta_2}{\tCut{a}{P}{Q}}{c}{C}
      }{
        \jtypem{\Psi}{\Delta_1}{P}{a}{A}
        &
        \jtypem{\Psi}{a:A,\Delta_2}{Q}{c}{C}
      }
    \]
    Recall \cref{eq:11}:
    \begin{equation*}
      \begin{aligned}
        &\sembr{\jtypem{\Psi}{\Delta_1, \Delta_2}{\tCut{a}{P}{Q}}{c}{C}}u\\
        &= \Tr{\left(\sembr{\jtypem{\Psi}{\Delta_1}{P}{a}{A}}u \times \sembr{\jtypem{\Psi}{a : A,\Delta_2}{Q}{c}{C}}u\right)}{a^- \times a^+}.
      \end{aligned}
    \end{equation*}
    We must show that there is a corresponding natural interpretation
    \begin{multline*}
      \eta : \moabc\left( {-} ,  \sembr{\Delta_1 \vdash a : A} \right) \times \moabc\left( {-} ,  \sembr{a : A, \Delta_2 \vdash c : C} \right) \nto {}\\
      {} \nto \moabc\left( {-} ,  \sembr{\Delta_1,\Delta_2 \vdash c : C} \right).
    \end{multline*}
    There exists a natural isomorphism
    \begin{multline*}
      \alpha : \moabc\left( {-} ,  \sembr{\Delta_1 \vdash a : A} \right) \times \moabc\left( {-} ,  \sembr{a : A, \Delta_2 \vdash c : C} \right) \nto {}\\
      {} \nto \moabc\left( {-} ,  \sembr{\Delta_1 \vdash a : A} \times \sembr{a : A, \Delta_2 \vdash c : C} \right)
    \end{multline*}
    whose $D$-component is $\alpha_D(p,q)(u) = (pu, qu)$.
    There exists an injection
    \begin{multline*}
      \psi : \sembr{\Delta_1 \vdash a : A} \times \sembr{a: A, \Delta_2 \vdash c : C} \to {}\\
      {} \to \moabc\left( \sembr{\Delta_1, \Delta_2, a : A}^+ \times \sembr{a : A, c : C}^- ,  \sembr{\Delta_1, \Delta_2, a : A}^- \times \sembr{a : A, c : C}^+ \right)
    \end{multline*}
    where $\psi(p,q) = p \times q$.
    The trace operator defines a morphism
    \begin{multline*}
      \Tr{}{a^- \times a^+} :  \moabc\left( \sembr{\Delta_1, \Delta_2, a : A}^+ \times \sembr{a : A, c : C}^-, {} \right.\\
      \left.  \sembr{\Delta_1, \Delta_2, a : A}^- \times \sembr{a : A, c : C}^+ \right) \to \sembr{\Delta_1, \Delta_2 \vdash c : C}.
    \end{multline*}
    By the Yoneda lemma, the morphism $\Tr{}{a^- \times a^+} \circ \psi$ uniquely determines a natural transformation
    \begin{multline*}
      \moabc\left( {-} ,  \Tr{}{a^- \times a^+} \circ \psi \right) : \moabc\left( {-} ,  \sembr{\Delta_1 \vdash a : A} \times \sembr{a : A, \Delta_2 \vdash c : C} \right) \nto {}\\
      {} \nto \moabc\left( {-} ,  \sembr{\Delta_1, \Delta_2 \vdash c : C} \right)
    \end{multline*}
    given by post-composition.
    The corresponding natural interpretation is then:
    \[
      \eta = \moabc\left( {-} ,  \Tr{}{b^- \times b^+} \circ \psi \right) \circ \alpha.
    \]
    Indeed, for all $u \in \sembr{\Psi}$,
    \begin{align*}
      &\eta_{\sembr{\Psi}}\left(\sembr{\jtypem{\Psi}{\Delta_1}{P}{b}{B}}, \sembr{\jtypem{\Psi}{b:B,\Delta_2}{Q}{c}{C}}\right)u\\
      &= \left(\Tr{}{b^- \times b^+} \circ \psi\right)\left(\sembr{\jtypem{\Psi}{\Delta_1}{P}{b}{B}}u, \sembr{\jtypem{\Psi}{b:B,\Delta_2}{Q}{c}{C}}u\right)\\
      &= \Tr{}{b^- \times b^+}\left(\sembr{\jtypem{\Psi}{\Delta_1}{P}{b}{B}}u \times \sembr{\jtypem{\Psi}{b:B,\Delta_2}{Q}{c}{C}}u\right).
    \end{align*}

  \item[Case \rn{$\Tu R$}.]
    \[
      \infer[\rn{$\Tu R$}]{
        \jtypem{\Psi}{\cdot}{\tClose a}{a}{\Tu}
      }{}
    \]
    Recall \cref{eq:1509}:
    \[
      \sembr{\jtypem{\Psi}{\cdot}{\tClose a}{a}{\Tu}}u(a^- : \bot) = (a^+ : \ast)
    \]
    Let $f = \lambda(a^- : \bot).(a^+ : \ast)$.
    The corresponding natural interpretation is
    \[
      \left(\lambda \bot . \lambda u \in \sembr{\Psi} . f \right)_{\sembr{\Psi}} : \{\bot\} \nto \moabc\left( \sembr{\Psi} ,  \sembr{\cdot \vdash a : \Tu} \right).
    \]
    This family is natural is because it is a constant family.

  \item[Case \rn{$\Tu L$}.]
    \[
      \infer[\rn{$\Tu L$}]{
        \jtypem{\Psi}{\Delta, a : \Tu}{\tWait{a}{P}}{c}{C}
      }{
        \jtypem{\Psi}{\Delta}{P}{c}{C}
      }
    \]
    Recall \cref{eq:21}:
    \begin{equation*}
      \begin{aligned}
        &\sembr{\jtypem{\Psi}{\Delta,a : \Tu}{\tWait{a}{P}}{c}{C}}u = \strictfn_{a^+}\left(\lambda (\delta^+,a^+,c^-).(\delta^-,\bot,c^+)\right))\\
        &\text{where }(\delta^-, c^+) = \sembr{\jtypem{\Psi}{\Delta}{P}{c}{C}}u(\delta^+,c^-).
      \end{aligned}
    \end{equation*}
    We must show that there exists a natural interpretation
    \[
      \eta : \moabc\left( {-} ,  \sembr{\Delta \vdash c : C} \right) \nto \moabc\left( {-} ,  \sembr{\Delta, a : \Tu \vdash c : C} \right).
    \]
    By the Yoneda lemma, these natural transformations are in natural bijection with morphisms
    \[
      f : \sembr{\Delta \vdash c : C} \to \sembr{\Delta, a : \Tu \vdash c : C},
    \]
    where each such $f$ uniquely determines the natural transformation $\hat f = \moabc\left({-} ,  f\right)$ whose component at $D$ is
    \[
      \hat f_D(g : D \to \sembr{\Delta \vdash c : C}) = f \circ g.
    \]
    In particular, the natural interpretation is induced by the morphism
    \[
      f(p) = \strictfn_{c^+}\left(\lambda (\delta^+,c^+,a^-).(\delta^-,\bot,a^+)\right)
    \]
    where $(\delta^-, a^+) = p(\delta^+,a^-)$.

  \item[Case \rn{$\Tds{} R$}.]
    \[
      \infer[\rn{$\Tds{} R$}]{
        \jtypem{\Psi}{\Delta}{\tSendS{a}{P}}{a}{\Tds A}
      }{
        \jtypem{\Psi}{\Delta}{P}{a}{A}
      }
    \]
    Recall \cref{eq:59}:
    \begin{equation*}
      \begin{aligned}
        &\sembr{\jtypem{\Psi}{\Delta}{\tSendS{a}{P}}{a}{\Tds A}}u\\
        &= \left(\ms{id} \times \left(a^+ : \up\right)\right) \circ \sembr{\jtypem{\Psi}{\Delta}{P}{a}{A}}u.
      \end{aligned}
    \end{equation*}
    Let
    \[
      f(p) = \left(\ms{id} \times \left(a^+ : \up\right)\right) \circ p : \sembr{\Delta \vdash a : A} \to \sembr{\Delta \vdash a : \Tds A}.
    \]
    The natural interpretation is then
    \[
      \moabc\left( {-} ,  f \right) : \moabc\left( {-} ,  \sembr{\Delta \vdash a : A} \right) \nto \moabc\left( {-} ,  \sembr{\Delta \vdash a : \Tds A} \right).
    \]
    Indeed, given a $u \in \sembr{\Psi}$, we get
    \begin{align*}
      &\moabc\left( \sembr{\Psi} ,  f \right)\left(\sembr{\jtypem{\Psi}{\Delta}{P}{a}{A}}\right)(u)\\
      &=\left( f \circ \sembr{\jtypem{\Psi}{\Delta}{P}{a}{A}}\right)(u)\\
      &= f(\sembr{\jtypem{\Psi}{\Delta}{P}{a}{A}}u)\\
      &= \left(\ms{id} \times \left(a^+ : \up\right)\right) \circ \sembr{\jtypem{\Psi}{\Delta}{P}{a}{A}}u.
    \end{align*}

  \item[Case \rn{$\Tds{} L$}.]
    \[
      \infer[\rn{$\Tds{} L$}]{
        \jtypem{\Psi}{\Delta,a : \Tds A}{\tRecvS{a}{P}}{c}{C}
      }{
        \jtypem{\Psi}{\Delta,a : A}{P}{c}{C}
      }
    \]
    Recall \cref{eq:56}:
    \begin{equation*}
      \begin{aligned}
        &\sembr{\jtypem{\Psi}{\Delta,a : \Tds A}{\tRecvS{a}{P}}{c}{C}}u\\
        &= \strictfn_{a^+}\left(\sembr{\jtypem{\Psi}{\Delta, a : A}{P}{c}{C}}u \circ \left(\ms{id} \times \left(a^+ : \down\right)\right)\right).
      \end{aligned}
    \end{equation*}
    Let
    \[
      f(p) = p \circ \left(\ms{id} \times \left(a^+ : \down\right)\right) : \sembr{\Delta, a : A \vdash c : C} \to \sembr{\Delta,a : \Tds A \vdash c : C}.
    \]
    The natural interpretation is then
    \[
      \moabc\left({-} ,  f\right) : \moabc\left( {-} ,  \sembr{\Delta, a : A \vdash c : C}\right) \nto \moabc\left({-} ,  \sembr{\Delta,a : \Tds A \vdash c : C}\right).
    \]

  \item[Case \rn{$\Tus R$}.]
    \[
      \infer[\rn{$\Tus R$}]{
        \jtypem{\Psi}{\Delta}{\tRecvS{a}{P}}{a}{\Tus{A}}
      }{
        \jtypem{\Psi}{\Delta}{P}{a}{A}
      }
    \]
    Analogous to case \rn{$\Tds L$}.

  \item[Case \rn{$\Tus L$}.]
    \[
      \infer[\rn{$\Tus L$}]{
        \jtypem{\Psi}{\Delta,a : \Tus A}{\tSendS{a}{P}}{c}{C}
      }{
        \jtypem{\Psi}{\Delta,a : A}{P}{c}{C}
      }
    \]
    Analogous to case \rn{$\Tds R$}.

  \item[Case \rn{$\Tplus R_k$}.]
    \[
      \infer[\rn{$\Tplus R_k$}]{
        \jtypem{\Psi}{\Delta}{\tSendL{a}{k}{P}}{a}{{\Tplus\{l:A_l\}}_{l \in L}}
      }{
        \jtypem{\Psi}{\Delta}{P}{a}{A_k}\quad(k \in L)
      }
    \]
    Recall \cref{eq:52}:
    \begin{equation*}
      \begin{aligned}
        &\sembr{\jtypem{\Psi}{\Delta}{\tSendL{a}{k}{P}}{a}{\Tplus\{l:A_l\}_{l \in L}}}u\left(\delta^+, \left(a_l^-\right)_{l \in L}\right) = \left(\delta^-, \left(k, \upim{a_k^+}\right)\right)\\
        &\text{where }\sembr{\jtypem{\Psi}{\Delta}{P}{a}{A_k}}u\left(\delta^+, a_k^-\right) = \left(\delta^-, a_k^+\right).
      \end{aligned}
    \end{equation*}
    Let $f : \sembr{\Delta \vdash a : A_k} \to \sembr{\Delta \vdash a : {\Tplus\{l:A_l\}}_{l \in L}}$ be given by
    \[
      f(p)\left(\delta^+, \left(a_l^-\right)_{l \in L}\right) = \left(\delta^-, \left(k, \upim{a_k^+}\right)\right)
    \]
    where $p\left(\delta^+, a_k^-\right) = \left(\delta^-, a_k^+\right)$.
    The natural interpretation is then
    \[
      \moabc\left( {-} ,  f \right) : \moabc\left( {-} ,  \sembr{\Delta \vdash a : A_k} \right) \nto \moabc\left( {-} ,  \sembr{\Delta \vdash a : {\Tplus\{l:A_l\}}_{l \in L}} \right).
    \]

  \item[Case \rn{$\Tplus L$}.]%
    \bgroup
    \newcommand{\thedom}{\sembr{\Delta, a : A_l \vdash c : C}}
    \newcommand{\thecodom}{\sembr{\Delta, a : {\Tplus\{l : A_l\}}_{l \in L} \vdash c : C}}
    \[
      \infer[\rn{$\Tplus L$}]{
        \jtypem{\Psi}{\Delta,a:{\Tplus\{l : A_l\}}_{l \in L}}{\tCase{a}{\left\{l_l \Rightarrow P_l\right\}_{i\in I}}}{c}{C}
      }{
        \jtypem{\Psi}{\Delta,a:A_l}{P_l}{c}{C}\quad(\forall l \in L)
      }
    \]
    Recall \cref{eq:53}:
    \begin{equation*}
      \begin{aligned}
        &\sembr{\jtypem{\Psi}{\Delta,a:\Tplus\{l:A_l\}_{l \in L}}{\tCase{a}{\left\{l \Rightarrow P_l\right\}_{l \in L}}}{c}{C}}u\\
        &= \strictfn_{a^+}\left(\lambda \left(\delta^+, a^+ : \left(k, \upim{a_k^+}\right), c^-\right).\left(\delta^-, a^- : \left(k : a_k^-, l \neq k : \bot\right)_{l \in L}, c^+\right)\right)\\
        &\text{where }(\delta^-, a_k^-, c^+) = \sembr{\jtypem{\Psi}{\Delta,a:A_k}{P_k}{c}{C}}u(\delta^+, a_k^+, c^-).
      \end{aligned}
    \end{equation*}
    We must find a corresponding natural interpretation
    \[
      \eta : \left(\prod_{l \in L} \moabc\left({-} ,  \thedom \right)\right) \nto \moabc\left({-} ,  \thecodom \right).
    \]
    There exists a canonical natural isomorphism
    \[
      \alpha : \left(\prod_{l \in L} \moabc\left({-} ,  \thedom \right)\right) \nto \moabc\left({-} ,  \prod_{l \in L} \thedom \right)
    \]
    whose $D$-component is $\alpha_D\left(\left(p_l\right)_{l \in L}\right)(u) = (\left(p_l(u)\right)_{l \in L}$, so it is sufficient to show that there exists a natural transformation
    \[
      \beta : \moabc\left({-} ,  \prod_{l \in L} \thedom \right) \nto \moabc\left({-} ,  \thecodom \right)
    \]
    such that $\eta = \beta \circ \alpha$.
    Consider the $\beta$ induced by
    \begin{equation*}
      \begin{aligned}
        &f :  \prod_{l \in L} \thedom \to \thecodom\\
        &f\left(\left(p_l\right)_{l \in L}\right) = \strictfn_{a^+}\left(\lambda \left(\delta^+, a^+ : \left(k, \upim{a_k^+}\right), c^-\right).\left(\delta^-, a^- : \left(k : a_k^-, l \neq k : \bot\right)_{l \in L}, c^+\right)\right)\\
        &\text{where } (\delta^-, a_k^-, c^+) = p_k(\delta^+, a_k^+, c^-).
      \end{aligned}
    \end{equation*}
    Recall that we abuse notation to pattern match on the component $a^+$.
    By strictness, we know that it will be an element of the form $\left(k, \upim{a_k^+}\right)$.
    We then get that $\eta = \moabc\left({-} ,  f\right) \circ \alpha$ is the desired natural interpretation.
    \egroup

  \item[Case \rn{$\Tamp R$}.]
    \[
      \infer[\rn{$\Tamp R$}]{
        \jtypem{\Psi}{\Delta}{\tCase{a}{\left\{l \Rightarrow P_l\right\}_{l \in L}}}{a}{{\Tamp\{l :A_l \}}_{l \in L}}
      }{
        \jtypem{\Psi}{\Delta}{P_l}{a}{A_l}\quad(\forall l \in L)
      }
    \]
    Analogous to case \rn{$\Tplus L$}.

  \item[Case \rn{$\Tamp L_k$}.]
    \[
      \infer[\rn{$\Tamp L_k$}]{
        \jtypem{\Psi}{\Delta,a:{\Tamp\{l : A_l\}}_{l \in L}}{\tSendL{a}{k}{P}}{c}{C}
      }{
        \jtypem{\Psi}{\Delta,a:A_k}{P}{c}{C}
        &
        (k \in L)
      }
    \]
    Analogous to case \rn{$\Tplus R_k$}.

  \item[Case \rn{$\Tot R^*$}.]
    \[
      \infer[\rn{$\Tot R^*$}]{
        \jtypem{\Psi}{\Delta, b : B}{\tSendC{a}{b}{P}}{a}{B \Tot A}
      }{
        \jtypem{\Psi}{\Delta}{P}{a}{A}
      }
    \]
    Recall \cref{eq:17}:
    \begin{equation*}
      \begin{aligned}
        &\sembr{\jtypem{\Psi}{\Delta, b : B}{\tSendC{a}{b}{P}}{a}{B \Tot A}}u(\delta^+,b^+,(a^-_B,a^-_A))\\
        &= \left(\delta^-,a_B^-,\upim{\left(b^+,a_A^+\right)}\right)\text{ where }\sembr{\jtypem{\Psi}{\Delta}{P}{a}{A}}u(\delta^+,a_A^-) = (\delta^-,a_A^+).
      \end{aligned}
    \end{equation*}
    Let $f : \sembr{\Delta \vdash a : A} \to \sembr{\Delta, b : B \vdash a : B \Tot A}$ be given by
    \[
      f(p)(\delta^+,b^+,(a^-_B,a^-_A)) = \left(\delta^-,a_B^-,\upim{\left(b^+,a_A^+\right)}\right)
    \]
    where $p(\delta^+,a_A^-) = (\delta^-,a_A^+)$.
    The corresponding natural interpretation is then
    \[
      \moabc\left( {-} ,  f \right) : \moabc\left( {-} ,  \sembr{\Delta \vdash a : A} \right) \nto \moabc\left( {-} ,  \sembr{\Delta, b : B \vdash a : B \Tot A} \right).
    \]

  \item[Case \rn{$\Tot L$}.]
    \[
      \infer[\rn{$\Tot L$}]{
        \jtypem{\Psi}{\Delta, a : B \Tot A}{\tRecvC{b}{a}{P}}{c}{C}
      }{
        \jtypem{\Psi}{\Delta, a : A, b : B}{P}{c}{C}
      }
    \]
    Recall \cref{eq:26}:
    \begin{equation*}
      \begin{aligned}
        &\sembr{\jtypem{\Psi}{\Delta, a : B \Tot A}{\tRecvC{b}{a}{Q}}{c}{C}}u(\delta^+,a^+,c^-)\\
        &= \strictfn_{a^+}\left(\lambda (\delta^+,a^+ : \upim{(a_B^+, a_A^+)},c^-) . (\delta^-,(b^-,a^-),c^+) \right)\\
        &\text{where }\sembr{\jtypem{\Psi}{\Delta, a : A, b : B}{Q}{c}{C}}u(\delta^+,a_A^+,a_B^+,c^-) = (\delta^-,a^-,b^-,c^+).
      \end{aligned}
    \end{equation*}
    Let $f : \sembr{\Delta, a : A, b : B \vdash c : C} \to \sembr{\Delta; a : B \Tot A \vdash c : C}$ be given by
    \[
      f(p) = \strictfn_{a^+}\left(\lambda (\delta^+,a^+ : \upim{(a_B^+, a_A^+)},c^-) . (\delta^-,(b^-,a^-),c^+) \right)
    \]
    where $p(\delta^+,a_A^+,a_B^+,c^-) = (\delta^-,a^-,b^-,c^+)$.
    The corresponding natural interpretation is then
    \[
      \moabc\left({-} ,  f\right) : \moabc\left( {-} ,   \sembr{\Delta, a : A, b : B \vdash c : C} \right) \nto \moabc\left( {-} ,  \sembr{\Delta; a : B \Tot A \vdash c : C} \right).
    \]

  \item[Case \rn{${\Tlolly}R$}.]
    \[
      \infer[\rn{${\Tlolly}R$}]{
        \jtypem{\Psi}{\Delta}{\tRecvC{b}{a}{P}}{a}{B \Tlolly A}
      }{
        \jtypem{\Psi}{\Delta, b : B}{P}{a}{A}
      }
    \]
    Analogous to case \rn{$\Tot L$}.

  \item[Case \rn{${\Tlolly}L$}.]
    \[
      \infer[\rn{${\Tlolly}L$}]{
        \jtypem{\Psi}{\Delta, b : B, a : B \Tlolly A}{\tSendC{a}{bb}{P}}{c}{C}
      }{
        \jtypem{\Psi}{\Delta,a : A}{P}{c}{C}
      }
    \]
    Analogous to case \rn{$\Tot R$}.

  \item[Case \rn{$\Tand{}{} R$}.]
    \[
      \infer[\rn{$\Tand{}{} R$}]{
        \jtypem{\Psi}{\Delta}{\tSendV{a}{M}{P}}{a}{\Tand{\tau}{A}}
      }{
        \jtypef{\Psi}{M}{\tau}
        &
        \jtypem{\Psi}{\Delta}{P}{a}{A}
      }
    \]
    Recall \cref{eq:1005}:
    \begin{equation*}
      \begin{aligned}
        &\sembr{\jtypem{\Psi}{\Delta}{\tSendV{a}{M}{P}}{a}{\Tand{\tau}{A}}}u(\delta^+,a^-)\\
        &= \begin{cases}
          \bot & \text{if }\sembr{\jtypef{\Psi}{M}{\tau}}u = \bot\\
          \left(\delta^-,\upim{\left(v,a^+\right)}\right) & \text{if }\sembr{\jtypef{\Psi}{M}{\tau}}u = v \neq \bot
        \end{cases}\\
        &\text{where }\sembr{\jtypem{\Psi}{\Delta}{P}{a}{A}}u(\delta^+,a^-) = (\delta^-, a^+).
      \end{aligned}
    \end{equation*}
    We must find a corresponding natural interpretation
    \[
      \eta : \moabc\left({-} ,  \sembr{\tau} \right) \times \moabc\left( {-} ,  \sembr{\Delta \vdash a : A} \right) \nto \moabc\left({-} ,  \sembr{\Delta \vdash a : \Tand{\tau}{A}} \right).
    \]
    Let
    \[
      \alpha : \moabc\left({-} ,  \sembr{\tau} \right) \times \moabc\left( {-} ,  \sembr{\Delta \vdash a : A} \right) \nto \moabc\left( {-} ,  \sembr{\tau} \times \sembr{\Delta \vdash a : A} \right)
    \]
    be the canonical natural isomorphism whose $D$-component is $\alpha_D(m,p)(u) = (mu, pu)$.
    Let $f : \sembr{\tau} \times \sembr{\Delta \vdash a : A} \to \sembr{\Delta \vdash a : \Tand{\tau}{A}}$ be the morphism given by
    \[
      f(v,p)(\delta^+,a^-) = \begin{cases}
        \bot & \text{if }v = \bot\\
        \left(\delta^-,\upim{\left(v,a^+\right)}\right) & \text{if }v \neq \bot
      \end{cases}
    \]
    where $p(\delta^+,a^-) = (\delta^-, a^+)$.
    We claim that the corresponding natural interpretation is given by
    \begin{multline*}
      \eta = \moabc\left( {-} ,  f \right) \circ \alpha : \moabc\left({-} ,  \sembr{\tau} \right) \times \moabc\left( {-} ,  \sembr{\Delta \vdash a : A} \right) \nto {} \\
      {} \nto \moabc\left({-} ,  \sembr{\Delta \vdash a : \Tand{\tau}{A}} \right).
    \end{multline*}
    Indeed, given a $u \in \sembr{\Psi}$, we get
    \begin{align*}
      &\eta_{\sembr{\Psi}}\left(\sembr{\jtypef{\Psi}{M}{\tau}}, \sembr{\jtypem{\Psi}{\Delta}{P}{a}{A}}\right)(u)\\
      &= f\left(\sembr{\jtypef{\Psi}{M}{\tau}}u,  \sembr{\jtypem{\Psi}{\Delta}{P}{a}{A}}u\right)\\
      &= \sembr{\jtypem{\Psi}{\Delta}{\tSendV{a}{M}{P}}{a}{\Tand{\tau}{A}}}u.
    \end{align*}

  \item[Case \rn{$\Tand{}{} L$}.]
    \[
      \infer[\rn{$\Tand{}{} L$}]{
        \jtypem{\Psi}{\Delta, a:\Tand{\tau}{A}}{\tRecvV{x}{a}{P}}{c}{C}
      }{
        \jtypem{\Psi,x:\tau}{\Delta, a:A}{P}{c}{C}
      }
    \]
    Recall \cref{eq:1006}:
    \begin{equation*}
      \begin{aligned}
        &\sembr{\jtypem{\Psi}{\Delta, a:\Tand{\tau}{A}}{\tRecvV{x}{a}{P}}{c}{C}}u\\
        &= \strictfn_{a^+}\left( \lambda \left(\delta^+,a^+ : \upim{\left(v, \alpha^+\right)},c^-\right) . \right.\\
        &\qquad\quad\qquad\qquad \left. \sembr{\jtypem{\Psi,x:\tau}{\Delta, a:A}{P}{c}{C}}\upd{u}{x \mapsto v}(\delta^+, \alpha^+, c^-) \right)
      \end{aligned}
    \end{equation*}
    We must find a corresponding natural interpretation
    \[
      \eta : \moabc\left({-} \times (x : \sembr{\tau}) ,  \sembr{\Delta, a : A \vdash c : C } \right) \nto \moabc\left({-} ,  \sembr{\Delta, a : \Tand{\tau}{A} \vdash c : C} \right).
    \]
    Let
    \begin{multline*}
      \Lambda : \moabc\left({-} \times (x : \sembr{\tau}) ,  \sembr{\Delta, a : A \vdash c : C } \right) \nto {}\\
      {} \nto \moabc\left({-} \to \left[ (x : \sembr{\tau}) ,  \sembr{\Delta, a : A \vdash c : C } \right) \right]
    \end{multline*}
    be the natural isomorphism given by the adjunction for the exponential in $\moabc$.
    Let
    \begin{align*}
      f &: \moabc\left( (x : \sembr{\tau}) ,  \sembr{\Delta, a : A \vdash c : C } \right) \to \sembr{\Delta, a : \Tand{\tau}{A} \vdash c : C}\\
      f(p) &= \strictfn_{a^+}\left( \lambda \left(\delta^+,a^+ : \upim{\left(v, \alpha^+\right)},c^-\right).pv(\delta^+, \alpha^+, c^-) \right).
    \end{align*}
    We claim that the corresponding natural interpretation is $\eta = \moabc\left( {-} ,  f \right) \circ \Lambda$.
    To see that this is so, set $p = \sembr{\jtypem{\Psi,x:\tau}{\Delta, a:A}{P}{c}{C}}$ and note that:
    \[
      \Lambda_{\sembr{\Psi}}(p) = \lambda u \in \sembr{\Psi}.\lambda v \in \sembr{\tau}.p\upd{u}{x \mapsto v}.
    \]
    Let $u \in \sembr{\Psi}$ be arbitrary and observe that:
    \begin{align*}
      &\left(\moabc\left( {-} ,  f \right) \circ \Lambda\right)_{\sembr{\Psi}}\left(\sembr{\jtypem{\Psi,x:\tau}{\Delta, a:A}{P}{c}{C}}\right)(u)\\
      &= f \left(\lambda v \in \sembr{\tau}.p\upd{u}{x \mapsto v}\right)\\
      &=  \strictfn_{a^+}\left( \lambda \left(\delta^+,a^+ : \upim{\left(v, \alpha^+\right)},c^-\right). \left(\lambda w \in \sembr{\tau}.p\upd{u}{x \mapsto w}\right)v(\delta^+, \alpha^+, c^-) \right)\\
      &=  \strictfn_{a^+}\left( \lambda \left(\delta^+,a^+ : \upim{\left(v, \alpha^+\right)},c^-\right). p\upd{u}{x \mapsto v}(\delta^+, \alpha^+, c^-) \right)\\
      &= \sembr{\jtypem{\Psi}{\Delta, a:\Tand{\tau}{A}}{\tRecvV{x}{a}{P}}{c}{C}}u.
    \end{align*}

  \item[Case \rn{${\Timp{}{}} R$}.]
    \[
      \infer[\rn{${\Timp{}{}} R$}]{
        \jtypem{\Psi}{\Delta}{\tRecvV{x}{a}{P}}{a}{\Timp{\tau}{A}}
      }{
        \jtypem{\Psi,x:\tau}{\Delta}{P}{a}{A}
      }
    \]
    Analogous to case \rn{${\Tand{}{}} L$}.

  \item[Case \rn{${\Timp{}{}} L$}.]
    \[
      \infer[\rn{${\Timp{}{}} L$}]{
        \jtypem{\Psi}{\Delta,a : \Timp{\tau}{A}}{\tSendV{a}{M}{P}}{c}{C}
      }{
        \jtypef{\Psi}{M}{\tau}
        &
        \jtypem{\Psi}{\Delta, a : A}{P}{c}{C}
      }
    \]
    Analogous to case \rn{${\Tand{}{}} R$}.

  \item[Case \rn{$\rho^+R$}.]
    \[
      \infer[\rn{$\rho^+R$}]{
        \jtypem{\Psi}{\Delta}{\tSendU{a}{P}}{a}{\Trec{\alpha}{A}}
      }{
        \jtypem{\Psi}{\Delta}{P}{a}{[\Trec{\alpha}{A}/\alpha]A}
        &
        {-} \vdash \jisst[+]{\Trec{\alpha}{A}}
      }
    \]
    Recall \cref{eq:fossacs:2}:
    \begin{equation*}
      \begin{aligned}
        &\sembr{\jtypem{\Psi}{\Delta}{\tSendU{a}{P}}{a}{\Trec{\alpha}{A}}}u\\
        &= \left(\ms{id} \times \left(a^+ : \ms{Fold}\right)\right)
        \circ \sembr{\jtypem{\Psi}{\Delta}{P}{a}{[\Trec{\alpha}{A}/\alpha]A}}u
        \circ \left(\ms{id} \times \left(a^- : \ms{Unfold}\right)\right)
      \end{aligned}
    \end{equation*}
    Let $f : \sembr{\jtypem{\Psi}{\Delta}{P}{a}{[\Trec{\alpha}{A}/\alpha]A}} \to \sembr{\jtypem{\Psi}{\Delta}{\tSendU{a}{P}}{a}{\Trec{\alpha}{A}}}$ be given by
    \[
      f(p) = \left(\ms{id} \times \left(a^+ : \ms{Fold}\right)\right)
      \circ p
      \circ \left(\ms{id} \times \left(a^- : \ms{Unfold}\right)\right).
    \]
    This is well-defined by semantic substitution of session types (\cref{prop:11}) and the definitions of $\ms{Fold}$ and $\ms{Unfold}$.
    The corresponding natural interpretation is then
    \begin{multline*}
      \moabc\left( {-} ,  f \right) : \moabc\left( {-} , \sembr{\jtypem{\Psi}{\Delta}{P}{a}{[\Trec{\alpha}{A}/\alpha]A}} \right) \nto {}\\
      {} \nto \moabc\left( {-}, \sembr{\jtypem{\Psi}{\Delta}{\tSendU{a}{P}}{a}{\Trec{\alpha}{A}}} \right).
    \end{multline*}

  \item[Case \rn{$\rho^+L$}.]
    \[
      \infer[\rn{$\rho^+L$}]{
        \jtypem{\Psi}{\Delta, a : \Trec{\alpha}{A}}{\tRecvU{a}{P}}{c}{C}
      }{
        \jtypem{\Psi}{\Delta, a : [\Trec{\alpha}{A}/\alpha]A}{P}{c}{C}
        &
        {-} \vdash \jisst[+]{\Trec{\alpha}{A}}
      }
    \]
    Analogous to case \rn{$\rho^+R$}.

  \item[Case \rn{$\rho^-R$}.]
    \[
      \infer[\rn{$\rho^-R$}]{
        \jtypem{\Psi}{\Delta}{\tRecvU{a}{P}}{a}{\Trec{\alpha}{A}}
      }{
        \jtypem{\Psi}{\Delta}{P}{a}{[\Trec{\alpha}{A}/\alpha]A}
        &
        {-} \vdash \jisst[-]{\Trec{\alpha}{A}}
      }
    \]
    Analogous to case \rn{$\rho^+R$}.

  \item[Case \rn{$\rho^-L$}.]
    \[
      \infer[\rn{$\rho^-L$}]{
        \jtypem{\Psi}{\Delta, a : \Trec{\alpha}{A}}{\tSendU{a}{P}}{c}{C}
      }{
        \jtypem{\Psi}{\Delta, a : [\Trec{\alpha}{A}/\alpha]A}{P}{c}{C}
        &
        {-} \vdash \jisst[-]{\Trec{\alpha}{A}}
      }
    \]
    Analogous to case \rn{$\rho^+R$}.

  \item[Case \rn{E-\{\}}.]
    \[
      \infer[\rn{E-\{\}}]{
        \jtypem{\Psi}{\overline{a_i:A_i}}{\tProc{a}{M}{\overline a_i}}{a}{A}
      }{
        \jtypef{\Psi}{M}{\Tproc{a:A}{\overline{a_i:A_i}}}
      }
    \]
    Recall \cref{eq:115}:
    \begin{equation*}
      \begin{aligned}
        &\sembr{\jtypem{\Psi}{\overline{a_i:A_i}}{\tProc{a}{M}{\overline a_i}}{a}{A}}u\\
        &= \down\left(\sembr{\jtypef{\Psi}{M}{\Tproc{a:A}{\overline{a_i:A_i}}}}u\right)
      \end{aligned}
    \end{equation*}
    We must find a corresponding natural interpretation
    \[
      \eta : \moabc\left( {-} ,  \sembr{\Tproc{a:A}{\overline{a_i:A_i}}} \right)
      \nto
      \moabc\left( {-} ,  \sembr{\overline{a_i:A_i}, \Delta \vdash a : A} \right).
    \]
    It is $\eta = \moabc\left( {-} ,  \down \right)$.\qedhere
  \end{description}
\end{proof}

\subsection{Semantic Weakening}
\label{sec:semantic-weakening}

We show that weakening is semantically well-behaved, \ie, that the semantic clauses are coherent~\cite[p.~218]{tennent_1995:_denot_seman}.

\cohterpro*

\begin{proof}
  By induction on the derivation of $\jtypef{\Psi}{M}{\tau}$ and $\jtypem{\Psi}{\Delta}{P}{a}{A}$.

  Except where stated otherwise, each of rule cases uses the same proof outline.
  We refer to it below as the ``standard proof''.
  We give this proof for a generic rule for forming functional terms.
  The proof for rules forming processes is analogous.
  Consider the rule
  \[
    \adjustbox{width=\linewidth,keepaspectratio}\bgroup
    \infer{
      \jtypef{\Psi}{F_{\mathcal{J}_1,\dotsc,\mathcal{J}_l}(P_1,\dotsc,P_n,M_1,\dotsc,M_m)}{\tau}
    }{
      \jtypem{\Psi}{\Delta_1}{P_1}{c_1}{C_1}
      &
      \cdots
      &
      \jtypem{\Psi}{\Delta_n}{P_n}{c_n}{C_n}
      &
      \jtypef{\Psi}{M_1}{\tau_1}
      &
      \cdots
      &
      \jtypef{\Psi}{M_m}{\tau_m}
      &
      \mathcal{J}_1
      &
      \cdots
      &
      \mathcal{J}_l
    }
    \egroup
  \]
  Assume its interpretation is given by
  \begin{equation}
    \label[intn]{eq:fscd:47}
    \begin{aligned}
      &\sembr{\jtypef{\Psi}{F_{\mathcal{J}_1,\dotsc,\mathcal{J}_l}(P_1,\dotsc,P_n,M_1,\dotsc,M_m)}{\tau}}\\
      &= \sembr{F_{\mathcal{J}_1,\dotsc,\mathcal{J}_l}}_{\sembr{\Psi}}\left(
        \sembr{\jtypem{\Psi}{\Delta_1}{P_1}{c_1}{C_1}},
        \dotsc,
        \sembr{\jtypem{\Psi}{\Delta_n}{P_n}{c_n}{C_n}},\right.\\
      &\qquad\qquad\left.
        \sembr{\jtypef{\Psi}{M_1}{\tau_1}},
        \dotsc,
        \sembr{\jtypef{\Psi}{M_m}{\tau_m}}
      \right)
    \end{aligned}
  \end{equation}
  where $\sembr{F_{\mathcal{J}_1,\dotsc,\mathcal{J}_l}}$ is a natural interpretation
  \begin{equation*}
    \sembr{F_{\mathcal{J}_1,\dotsc,\mathcal{J}_l}} :
    \left(\prod_{i = 1}^n \moabc\left( {-} ,  \sembr{\Delta_i \vdash c_i : C_i} \right)\right)
    \times
    \left(\prod_{i = 1}^m \moabc\left( {-} ,  \sembr{\tau_i} \right) \right)
    \nto
    \moabc\left( {-} ,  \sembr{\tau} \right).
  \end{equation*}

  Given any other functional context $\Phi$ disjoint from $\Psi$, we would like to show that
  \begin{align*}
    &\sembr{\jtypef{\Phi,\Psi}{F_{\mathcal{J}_1,\dotsc,\mathcal{J}_l}(P_1,\dotsc,P_n,M_1,\dotsc,M_m)}{\tau}}\\
    &= \sembr{\jtypef{\Psi}{F_{\mathcal{J}_1,\dotsc,\mathcal{J}_l}(P_1,\dotsc,P_n,M_1,\dotsc,M_m)}{\tau}} \circ \pi^{\Phi,\Psi}_\Psi.
  \end{align*}
  By the induction hypothesis, we know that for $1 \leq i \leq n$ and $1 \leq j \leq m$,
  \begin{align*}
    \sembr{\jtypem{\Phi,\Psi}{\Delta_i}{P_i}{c_i}{C_i}} &= \sembr{\jtypem{\Psi}{\Delta_i}{P_i}{c_i}{C_i}} \circ \pi^{\Phi,\Psi}_\Psi\\
    \sembr{\jtypef{\Phi,\Psi}{M_j}{\tau_j}} &= \sembr{\jtypef{\Psi}{M_j}{\tau_j}} \circ \pi^{\Phi,\Psi}_\Psi.
  \end{align*}
  Using these facts we get:
  \begin{align*}
    &\sembr{\jtypef{\Phi,\Psi}{F_{\mathcal{J}_1,\dotsc,\mathcal{J}_l}(P_1,\dotsc,P_n,M_1,\dotsc,M_m)}{\tau}}\\
    &= \sembr{F_{\mathcal{J}_1,\dotsc,\mathcal{J}_l}}_{\sembr{\Phi}}\left(
      \sembr{\jtypem{\Phi,\Psi}{\Delta_1}{P_1}{c_1}{C_1}},
      \dotsc,
      \sembr{\jtypem{\Phi,\Psi}{\Delta_n}{P_n}{c_n}{C_n}},\right.\\
    &\qquad\qquad\left.
      \sembr{\jtypef{\Phi,\Psi}{M_1}{\tau_1}},
      \dotsc,
      \sembr{\jtypef{\Phi,\Psi}{M_m}{\tau_m}}
      \right)\\
    \shortintertext{which by the induction hypothesis,}
    &= \sembr{F_{\mathcal{J}_1,\dotsc,\mathcal{J}_l}}_{\sembr{\Phi}}\left(
      \sembr{\jtypem{\Psi}{\Delta_1}{P_1}{c_1}{C_1}} \circ \pi^{\Phi,\Psi}_\Psi,
      \dotsc,
      \sembr{\jtypem{\Psi}{\Delta_n}{P_n}{c_n}{C_n}}  \circ \pi^{\Phi,\Psi}_\Psi,\right.\\
    &\qquad\qquad\left.
      \sembr{\jtypef{\Psi}{M_1}{\tau_1}} \circ \pi^{\Phi,\Psi}_\Psi,
      \dotsc,
      \sembr{\jtypef{\Psi}{M_m}{\tau_m}} \circ \pi^{\Phi,\Psi}_\Psi
      \right)\\
    \shortintertext{which by naturality of $\sembr{F_{\mathcal{J}_1,\dotsc,\mathcal{J}_l}}$,}
    &= \sembr{F_{\mathcal{J}_1,\dotsc,\mathcal{J}_l}}_{\sembr{\Psi}}\left(
      \sembr{\jtypem{\Psi}{\Delta_1}{P_1}{c_1}{C_1}}, \dotsc,
      \sembr{\jtypem{\Psi}{\Delta_n}{P_n}{c_n}{C_n}},\right.\\
    &\qquad\qquad\left.
      \sembr{\jtypef{\Psi}{M_1}{\tau_1}},
      \dotsc,
      \sembr{\jtypef{\Psi}{M_m}{\tau_m}}
      \right) \circ \pi^{\Phi,\Psi}_\Psi\\
    &= \sembr{\jtypef{\Psi}{F_{\mathcal{J}_1,\dotsc,\mathcal{J}_l}(P_1,\dotsc,P_n,M_1,\dotsc,M_m)}{\tau}} \circ \pi^{\Phi,\Psi}_\Psi.
  \end{align*}
  This is what we wanted to show.

  \begin{description}
  \item[Case \rn{I-\{\}}.]
    By the standard proof and \cref{prop:fscd:2}.

  \item[Case \rn{F-Var}.]
    \[
      \infer[\rn{F-Var}]{\jtypef{\Psi, x: \tau}{x}{\tau}}{\mathstrut}
    \]
    Recall \cref{eq:39}:
    \begin{align*}
      \sembr{\jtypef{\Psi,x:\tau}{x}{\tau}}u &= \pi^{\Psi,x}_{x}u.
    \end{align*}
    We use identities of products and projections to compute:
    \begin{align*}
      &\sembr{\jtypef{\Phi,\Psi,x : \tau}{x}{\tau}}\\
      &= \pi^{\Phi,\Psi,x}_x\\
      &= \pi^{\Psi,x}_x \circ \pi^{\Phi,\Psi,x}_{\Psi,x}\\
      &= \sembr{\jtypef{\Psi,x : \tau}{x}{\tau}} \circ \pi^{\Phi,\Psi,x}_{\Psi,x}.
    \end{align*}

  \item[Case \rn{F-Fix}.]
    \[
      \infer[\rn{F-Fix}]{
        \jtypef{\Psi}{\tFix{x}{M}}{\tau}
      }{
        \jtypef{\Psi,x:\tau}{M}{\tau}
      }
    \]
    Recall \cref{eq:28}:
    \begin{align*}
      \sembr{\jtypef{\Psi}{\tFix{x}{M}}{\tau}}u &= \sfix{\sembr{\jtypef{\Psi,x:\tau}{M}{\tau}}}u
    \end{align*}
    The result follows from a tweak of the standard proof and \cref{prop:fscd:2}.
    Let $\eta : \moabc\left( {-} \times \sembr{x : \tau} ,  \sembr{\tau} \right) \nto \moabc\left( {-} ,  \sembr{\tau} \right)$ be the natural interpretation given by \cref{prop:fscd:2}.
    Then
    \begin{align*}
      &\sembr{\jtypef{\Phi,\Psi}{\tFix{x}{M}}{\tau}}\\
      &= \eta_{\sembr{\Phi}}\left(\sembr{\jtypef{\Phi,\Psi,x:\tau}{M}{\tau}}\right)\\
      \shortintertext{which by the induction hypothesis:}
      &= \eta_{\sembr{\Phi}}\left(\sembr{\jtypef{\Psi,x:\tau}{M}{\tau}} \circ \pi^{\Phi,\Psi,x}_{\Psi,x}\right)\\
      \shortintertext{which by an identity of projections:}
      &= \eta_{\sembr{\Phi}}\left(\sembr{\jtypef{\Psi,x:\tau}{M}{\tau}} \circ \left(\pi^{\Phi,\Psi}_{\Psi} \times \left(x : \ms{id}_{\sembr{\tau}}\right)\right)\right)\\
      \shortintertext{which by naturality:}
      &= \eta_{\sembr{\Psi}}\left(\sembr{\jtypef{\Psi,x:\tau}{M}{\tau}}\right) \circ \pi^{\Phi,\Psi}_{\Psi}\\
      &= \sembr{\jtypef{\Psi}{\tFix{x}{M}}{\tau}} \circ \pi^{\Phi,\Psi}_{\Psi}.
    \end{align*}

  \item[Case \rn{F-Fun}.]
    \[
      \infer[\rn{F-Fun}]{
        \jtypef{\Psi}{\lambda x : \tau.M}{\tau \to \sigma}
      }{
        \jtypef{\Psi, x:\tau}{M}{\sigma}
      }
    \]
    This case is analogous to the case \rn{F-Fix}.

  \item[Case \rn{F-App}.]
    By the standard proof and \cref{prop:fscd:2}.

  \item[Case \rn{Fwd}.]
    By the standard proof and \cref{prop:fscd:2}.

  \item[Case \rn{Cut}.]
    By the standard proof and \cref{prop:fscd:2}.

  \item[Case \rn{$\Tu R$}.]
    By the standard proof and \cref{prop:fscd:2}.

  \item[Case \rn{$\Tu L$}.]
    By the standard proof and \cref{prop:fscd:2}.

  \item[Case \rn{$\Tds{} R$}.]
    By the standard proof and \cref{prop:fscd:2}.

  \item[Case \rn{$\Tds{} L$}.]
    By the standard proof and \cref{prop:fscd:2}.

  \item[Case \rn{$\Tus R$}.]
    By the standard proof and \cref{prop:fscd:2}.

  \item[Case \rn{$\Tus L$}.]
    By the standard proof and \cref{prop:fscd:2}.

  \item[Case \rn{$\Tplus R_k$}.]
    By the standard proof and \cref{prop:fscd:2}.

  \item[Case \rn{$\Tplus L$}.]
    By the standard proof and \cref{prop:fscd:2}.

  \item[Case \rn{$\Tamp R$}.]
    By the standard proof and \cref{prop:fscd:2}.

  \item[Case \rn{$\Tamp L_k$}.]
    By the standard proof and \cref{prop:fscd:2}.

  \item[Case \rn{$\Tot R^*$}.]
    By the standard proof and \cref{prop:fscd:2}.

  \item[Case \rn{$\Tot L$}.]
    By the standard proof and \cref{prop:fscd:2}.

  \item[Case \rn{${\Tlolly}R$}.]
    By the standard proof and \cref{prop:fscd:2}.

  \item[Case \rn{${\Tlolly}L$}.]
    By the standard proof and \cref{prop:fscd:2}.

  \item[Case \rn{$\Tand{}{} R$}.]
    By the standard proof and \cref{prop:fscd:2}.

  \item[Case \rn{$\Tand{}{} L$}.]
    \[
      \infer[\rn{$\Tand{}{} L$}]{
        \jtypem{\Psi}{\Delta, a:\Tand{\tau}{A}}{\tRecvV{x}{a}{P}}{c}{C}
      }{
        \jtypem{\Psi,x:\tau}{\Delta, a:A}{P}{c}{C}
      }
    \]
    This case is analogous to the case \rn{F-Fix}.

  \item[Case \rn{${\Timp{}{}} R$}.]
    \[
      \infer[\rn{${\Timp{}{}} R$}]{
        \jtypem{\Psi}{\Delta}{\tRecvV{x}{a}{P}}{a}{\Timp{\tau}{A}}
      }{
        \jtypem{\Psi,x:\tau}{\Delta}{P}{a}{A}
      }
    \]
    This case is analogous to the case \rn{F-Fix}.

  \item[Case \rn{${\Timp{}{}} L$}.]
    By the standard proof and \cref{prop:fscd:2}.

  \item[Case \rn{$\rho^+R$}.]
    By the standard proof and \cref{prop:fscd:2}.

  \item[Case \rn{$\rho^+L$}.]
    By the standard proof and \cref{prop:fscd:2}.

  \item[Case \rn{$\rho^-R$}.]
    By the standard proof and \cref{prop:fscd:2}.

  \item[Case \rn{$\rho^-L$}.]
    By the standard proof and \cref{prop:fscd:2}.

  \item[Case \rn{E-\{\}}.]
    By the standard proof and \cref{prop:fscd:2}.\qedhere
  \end{description}
\end{proof}

\subsection{Semantic Substitution}
\label{sec:semant-subst}

Our goal in this section is to show that substitution is given by composition.
\defin{Context morphisms} give us a semantic account of substitutions.
The judgment $\jcmf{\sigma}{\Phi}{\Psi}$ means the substitution $\sigma$ is a morphism of functional variable contexts from $\Phi$ to $\Psi$.
It is inductively defined by the rules
\[
  \infer[\rn{S-T-Empty}]{
    \jcmf{\cdot}{\Phi}{\cdot}
  }{
  }
  \qquad
  \infer[\rn{S-T-F}]{
    \jcmf{\sigma,M}{\Phi}{\Psi,x:\tau}
  }{
    \jcmf{\sigma}{\Phi}{\Psi}
    &
    \jtypef{\Phi}{M}{\tau}
  }
\]
Consider a morphism $\sigma = N_1,\dotsc,N_n$ with $n \geq 0$ satisfying $\jcmf{\sigma{}}{\Phi}{{x_1:\tau_1},\dotsc,{x_n:\tau_n}}$.
Given a functional term $\jtypef{x_1:\tau_1,\dotsc,x_n:\tau_n}{M}{\tau}$ or a process $\jtypem{x_1:\tau_1,\dotsc,x_n:\tau_n}{\Delta}{P}{c}{C}$, we write $\sigma M$ and $\sigma P$ for the results of the simultaneous substitutions $[N_1,\dotsc,N_n/x_1,\dotsc,x_n] M$ and $[N_1,\dotsc, N_n/x_1,\dotsc,x_n] P$, respectively.

The judgments $\jtypef{\Psi}{M}{\tau}$ and $\jtypem{\Psi}{\Delta}{P}{a}{A}$ satisfy the following syntactic substitution property:

\begin{proposition}[Syntactic Substitution of Terms]
  \label{prop:fscd:3}
  Let $\jcmf{\sigma}{\Phi}{\Psi}$ be arbitrary.
  \begin{proplist}
  \item If $\jtypef{\Psi}{N}{\tau}$, then $\jtypef{\Phi}{\sigma N}{\tau}$.
  \item If $\jtypem{\Psi}{\Delta}{P}{c}{C}$, then $\jtypem{\Phi}{\Delta}{\sigma P}{c}{C}$.
  \end{proplist}
\end{proposition}

\begin{proof}
  By induction on the derivation of $\jtypef{\Psi}{M}{\tau}$ and $\jtypem{\Psi}{\Delta}{P}{a}{A}$.
\end{proof}

A context morphism $\jcmf{\sigma}{\Phi}{\Psi}$ is interpreted as a continuous morphism $\sembr{\jcmf{\sigma}{\Phi}{\Psi}} : \sembr{\Phi} \to \sembr{\Psi}$.
It is recursively defined on the derivation of $\jcmf{\sigma}{\Phi}{\Psi}$.
The interpretations of \rn{S-T-Empty} and \rn{S-T-F} are respectively
\begin{align}
  \sembr{\jcmf{\cdot}{\Phi}{\cdot}} &= \top_{\sembr{\Phi}}\label[intn]{eq:fscd:32}\\
  \sembr{\jcmf{\sigma, M}{\Phi}{\Psi,x:\tau}} &= \langle \sembr{\jcmf{\sigma}{\Phi}{\Psi}}, x : \sembr{\jtypef{\Phi}{M}{\tau}} \rangle\label[intn]{eq:fscd:33}
\end{align}
where $\top_{\sembr{\Phi}}$ is the unique morphism from $\sembr{\Phi}$ to the terminal object $\top$ of $\moabc$.

\begin{lemma}[Weakening of Context Morphisms]
  \label{lemma:fscd:13}
  Let $\jcmf{\sigma}{\Phi}{\Psi}$ be arbitrary and $\Gamma,\Phi$ a context.
  Then $\jcmf{\sigma}{\Gamma,\Phi}{\Psi}$ and
  \[
    \sembr{\jcmf{\sigma}{\Gamma,\Phi}{\Psi}} = \sembr{\jcmf{\sigma}{\Phi}{\Psi}} \circ \pi^{\Gamma,\Phi}_\Phi.
  \]
\end{lemma}

\begin{proof}
  By induction on the derivation of $\jcmf{\sigma}{\Phi}{\Psi}$.

  \begin{description}[listparindent=\parindent]
  \item[Case \rn{S-T-Empty}.]
    Then $\jcmf{\sigma}{\Gamma,\Phi}{\cdot}$ by \rn{S-T-Empty}.
    By terminality,
    \[
      \sembr{\jcmf{\sigma}{\Gamma,\Phi}{\cdot}} = \top_{\sembr{\Gamma,\Phi}} = \top_{\sembr{\Phi}} \circ \pi^{\Gamma,\Phi}_\Phi = \sembr{\jcmf{\sigma}{\Phi}{\cdot}} \circ \pi^{\Gamma,\Phi}_\Phi.
    \]
  \item[Case \rn{S-T-F}.]
    \[
      \infer[\rn{S-T-F}]{
        \jcmf{\sigma,M}{\Phi}{\Psi,x:\tau}
      }{
        \jcmf{\sigma}{\Phi}{\Psi}
        &
        \jtypef{\Phi}{M}{\tau}
      }
    \]
    By the induction hypothesis, $\jcmf{\sigma}{\Gamma,\Phi}{\Psi}$ and
    \[
      \sembr{\jcmf{\sigma}{\Gamma,\Phi}{\Psi}} = \sembr{\jcmf{\sigma}{\Phi}{\Psi}} \circ \pi^{\Gamma,\Phi}_\Phi.
    \]
    By weakening, $\jtypef{\Gamma,\Phi}{M}{\tau}$, and by \cref{prop:extended:9},
    \[
      \sembr{\jtypef{\Gamma,\Phi}{M}{\tau}} = \sembr{\jtypef{\Phi}{M}{\tau}} \circ \pi^{\Gamma,\Phi}_\Phi.
    \]
    So $\jcmf{\sigma,M}{\Gamma,\Phi}{\Psi,x:\tau}$ by \rn{S-T-F}.
    By \cref{eq:fscd:33},
    \begin{align*}
      &\sembr{\jcmf{\sigma,M}{\Gamma,\Phi}{\Psi,x:\tau}}\\
      &= \langle \sembr{\jcmf{\sigma}{\Gamma,\Phi}{\Psi}}, x : \sembr{\jtypef{\Gamma,\Phi}{M}{\tau}} \rangle\\
      &= \langle \sembr{\jcmf{\sigma}{\Phi}{\Psi}} \circ \pi^{\Gamma,\Phi}_\Phi, x : \sembr{\jtypef{\Phi}{M}{\tau}} \circ \pi^{\Gamma,\Phi}_\Phi \rangle\\
      &= \langle \sembr{\jcmf{\sigma}{\Phi}{\Psi}}, x : \sembr{\jtypef{\Phi}{M}{\tau}} \rangle  \circ \pi^{\Gamma,\Phi}_\Phi\\
      &= \sembr{\jcmf{\sigma,M}{\Phi}{\Psi,x:\tau}}  \circ \pi^{\Gamma,\Phi}_\Phi.
    \end{align*}
  \end{description}
  We conclude the result by induction.
\end{proof}

\begin{proposition}[Semantic Substitution of Terms]\label{prop:extended:10}
  Let $\jcmf{\sigma}{\Phi}{\Psi}$ be arbitrary.
  \begin{proplist}
  \item If $\jtypef{\Psi}{N}{\tau}$, then $\sembr{\jtypef{\Phi}{\sigma N}{\tau}} = \sembr{\jtypef{\Psi}{N}{\tau}} \circ \sembr{\jcmf{\sigma}{\Phi}{\Psi}}$. \label{item:extended:2}
  \item If $\jtypem{\Psi}{\Delta}{P}{c}{C}$, then $\sembr{\jtypem{\Phi}{\Delta}{\sigma P}{c}{C}} = \sembr{\jtypem{\Psi}{\Delta}{P}{c}{C}} \circ \sembr{\jcmf{\sigma}{\Phi}{\Psi}}$. \label{item:extended:3}
  \end{proplist}
\end{proposition}

\begin{proof}
  By induction on the derivation of $\jtypef{\Psi}{M}{\tau}$ and $\jtypem{\Psi}{\Delta}{P}{a}{A}$.

  Except where stated otherwise, each of rule cases uses the same proof outline.
  We refer to it below as the ``standard proof''.
  We give this proof for a generic rule.
  Consider the rule
  \[
    \adjustbox{width=\linewidth,keepaspectratio}\bgroup
    \infer{
      \jtypef{\Psi}{F_{\mathcal{J}_1,\dotsc,\mathcal{J}_l}(P_1,\dotsc,P_n,M_1,\dotsc,M_m)}{\tau}
    }{
      \jtypem{\Psi}{\Delta_1}{P_1}{c_1}{C_1}
      &
      \cdots
      &
      \jtypem{\Psi}{\Delta_n}{P_n}{c_n}{C_n}
      &
      \jtypef{\Psi}{M_1}{\tau_1}
      &
      \cdots
      &
      \jtypef{\Psi}{M_m}{\tau_m}
      &
      \mathcal{J}_1
      &
      \cdots
      &
      \mathcal{J}_l
    }
    \egroup
  \]
  Assume its interpretation is given by
  \begin{equation}
    \label[intn]{eq:fscd:1}
    \begin{aligned}
      &\sembr{\jtypem{\Psi}{\Delta}{F_{\mathcal{J}_1,\dotsc,\mathcal{J}_l}(P_1,\dotsc,P_n,M_1,\dotsc,M_m)}{c}{C}}\\
      &= \sembr{F_{\mathcal{J}_1,\dotsc,\mathcal{J}_l}}_{\sembr{\Psi}}\left(
        \sembr{\jtypem{\Psi}{\Delta_1}{P_1}{c_1}{C_1}},
        \dotsc,
        \sembr{\jtypem{\Psi}{\Delta_n}{P_n}{c_n}{C_n}},\right.\\
      &\qquad\qquad\left.
        \sembr{\jtypef{\Psi}{M_1}{\tau_1}},
        \dotsc,
        \sembr{\jtypef{\Psi}{M_m}{\tau_m}}
      \right)
    \end{aligned}
  \end{equation}
  where $\sembr{F_{\mathcal{J}_1,\dotsc,\mathcal{J}_l}}$ is a natural interpretation
  \begin{equation*}
    \sembr{F_{\mathcal{J}_1,\dotsc,\mathcal{J}_l}} :
    \left(\prod_{i = 1}^n \moabc\left( {-} ,  \sembr{\Delta_i \vdash c_i : C_i} \right)\right)
    \times
    \left(\prod_{i = 1}^m \moabc\left( {-} ,  \sembr{\tau_i} \right) \right)
    \nto
    \moabc\left( {-} ,  \sembr{\tau} \right).
  \end{equation*}
  Given any context morphism $\jcmf{\sigma}{\Phi}{\Psi}$, we would like to show that
  \begin{align*}
    &\sembr{\jtypem{\Phi}{\Delta}{\sigma\left(F_{\mathcal{J}_1,\dotsc,\mathcal{J}_l}(P_1,\dotsc,P_n,M_1,\dotsc,M_m)\right)}{c}{C}}\\
    &= \sembr{\jtypem{\Psi}{\Delta}{F_{\mathcal{J}_1,\dotsc,\mathcal{J}_l}(P_1,\dotsc,P_n,M_1,\dotsc,M_m)}{c}{C}} \circ \sembr{\jcmf{\sigma}{\Phi}{\Psi}}.
  \end{align*}
  By the definition of syntactic substitution, we know that
  \[
    \sigma\left(F_{\mathcal{J}_1,\dotsc,\mathcal{J}_l}(P_1,\dotsc,P_n,M_1,\dotsc,M_m)\right) = F_{\mathcal{J}_1,\dotsc,\mathcal{J}_l}(\sigma P_1,\dotsc,\sigma P_n,\sigma M_1,\dotsc,\sigma M_m).
  \]
  By the induction hypothesis, we know that for $1 \leq i \leq n$ and $1 \leq j \leq m$,
  \begin{align*}
    \sembr{\jtypem{\Phi}{\Delta_i}{\sigma P_i}{c_i}{C_i}} &= \sembr{\jtypem{\Psi}{\Delta_i}{P_i}{c_i}{C_i}} \circ \sembr{\jcmf{\sigma}{\Phi}{\Psi}}\\
    \sembr{\jtypef{\Phi}{\sigma M_j}{\tau_j}} &= \sembr{\jtypef{\Psi}{M_j}{\tau_j}} \circ \sembr{\jcmf{\sigma}{\Phi}{\Psi}}.
  \end{align*}
  Using these facts we get:
  \begin{align*}
    &\sembr{\jtypem{\Phi}{\Delta}{\sigma\left(F_{\mathcal{J}_1,\dotsc,\mathcal{J}_l}(P_1,\dotsc,P_n,M_1,\dotsc,M_m)\right)}{c}{C}}\\
    &= \sembr{\jtypem{\Phi}{\Delta}{F_{\mathcal{J}_1,\dotsc,\mathcal{J}_l}(\sigma P_1,\dotsc,\sigma P_n,\sigma M_1,\dotsc,\sigma M_m)}{c}{C}}\\
    &= \sembr{F_{\mathcal{J}_1,\dotsc,\mathcal{J}_l}}_{\sembr{\Phi}}\left(
      \sembr{\jtypem{\Phi}{\Delta_1}{\sigma P_1}{c_1}{C_1}},
      \dotsc,
      \sembr{\jtypem{\Phi}{\Delta_n}{\sigma P_n}{c_n}{C_n}},\right.\\
    &\qquad\qquad\left.
      \sembr{\jtypef{\Phi}{\sigma M_1}{\tau_1}},
      \dotsc,
      \sembr{\jtypef{\Phi}{\sigma M_m}{\tau_m}}
      \right)\\
    \shortintertext{which by the induction hypothesis,}
    &= \sembr{F_{\mathcal{J}_1,\dotsc,\mathcal{J}_l}}_{\sembr{\Phi}}\left(
      \sembr{\jtypem{\Psi}{\Delta_1}{P_1}{c_1}{C_1}} \circ \sembr{\jcmf{\sigma}{\Phi}{\Psi}},\dotsc,\right.\\
    &\qquad\qquad\sembr{\jtypem{\Psi}{\Delta_n}{P_n}{c_n}{C_n}} \circ \sembr{\jcmf{\sigma}{\Phi}{\Psi}},\\
    &\qquad\qquad\left.
      \sembr{\jtypef{\Psi}{M_1}{\tau_1}} \circ \sembr{\jcmf{\sigma}{\Phi}{\Psi}},
      \dotsc,
      \sembr{\jtypef{\Psi}{M_m}{\tau_m}} \circ \sembr{\jcmf{\sigma}{\Phi}{\Psi}}
      \right)\\
    \shortintertext{which by naturality of $\sembr{F_{\mathcal{J}_1,\dotsc,\mathcal{J}_l}}$,}
    &= \sembr{F_{\mathcal{J}_1,\dotsc,\mathcal{J}_l}}_{\sembr{\Psi}}\left(
      \sembr{\jtypem{\Psi}{\Delta_1}{P_1}{c_1}{C_1}}, \dotsc,
      \sembr{\jtypem{\Psi}{\Delta_n}{P_n}{c_n}{C_n}},\right.\\
    &\qquad\qquad\left.
      \sembr{\jtypef{\Psi}{M_1}{\tau_1}},
      \dotsc,
      \sembr{\jtypef{\Psi}{M_m}{\tau_m}}
      \right) \circ \sembr{\jcmf{\sigma}{\Phi}{\Psi}}\\
    &= \sembr{\jtypem{\Psi}{\Delta}{F_{\mathcal{J}_1,\dotsc,\mathcal{J}_l}(P_1,\dotsc,P_n,M_1,\dotsc,M_m)}{c}{C}} \circ \sembr{\jcmf{\sigma}{\Phi}{\Psi}}.
  \end{align*}

  \begin{description}
  \item[Case \rn{I-\{\}}.]
    \[
      \infer[\rn{I-\{\}}]{
        \jtypef{\Psi}{\tProc{a}{P}{\overline{a_i}}}{\Tproc{a:A}{\overline{a_i:A_i}}}
      }{
        \jtypem{\Psi}{\overline{a_i:A_i}}{P}{a}{A}
      }
    \]
    By the standard proof and \cref{prop:fscd:2}.

  \item[Case \rn{F-Var}.]
    \[
      \infer[\rn{F-Var}]{\jtypef{\Psi, x: \tau}{x}{\tau}}{\mathstrut}
    \]
    Recall \cref{eq:39}:
    \begin{align*}
      \sembr{\jtypef{\Psi,x:\tau}{x}{\tau}}u &= \pi^{\Psi,x}_{x}u.
    \end{align*}
    Let $\jcmf{\sigma, M}{\Phi}{\Psi, x : \tau}$ be arbitrary.
    Then
    \begin{align*}
      &\sembr{\jtypef{\Phi}{(\sigma,M) x}{\tau}}\\
      &= \sembr{\jtypef{\Phi}{M}{\tau}}\\
      &= \pi^{\Psi,x}_x \circ \langle \sembr{\jcmf{\sigma}{\Phi}{\Psi}}, x : \sembr{\jtypef{\Phi}{M}{\tau}} \rangle\\
      &= \pi^{\Psi,x}_x \circ \sembr{\jcmf{\sigma, M}{\Phi}{\Psi, x : \tau}}\\
      &= \sembr{\jtypef{\Psi,x : \tau}{x}{\tau}} \circ \sembr{\jcmf{\sigma, M}{\Phi}{\Psi, x : \tau}}.
    \end{align*}

  \item[Case \rn{F-Fix}.]
    \[
      \infer[\rn{F-Fix}]{
        \jtypef{\Psi}{\tFix{x}{M}}{\tau}
      }{
        \jtypef{\Psi,x:\tau}{M}{\tau}
      }
    \]
    Recall \cref{eq:28}:
    \begin{align*}
      \sembr{\jtypef{\Psi}{\tFix{x}{M}}{\tau}}u &= \sfix{\sembr{\jtypef{\Psi,x:\tau}{M}{\tau}}}u
    \end{align*}
    The result follows from a tweak of the standard proof and \cref{prop:fscd:2}.
    Let $\jcmf{\sigma}{\Phi}{\Psi}$ be arbitrary.
    By \cref{lemma:fscd:13}, we can weaken it to $\jcmf{\sigma,x}{\Phi,x : \tau}{\Psi, x : \tau}$.
    Let $\eta : \moabc\left( \cdot \times \sembr{x : \tau} ,  \sembr{\tau} \right) \nto \moabc\left( \cdot ,  \sembr{\tau} \right)$ be the natural interpretation given by \cref{prop:fscd:2}.
    Then
    \begin{align*}
      &\sembr{\jtypef{\Phi}{\sigma(\tFix{x}{M})}{\tau}}\\
      &= \sembr{\jtypef{\Phi}{\tFix{x}{(\sigma,x)M}}{\tau}}\\
      &= \eta_{\sembr{\Phi}}\left(\sembr{\jtypef{\Phi,x:\tau}{(\sigma,x)M}{\tau}}\right)\\
      \shortintertext{which by the induction hypothesis:}
      &= \eta_{\sembr{\Phi}}\left(\sembr{\jtypef{\Psi,x:\tau}{M}{\tau}} \circ \sembr{\jcmf{\sigma,x}{\Phi, x : \tau}{\Psi,x:\tau}}\right)\\
      \shortintertext{which by \cref{eq:fscd:33}:}
      &= \eta_{\sembr{\Phi}}\left(\sembr{\jtypef{\Psi,x:\tau}{M}{\tau}} \circ \langle \sembr{\jcmf{\sigma}{\Phi, x : \tau}{\Psi}}, x : \sembr{\jtypef{\Phi,x : \tau}{x}{\tau}} \rangle\right)\\
      \shortintertext{which by \cref{eq:39}:}
      &= \eta_{\sembr{\Phi}}\left(\sembr{\jtypef{\Psi,x:\tau}{M}{\tau}} \circ \langle \sembr{\jcmf{\sigma}{\Phi, x : \tau}{\Psi}}, x : \pi^{\Phi,x}_x \rangle\right)\\
      \shortintertext{which by \cref{lemma:fscd:13}:}
      &= \eta_{\sembr{\Phi}}\left(\sembr{\jtypef{\Psi,x:\tau}{M}{\tau}} \circ \langle \sembr{\jcmf{\sigma}{\Phi}{\Psi}} \circ \pi^{\Phi, x}_\Phi, x : \pi^{\Phi,x}_x \rangle\right)\\
      \shortintertext{which by a property of products and projections:}
      &= \eta_{\sembr{\Phi}}\left(\sembr{\jtypef{\Psi,x:\tau}{M}{\tau}} \circ \left(\sembr{\jcmf{\sigma}{\Phi}{\Psi}} \times \sembr{x : \tau}\right)\right)\\
      &= \left(\eta_{\sembr{\Phi}} \circ \moabc\left( \sembr{\jcmf{\sigma}{\Phi}{\Psi}} \times \sembr{x : \tau} ,  \sembr{\tau} \right)\right) \left(\sembr{\jtypef{\Psi,x:\tau}{M}{\tau}}\right)\\
      \shortintertext{which by naturality:}
      &= \left(\moabc\left( \sembr{\jcmf{\sigma}{\Phi}{\Psi}} ,  \sembr{\tau} \right) \circ \eta_{\sembr{\Psi}}\right)\left(\sembr{\jtypef{\Psi,x:\tau}{M}{\tau}}\right)\\
      \shortintertext{which by definition of natural interpretation:}
      &= \moabc\left( \sembr{\jcmf{\sigma}{\Phi}{\Psi}} ,  \sembr{\tau} \right)\left(\sembr{\jtypef{\Psi}{\tFix{x}{M}}{\tau}}\right)\\
      &= \sembr{\jtypef{\Psi}{\tFix{x}{M}}{\tau}} \circ \sembr{\jcmf{\sigma}{\Phi}{\Psi}}.
    \end{align*}

  \item[Case \rn{F-Fun}.]
    \[
      \infer[\rn{F-Fun}]{
        \jtypef{\Psi}{\lambda x : \tau.M}{\tau \to \sigma}
      }{
        \jtypef{\Psi, x:\tau}{M}{\sigma}
      }
    \]
    This case is analogous to the case \rn{F-Fix}.

  \item[Case \rn{F-App}.]
    \[
      \infer[\rn{F-App}]{
        \jtypef{\Psi}{MN}{\sigma}
      }{
        \jtypef{\Psi}{M}{\tau \to \sigma}
        &
        \jtypef{\Psi}{N}{\tau}
      }
    \]
    By the standard proof and \cref{prop:fscd:2}.

  \item[Case \rn{Fwd}.]
    \[
      \infer[\rn{Fwd}]{
        \jtypem{\Psi}{a:A}{\tFwd{b}{a}}{b}{A}
      }{}
    \]
    By the standard proof and \cref{prop:fscd:2}.

  \item[Case \rn{Cut}.]
    \[
      \infer[\rn{Cut}]{
        \jtypem{\Psi}{\Delta_1,\Delta_2}{\tCut{a}{P}{Q}}{c}{C}
      }{
        \jtypem{\Psi}{\Delta_1}{P}{a}{A}
        &
        \jtypem{\Psi}{a:A,\Delta_2}{Q}{c}{C}
      }
    \]
    By the standard proof and \cref{prop:fscd:2}.

  \item[Case \rn{$\Tu R$}.]
    \[
      \infer[\rn{$\Tu R$}]{
        \jtypem{\Psi}{\cdot}{\tClose a}{a}{\Tu}
      }{}
    \]
    By the standard proof and \cref{prop:fscd:2}.

  \item[Case \rn{$\Tu L$}.]
    \[
      \infer[\rn{$\Tu L$}]{
        \jtypem{\Psi}{\Delta, a : \Tu}{\tWait{a}{P}}{c}{C}
      }{
        \jtypem{\Psi}{\Delta}{P}{c}{C}
      }
    \]
    By the standard proof and \cref{prop:fscd:2}.

  \item[Case \rn{$\Tds{} R$}.]
    \[
      \infer[\rn{$\Tds{} R$}]{
        \jtypem{\Psi}{\Delta}{\tSendS{a}{P}}{a}{\Tds A}
      }{
        \jtypem{\Psi}{\Delta}{P}{a}{A}
      }
    \]
    By the standard proof and \cref{prop:fscd:2}.

  \item[Case \rn{$\Tds{} L$}.]
    \[
      \infer[\rn{$\Tds{} L$}]{
        \jtypem{\Psi}{\Delta,a : \Tds A}{\tRecvS{a}{P}}{c}{C}
      }{
        \jtypem{\Psi}{\Delta,a : A}{P}{c}{C}
      }
    \]
    By the standard proof and \cref{prop:fscd:2}.

  \item[Case \rn{$\Tus R$}.]
    \[
      \infer[\rn{$\Tus R$}]{
        \jtypem{\Psi}{\Delta}{\tRecvS{a}{P}}{a}{\Tus{A}}
      }{
        \jtypem{\Psi}{\Delta}{P}{a}{}
      }
    \]
    By the standard proof and \cref{prop:fscd:2}.

  \item[Case \rn{$\Tus L$}.]
    \[
      \infer[\rn{$\Tus L$}]{
        \jtypem{\Psi}{\Delta,a : \Tus A}{\tSendS{a}{P}}{c}{C}
      }{
        \jtypem{\Psi}{\Delta,a : A}{P}{c}{C}
      }
    \]
    By the standard proof and \cref{prop:fscd:2}.

  \item[Case \rn{$\Tplus R_k$}.]
    \[
      \infer[\rn{$\Tplus R_k$}]{
        \jtypem{\Psi}{\Delta}{\tSendL{a}{k}{P}}{a}{{\Tplus\{l:A_l\}}_{l \in L}}
      }{
        \jtypem{\Psi}{\Delta}{P}{a}{A_k}\quad(k \in L)
      }
    \]
    By the standard proof and \cref{prop:fscd:2}.

  \item[Case \rn{$\Tplus L$}.]
    \[
      \infer[\rn{$\Tplus L$}]{
        \jtypem{\Psi}{\Delta,a:{\Tplus\{l : A_l\}}_{l \in L}}{\tCase{a}{\left\{l_l \Rightarrow P_l\right\}_{i\in I}}}{c}{C}
      }{
        \jtypem{\Psi}{\Delta,a:A_l}{P_l}{c}{C}\quad(\forall l \in L)
      }
    \]
    By the standard proof and \cref{prop:fscd:2}.

  \item[Case \rn{$\Tamp R$}.]
    \[
      \infer[\rn{$\Tamp R$}]{
        \jtypem{\Psi}{\Delta}{\tCase{a}{\left\{l \Rightarrow P_l\right\}_{l \in L}}}{a}{{\Tamp\{l :A_l \}}_{l \in L}}
      }{
        \jtypem{\Psi}{\Delta}{P_l}{a}{A_l}\quad(\forall l \in L)
      }
    \]
    By the standard proof and \cref{prop:fscd:2}.

  \item[Case \rn{$\Tamp L_k$}.]
    \[
      \infer[\rn{$\Tamp L_k$}]{
        \jtypem{\Psi}{\Delta,a:{\Tamp\{l : A_l\}}_{l \in L}}{\tSendL{a}{k}{P}}{c}{C}
      }{
        \jtypem{\Psi}{\Delta,a:A_k}{P}{c}{C}
        &
        (k \in L)
      }
    \]
    By the standard proof and \cref{prop:fscd:2}.

  \item[Case \rn{$\Tot R^*$}.]
    \[
      \infer[\rn{$\Tot R^*$}]{
        \jtypem{\Psi}{\Delta, b : B}{\tSendC{a}{b}{P}}{a}{B \Tot A}
      }{
        \jtypem{\Psi}{\Delta}{P}{a}{A}
      }
    \]
    By the standard proof and \cref{prop:fscd:2}.

  \item[Case \rn{$\Tot L$}.]
    \[
      \infer[\rn{$\Tot L$}]{
        \jtypem{\Psi}{\Delta, a : B \Tot A}{\tRecvC{b}{a}{P}}{c}{C}
      }{
        \jtypem{\Psi}{\Delta, a : A, b : B}{P}{c}{C}
      }
    \]
    By the standard proof and \cref{prop:fscd:2}.

  \item[Case \rn{${\Tlolly}R$}.]
    \[
      \infer[\rn{${\Tlolly}R$}]{
        \jtypem{\Psi}{\Delta}{\tRecvC{b}{a}{P}}{a}{B \Tlolly A}
      }{
        \jtypem{\Psi}{\Delta, b : B}{P}{a}{A}
      }
    \]
    By the standard proof and \cref{prop:fscd:2}.

  \item[Case \rn{${\Tlolly}L$}.]
    \[
      \infer[\rn{${\Tlolly}L$}]{
        \jtypem{\Psi}{\Delta, b : B, a : B \Tlolly A}{\tSendC{a}{bb}{P}}{c}{C}
      }{
        \jtypem{\Psi}{\Delta,a : A}{P}{c}{C}
      }
    \]
    By the standard proof and \cref{prop:fscd:2}.

  \item[Case \rn{$\Tand{}{} R$}.]
    \[
      \infer[\rn{$\Tand{}{} R$}]{
        \jtypem{\Psi}{\Delta}{\tSendV{a}{M}{P}}{a}{\Tand{\tau}{A}}
      }{
        \jtypef{\Psi}{M}{\tau}
        &
        \jtypem{\Psi}{\Delta}{P}{a}{A}
      }
    \]
    By the standard proof and \cref{prop:fscd:2}.

  \item[Case \rn{$\Tand{}{} L$}.]
    \[
      \infer[\rn{$\Tand{}{} L$}]{
        \jtypem{\Psi}{\Delta, a:\Tand{\tau}{A}}{\tRecvV{x}{a}{P}}{c}{C}
      }{
        \jtypem{\Psi,x:\tau}{\Delta, a:A}{P}{c}{C}
      }
    \]
    This case is analogous to the case \rn{F-Fix}.

  \item[Case \rn{${\Timp{}{}} R$}.]
    \[
      \infer[\rn{${\Timp{}{}} R$}]{
        \jtypem{\Psi}{\Delta}{\tRecvV{x}{a}{P}}{a}{\Timp{\tau}{A}}
      }{
        \jtypem{\Psi,x:\tau}{\Delta}{P}{a}{A}
      }
    \]
    This case is analogous to the case \rn{F-Fix}.

  \item[Case \rn{${\Timp{}{}} L$}.]
    \[
      \infer[\rn{${\Timp{}{}} L$}]{
        \jtypem{\Psi}{\Delta,a : \Timp{\tau}{A}}{\tSendV{a}{M}{P}}{c}{C}
      }{
        \jtypef{\Psi}{M}{\tau}
        &
        \jtypem{\Psi}{\Delta, a : A}{P}{c}{C}
      }
    \]
    By the standard proof and \cref{prop:fscd:2}.

  \item[Case \rn{$\rho^+R$}.]
    \[
      \infer[\rn{$\rho^+R$}]{
        \jtypem{\Psi}{\Delta}{\tSendU{a}{P}}{a}{\Trec{\alpha}{A}}
      }{
        \jtypem{\Psi}{\Delta}{P}{a}{[\Trec{\alpha}{A}/\alpha]A}
        &
        \cdot \vdash \jisst[+]{\Trec{\alpha}{A}}
      }
    \]
    By the standard proof and \cref{prop:fscd:2}.

  \item[Case \rn{$\rho^+L$}.]
    \[
      \infer[\rn{$\rho^+L$}]{
        \jtypem{\Psi}{\Delta, a : \Trec{\alpha}{A}}{\tRecvU{a}{P}}{c}{C}
      }{
        \jtypem{\Psi}{\Delta, a : [\Trec{\alpha}{A}/\alpha]A}{P}{c}{C}
        &
        \cdot \vdash \jisst[+]{\Trec{\alpha}{A}}
      }
    \]
    By the standard proof and \cref{prop:fscd:2}.

  \item[Case \rn{$\rho^-R$}.]
    \[
      \infer[\rn{$\rho^-R$}]{
        \jtypem{\Psi}{\Delta}{\tRecvU{a}{P}}{a}{\Trec{\alpha}{A}}
      }{
        \jtypem{\Psi}{\Delta}{P}{a}{[\Trec{\alpha}{A}/\alpha]A}
        &
        \cdot \vdash \jisst[-]{\Trec{\alpha}{A}}
      }
    \]
    By the standard proof and \cref{prop:fscd:2}.

  \item[Case \rn{$\rho^-L$}.]
    \[
      \infer[\rn{$\rho^-L$}]{
        \jtypem{\Psi}{\Delta, a : \Trec{\alpha}{A}}{\tSendU{a}{P}}{c}{C}
      }{
        \jtypem{\Psi}{\Delta, a : [\Trec{\alpha}{A}/\alpha]A}{P}{c}{C}
        &
        \cdot \vdash \jisst[-]{\Trec{\alpha}{A}}
      }
    \]
    By the standard proof and \cref{prop:fscd:2}.

  \item[Case \rn{E-\{\}}.]
    \[
      \infer[\rn{E-\{\}}]{
        \jtypem{\Psi}{\overline{a_i:A_i}}{\tProc{a}{M}{\overline a_i}}{a}{A}
      }{
        \jtypef{\Psi}{M}{\Tproc{a:A}{\overline{a_i:A_i}}}
      }
    \]
    By the standard proof and \cref{prop:fscd:2}.\qedhere
  \end{description}
\end{proof}

\section{Semantic Results for Type Formers}
\label{sec:semant-results-type}

We show various semantic results for type interpretations.
These results have a similar structure as \cref{sec:semant-results-terms}.
We ignore size issues: they can be addressed using a hierarchy of universes~\cite[\S~3]{schubert_1972:_categ}.

Recall that we interpret judgments $\jstype{\Xi}{A}$ as locally continuous functors from $\sembr{\Xi}$ to $\moabcs$.
These categories and functors are respectively objects and morphisms in the category $\CFP$ of $\mb{O}$-categories that support canonical fixed points ($\CFP$ is also known as $\mb{Kind}$~\cite[\S~7.3.2]{fiore_1994:_axiom_domain_theor}).
Consider a type-forming rule
\[
  \infer{
    \jstype{\Xi}{F_{\mathcal{J}_1,\dotsc,\mathcal{J}_m}(A_1,\dotsc,A_n)}
  }{
    \jstype{\Xi}{A_1}
    &
    \cdots
    &
    \jstype{\Xi}{A_n}
    &
    \mathcal{J}_1
    &
    \cdots
    &
    \mathcal{J}_m
  }
\]
Assume the polarized aspect $\sembr{\jstype{\Xi}{F_{\mathcal{J}_1,\dotsc,\mathcal{J}_m}(A_1,\dotsc,A_n)}}^p$ is given by
\begin{equation}
  \label[intn]{eq:fscd:35}
  \begin{aligned}
    &\sembr{\jstype{\Xi}{F_{\mathcal{J}_1,\dotsc,\mathcal{J}_m}(A_1,\dotsc,A_n)}}^p\\
    &= \sembr{F_{\mathcal{J}_1,\dotsc,\mathcal{J}_m}}_{\sembr{\Xi}}\left(\sembr{\jstype{\Xi}{A_1}}^p, \dotsc, \sembr{\jstype{\Xi}{A_n}}^p\right).
  \end{aligned}
\end{equation}
where $\sembr{F_{\mathcal{J}_1,\dotsc,\mathcal{J}_m}}^p$ is a family of morphisms
\begin{equation}
  \label{eq:fscd:36}
  \sembr{F_{\mathcal{J}_1,\dotsc,\mathcal{J}_m}}_{\sembr{\Xi}} : \left( \prod_{i = 1}^n \CFP\left( \sembr{\Xi}, \moabcs \right) \right) \to \CFP\left( \sembr{\Xi} , \moabcs \right).
\end{equation}
We say that \cref{eq:fscd:35} is \defin{natural in its environment} if the family $\sembr{F_{\mathcal{J}_1,\dotsc,\mathcal{J}_m}}^p_{\sembr{\Xi}}$ is natural in $\sembr{\Xi}$, \ie, if for all functors $\sigma : \sembr{\Xi} \to \sembr{\Phi}$, the following commutes:
\[
  \begin{tikzcd}[column sep=7em]
    \prod_{i = 1}^n \CFP\left( \sembr{\Phi}, \moabcs \right) \ar[r, "{\sembr{F_{\mathcal{J}_1,\dotsc,\mathcal{J}_m}}_{\sembr{\Phi}}}"] \ar[d, swap, "{\prod_{i = 1}^n \CFP\left( \sigma, \moabcs \right)}"] & \CFP\left( \sembr{\Phi}, \moabcs \right) \ar[d, "{\CFP\left( \sigma, \moabcs \right)}"]\\
  \prod_{i = 1}^n \CFP\left( \sembr{\Xi}, \moabcs \right) \ar[r, "{\sembr{F_{\mathcal{J}_1,\dotsc,\mathcal{J}_m}}_{\sembr{\Xi}}}"] & \CFP\left( \sembr{\Xi}, \moabcs \right)
\end{tikzcd}
\]
Concretely, this means that for all $n$-tuples of locally continuous functors $(G_i : \sembr{\Phi} \to \moabcs)_{1 \leq i \leq n}$,
\[
  \sembr{F_{\mathcal{J}_1,\dotsc,\mathcal{J}_m}}_{\sembr{\Xi}}\left((G_i\sigma)_{1 \leq i \leq n}\right) = \sembr{F_{\mathcal{J}_1,\dotsc,\mathcal{J}_m}}_{\sembr{\Phi}}\left((G_i)_{1 \leq i \leq n}\right) \sigma.
\]
In this case, we call $\sembr{F_{\mathcal{J}_1,\dotsc,\mathcal{J}_l}}^p$ a \defin{natural interpretation} of the rule.

\begin{proposition}
  \label{prop:fscd:4}
  If $\jstype{\Xi}{A}$, then $\sembr{\jstype{\Xi}{A}}^-$ and $\sembr{\jstype{\Xi}{A}}^+$ are natural in their environments.
\end{proposition}

\begin{proof}
  By case analysis on the last rule in the derivation of $\jstype{\Xi}{A}$.

  \begin{description}[listparindent=\parindent]
  \item[Case \rn{C$\Tu$}.]
    \[
      \infer[\rn{C$\Tu$}]{\jstype[+]{\Xi}{\Tu}}{}
    \]
    Recall \cref{eq:1502,eq:1501}:
    \begin{align*}
      \sembr{\jstype{\Xi}{\Tu}}^- &= \lambda \xi.\top_{\moabcs}\\
      \sembr{\jstype{\Xi}{\Tu}}^+ &= \lambda \xi.\{\ast\}_\bot.
    \end{align*}
    We must show that there exist natural transformations
    \begin{align*}
      \eta^- &: \top_{\mb{Set}} \nto \CFP\left( {-} , \moabcs \right)\\
      \eta^+ &: \top_{\mb{Set}} \nto \CFP\left( {-} , \moabcs \right)
    \end{align*}
    that are the respective natural interpretations.
    Given a category $\mb{C}$, let
    \[
      \Delta_{\mb{C}} : \moabcs \to [\mb{C} \to \moabcs]
    \]
    be the diagonal functor.
    Take
    \begin{align*}
      \eta_{\mb{C}}^-(\ast) &= \Delta_{\mb{C}}(\top_\moabcs) : \mb{C} \to \moabcs\\
      \eta_{\mb{C}}^+(\ast) &= \Delta_{\mb{C}}(\{\ast\}_\bot) : \mb{C} \to \moabcs.
    \end{align*}
    These are obviously natural.
    When $\mb{C} = \sembr{\Xi}$, we recover:
    \begin{align*}
      \eta_{\sembr{\Xi}}^-(\ast) &= \sembr{\jstype{\Xi}{\Tu}}^-\\
      \eta_{\sembr{\Xi}}^+(\ast) &= \sembr{\jstype{\Xi}{\Tu}}^+.
    \end{align*}

  \item[Case \rn{CVar}.]
    \[
      \infer[\rn{CVar}]{\Xi,\jisst[p]{\alpha}\vdash\jisst[p]{\alpha}}{}
    \]
    Recall \cref{eq:fscd:7}:
    \begin{align*}
      \sembr{\Xi,\jisst[p]{\alpha}\vdash\jisst[p]{\alpha}}^q &= \pi^{\Xi,\alpha}_\alpha \quad (q \in \{{-},{+}\})
    \end{align*}
    We must show that there exist natural transformations
    \begin{align*}
      \eta^- &: \top_{\mb{Set}} \nto \CFP\left( {-} \times \sembr{\alpha}, \moabcs \right)\\
      \eta^+ &: \top_{\mb{Set}} \nto \CFP\left( {-} \times \sembr{\alpha}, \moabcs \right)
    \end{align*}
    that are the respective natural interpretations.
    Take
    \begin{align*}
      \eta_{\mb{C}}^-(\ast) &= \pi_\alpha : \mb{C} \times \sembr{\alpha} \to \moabcs\\
      \eta_{\mb{C}}^+(\ast) &= \pi_\alpha : \mb{C} \times \sembr{\alpha} \to \moabcs.
    \end{align*}
    These are are obviously natural.
    When $\mb{C} = \sembr{\Xi}$, we recover:
    \begin{align*}
      \eta_{\sembr{\Xi}}^-(\ast) &= \sembr{\jstype{\Xi, \alpha}{\alpha}}^-\\
      \eta_{\sembr{\Xi}}^+(\ast) &= \sembr{\jstype{\Xi, \alpha}{\alpha}}^+.
    \end{align*}

  \item[Case \rn{C$\rho$}.]
    \[
      \infer[\rn{C$\rho$}]{
        \Xi \vdash \jisst[p]{\Trec{\alpha}{A}}
      }{
        \Xi, \jisst[p]{\alpha} \vdash \jisst[p]{A}
      }
    \]
    Recall \cref{eq:20051}:
    \begin{align*}
      \sembr{\jstype{\Xi}{\Trec{\alpha}{A}}}^p &= \sfix{\left(\sembr{\Xi,\jisst{\alpha}\vdash \jisst{A}}^p\right)}\quad (p \in \{{-},{+}\})
    \end{align*}
    We must show that there exist natural transformations
    \begin{align*}
      \eta^p &: \CFP\left( {-} \times \sembr{\alpha}, \moabcs \right) \nto \CFP\left( {-}, \moabcs \right) \quad (p \in \{{-},{+}\})
    \end{align*}
    that are the respective natural interpretations.
    Take
    \[
      \eta_{\mb{C}}^p(G) = \sfix{G} : \mb{C} \to \moabcs
    \]
    To show the naturality of $\eta^p$, we must show that for all functors $\sigma : \mb{D} \to \mb{C}$,
    \[
      \sfix{(G \circ (\sigma \times \ms{id}))} = \sfix{G} \circ \sigma : \mb{D} \to \moabcs.
    \]
    This is exactly the parameter identify~\cite{bloom_esik_1996:_fixed_point_operat} satisfied by the dagger $\sfix{(\cdot)}$, so we are done.

  \item[Case \rn{C$\Tds{}$}.]
    \[
      \infer[\rn{C$\Tds{}$}]{
        \jstype[+]{\Xi}{\Tds A}
      }{
        \jstype[-]{\Xi}{A}
      }
    \]
    Recall \cref{eq:1506,eq:1507}:
    \begin{align*}
      \sembr{\jstype{\Xi}{\Tds A}}^- &= \sembr{\jstype{\Xi}{A}}^-\\
      \sembr{\jstype{\Xi}{\Tds A}}^+ &= \sembr{\jstype{\Xi}{A}}^+_\bot
    \end{align*}
    We must show that there exist natural transformations
    \begin{align*}
      \eta^- &: \CFP\left( {-}, \moabcs \right) \nto \CFP\left( {-}, \moabcs \right)\\
      \eta^+ &: \CFP\left( {-}, \moabcs \right) \nto \CFP\left( {-}, \moabcs \right)
    \end{align*}
    that are the respective natural interpretations.
    Recall that we write $F_\bot$ for the composition ${({-})_\bot}F$ and $\sigma_\bot$ for the natural transformation ${({-})}_\bot\sigma$.
    Take
    \begin{align*}
      \eta_{\mb{C}}^-(F) &= F : \mb{C} \to \moabcs\\
      \eta_{\mb{C}}^+(F) &= F_\bot : \mb{C} \to \moabcs.
    \end{align*}
    To show the naturality of $\eta^-$ and $\eta^+$, we must show for all functors $F : \mb{C} \to \moabcs$ and $\sigma : \mb{D} \to \mb{C}$ that
    \begin{align*}
      \eta_{\mb{D}}^-(F\sigma) &= \eta_{\mb{C}}^-(F)\sigma\\
      \eta_{\mb{D}}^+(F\sigma) &= \eta_{\mb{C}}^+(F)\sigma.
    \end{align*}
    In the negative case,
    \[
      \eta_{\mb{D}}^-(F\sigma) = F\sigma = \eta_{\mb{C}}^-(F)\sigma.
    \]
    In the positive case, we have by associativity of composition:
    \[
      \eta_{\mb{D}}^+(F\sigma) = ({-})_\bot (F\sigma) =  (({-})_\bot F)\sigma = \eta_{\mb{C}}^+(F)\sigma.
    \]

  \item[Case \rn{C$\Tus{}$}.]
    \[
      \infer[\rn{C$\Tus{}$}]{
        \jstype[-]{\Xi}{\Tus A}
      }{
        \jstype[+]{\Xi}{A}
      }
    \]
    This case is analogous to the case \rn{C$\Tds{}$}.

  \item[Case \rn{C$\Tplus$}.]
    \[
      \infer[\rn{C$\Tplus$}]{
        \jstype[+]{\Xi}{{\Tplus\{l : A_l\}}_{l \in L}}
      }{
        \jstype[+]{\Xi}{A_l}\quad(\forall l \in L)
      }
    \]
    Recall \cref{eq:202020,eq:222222}:
    \begin{align*}
      \sembr{\jstype{\Xi}{\Tplus\{l:A_l\}_{l \in L}}}^- &= \prod_{l \in L} \sembr{\jstype{\Xi}{A_l}}^-\\
      \sembr{\jstype{\Xi}{\Tplus\{l:A_l\}_{l \in L}}}^+ &= \bigoplus_{l \in L} \sembr{\jstype{\Xi}{A_l}}^+_\bot
    \end{align*}
    We must show that there exist natural transformations
    \begin{align*}
      \eta^- &: \left(\prod_{l \in L} \CFP\left( {-}, \moabcs \right)\right) \nto \CFP\left( {-}, \moabcs \right)\\
      \eta^+ &: \left(\prod_{l \in L} \CFP\left( {-}, \moabcs \right)\right) \nto \CFP\left( {-}, \moabcs \right)
    \end{align*}
    that are the respective natural interpretations.
    Take
    \begin{align*}
      \eta^-_{\mb{C}}\left(\left(F_l\right)_{l \in L}\right) &= \prod_{l \in L} F_l : \mb{C} \to \moabcs\\
      \eta^+_{\mb{C}}\left(\left(F_l\right)_{l \in L}\right) &= \bigoplus_{l \in L} ({-})_\bot F_l : \mb{C} \to \moabcs.
    \end{align*}
    These are both easily seen to be natural.

  \item[Case \rn{C$\Tamp$}.]
    \[
      \infer[\rn{C$\Tamp$}]{
        \jstype[-]{\Xi}{{\Tamp\{l :A_l \}}_{l \in L}}
      }{
        \jstype[-]{\Xi}{A_l}\quad(\forall l \in L)
      }
    \]
    This case is analogous to the case \rn{C$\Tplus$}.

  \item[Case \rn{C$\Tot$}.]
    \[
      \infer[\rn{C$\Tot$}]{
        \jstype[+]{\Xi}{A \Tot B}
      }{
        \jstype[+]{\Xi}{A}
        &
        \jstype[+]{\Xi}{B}
      }
    \]
    Recall \cref{eq:13,eq:14}:
    \begin{align*}
      \sembr{\jstype{\Xi}{A \Tot B}}^- &= \sembr{\jstype{\Xi}{A}}^- \times \sembr{\jstype{\Xi}{B}}^-\\
      \sembr{\jstype{\Xi}{A \Tot B}}^+ &= \left(\sembr{\jstype{\Xi}{A}}^+ \times \sembr{\jstype{\Xi}{B}}^+\right)_\bot
    \end{align*}
    We must show that there exist natural transformations
    \begin{align*}
      \eta^- &: \CFP\left({-}, \moabcs \right) \times \CFP\left( {-}, \moabcs \right) \nto \CFP\left( {-}, \moabcs \right)\\
      \eta^+ &: \CFP\left({-}, \moabcs \right) \times \CFP\left( {-}, \moabcs \right) \nto \CFP\left( {-}, \moabcs \right)
    \end{align*}
    that are the respective natural interpretations.
    Take
    \begin{align*}
      \eta_{\mb{C}}^-(A, B) &= A \times B : \mb{C} \to \moabcs\\
      \eta_{\mb{C}}^+(A, B) &= (A \times B)_\bot : \mb{C} \to \moabcs.
    \end{align*}
    These are easily seen to be natural.

  \item[Case \rn{C$\Tlolly$}.]
    \[
      \infer[\rn{C$\Tlolly$}]{
        \jstype[-]{\Xi}{A \Tlolly B}
      }{
        \jstype[+]{\Xi}{A}
        &
        \jstype[-]{\Xi}{B}
      }
    \]
    This case is analogous to the case \rn{C$\Tot$}.

  \item[Case \rn{C$\Tand{}{}$}.]
    \[
      \infer[\rn{C$\Tand{}{}$}]{
        \jstype[+]{\Xi}{\Tand{\tau}{A}}
      }{
        \jisft{\tau}
        &
        \jstype[+]{\Xi}{A}
      }
    \]
    Recall \cref{eq:15104,eq:15106}:
    \begin{align*}
      \sembr{\jstype{\Xi}{\Tand{\tau}{A}}}^- &= \sembr{\Xi\vdash \jisst{A}}^-\\
      \sembr{\jstype{\Xi}{\Tand{\tau}{A}}}^+ &= \left(\sembr{\tau}\times\sembr{\Xi\vdash \jisst{A}}^+\right)_\bot
    \end{align*}
    We must show that there exist natural transformations
    \begin{align*}
      \eta^- &: \Cell{\moabcs}\left({-}, \moabcs \right) \nto \Cell{\moabcs} \left( {-}, \moabcs \right)\\
      \eta^+ &: \Cell{\moabcs}\left({-}, \moabcs \right) \nto \Cell{\moabcs} \left( {-}, \moabcs \right)
    \end{align*}
    that are the respective natural interpretations.
    Take
    \begin{align*}
      \eta_{\mb{C}}^-(F) &= F : \mb{C} \to \moabcs\\
      \eta_{\mb{C}}^+(F) &= (\sembr{\tau} \times F)_\bot : \mb{C} \to \moabcs\\
    \end{align*}
    These are easily seen to be natural.

  \item[Case \rn{C$\Timp{}{}$}.]
    \[
      \infer[\rn{C$\Timp{}{}$}]{
        \jstype[-]{\Xi}{\Timp{\tau}{A}}
      }{
        \jisft{\tau}
        &
        \jstype[-]{\Xi}{A}
      }
    \]
    This case is analogous to the case \rn{C$\Tand{}{}$}.\qedhere
  \end{description}
\end{proof}

\subsection{Semantic Weakening}
\label{sec:semantic-weakening-types}

We show that weakening is semantically well-behaved, \ie, that the semantic clauses are coherent~\cite[p.~218]{tennent_1995:_denot_seman}.

\begin{proposition}[Coherence]
  \label{prop:fscd:8}
  Let $\Theta,\Xi$ be a context of type variables.
  If $\jstype[p]{\Xi}{A}$, then the following diagram commutes for $q \in \{{-},{+}\}$:
  \begin{equation}
    \label[diagram]{eq:114}
    \begin{tikzcd}[column sep=7em]
      \sembr{\Theta,\Xi} \ar[dr, "{\sembr{\jstype[p]{\Theta,\Xi}{A}}^q}"] \ar[d, swap, "\pi^{\Theta,\Xi}_\Xi"] &\\
      \sembr{\Xi} \ar[r, swap, "{\sembr{\jstype[p]{\Xi}{{A}}}^q}"] & \moabcs
    \end{tikzcd}
  \end{equation}
\end{proposition}

\begin{proof}
  By induction on the derivation of $\jstype[p]{\Xi}{A}$.

  Except where otherwise stated, each rule uses the same proof outline.
  We refer to it below as the ``standard proof''.
  Consider a type-forming rule
  \[
    \infer{
      \jstype{\Xi}{F_{\mathcal{J}_1,\dotsc,\mathcal{J}_m}(A_1,\dotsc,A_n)}
    }{
      \jstype{\Xi}{A_1}
      &
      \cdots
      &
      \jstype{\Xi}{A_n}
      &
      \mathcal{J}_1
      &
      \cdots
      &
      \mathcal{J}_m
    }
  \]
  Assume its interpretation is given by
  \[
    \begin{aligned}
      &\sembr{\jstype{\Xi}{F_{\mathcal{J}_1,\dotsc,\mathcal{J}_m}(A_1,\dotsc,A_n)}}^p\\
      &= \sembr{F_{\mathcal{J}_1,\dotsc,\mathcal{J}_m}}_{\sembr{\Xi}}\left(\sembr{\jstype{\Xi}{A_1}}^p, \dotsc, \sembr{\jstype{\Xi}{A_n}}^p\right).
    \end{aligned}
  \]
  where $\sembr{F_{\mathcal{J}_1,\dotsc,\mathcal{J}_m}}$ is a natural interpretation
  \[
    \sembr{F_{\mathcal{J}_1,\dotsc,\mathcal{J}_m}}_{\sembr{\Xi}} : \left( \prod_{i = 1}^n \CFP\left( \sembr{\Xi}, \moabcs \right) \right) \to \CFP\left( \sembr{\Xi} , \moabcs \right).
  \]
  Given any other context of type variables $\Theta$ disjoint from $\Xi$, we would like to show that
  \begin{align*}
    &\sembr{\jstype{\Theta,\Xi}{F_{\mathcal{J}_1,\dotsc,\mathcal{J}_m}(A_1,\dotsc,A_n)}}^p\\
    &= \sembr{\jstype{\Xi}{F_{\mathcal{J}_1,\dotsc,\mathcal{J}_m}(A_1,\dotsc,A_n)}}^p \circ \pi^{\Theta,\Xi}_\Xi.
  \end{align*}
  By the induction hypothesis, we have for all $1 \leq i \leq n$,
  \[
    \sembr{\jstype{\Theta,\Xi}{A_i}}^p = \sembr{\jstype{\Xi}{A_i}}^p \circ \pi^{\Theta,\Xi}_\Xi.
  \]
  Using these facts we get:
  \begin{align*}
    &\sembr{\jstype{\Theta,\Xi}{F_{\mathcal{J}_1,\dotsc,\mathcal{J}_m}(A_1,\dotsc,A_n)}}^p\\
    &= \sembr{F_{\mathcal{J}_1,\dotsc,\mathcal{J}_m}}_{\sembr{\Theta,\Xi}}\left(\sembr{\jstype{\Theta,\Xi}{A_1}}^p, \dotsc, \sembr{\jstype{\Theta,\Xi}{A_n}}^p\right)\\
    \shortintertext{which by the induction hypothesis,}
    &= \sembr{F_{\mathcal{J}_1,\dotsc,\mathcal{J}_m}}_{\sembr{\Theta,\Xi}}\left(
      \sembr{\jstype{\Xi}{A_1}}^p \circ \pi^{\Theta,\Xi}_\Xi,
      \dotsc,
      \sembr{\jstype{\Xi}{A_n}}^p \circ \pi^{\Theta,\Xi}_\Xi\right)\\
    \shortintertext{which by naturality of $\sembr{F_{\mathcal{J}_1,\dotsc,\mathcal{J}_m}}$,}
    &= \sembr{F_{\mathcal{J}_1,\dotsc,\mathcal{J}_m}}_{\sembr{\Xi}}\left(
      \sembr{\jstype{\Xi}{A_1}}^p,
      \dotsc,
      \sembr{\jstype{\Xi}{A_n}}^p\right) \circ \pi^{\Theta,\Xi}_\Xi\\
    &=\sembr{\jstype{\Xi}{F_{\mathcal{J}_1,\dotsc,\mathcal{J}_m}(A_1,\dotsc,A_n)}}^p  \circ \pi^{\Theta,\Xi}_\Xi.
  \end{align*}
  This is what we wanted to show.

  \begin{description}[listparindent=\parindent]
  \item[Case \rn{C$\Tu$}.]
    By the standard proof and \cref{prop:fscd:4}.

  \item[Case \rn{CVar}.]
    \[
      \infer[\rn{CVar}]{\Xi,\jisst[p]{\alpha}\vdash\jisst[p]{\alpha}}{}
    \]
    Recall \cref{eq:fscd:7}:
    \begin{align*}
      \sembr{\Xi,\jisst[p]{\alpha}\vdash\jisst[p]{\alpha}}^q &= \pi^{\Xi,\alpha}_\alpha \quad (q \in \{{-},{+}\})
    \end{align*}
    Then
    \[
      \sembr{\jstype{\Theta,\Xi,\jisst{\alpha}}{\alpha}}^q = \pi^{\Theta,\Xi,\alpha}_\alpha = \pi^{\Xi,\alpha}_\alpha\pi^{\Theta,\Xi,\alpha}_{\Xi,\alpha} = \sembr{\jstype{\Xi,\jisst{\alpha}}{\alpha}}\pi^{\Theta,\Xi,\alpha}_{\Xi,\alpha}.
    \]
    This is what we wanted to show.

  \item[Case \rn{C$\rho$}.]
    \[
      \infer[\rn{C$\rho$}]{
        \Xi \vdash \jisst[p]{\Trec{\alpha}{A}}
      }{
        \Xi, \jisst[p]{\alpha} \vdash \jisst[p]{A}
      }
    \]
    Recall \cref{eq:20051}:
    \begin{align*}
      \sembr{\Xi\vdash\jisst{\Trec{\alpha}{A}}}^p &= \sfix{\left(\sembr{\Xi,\jisst{\alpha}\vdash \jisst{A}}^p\right)}\quad (p \in \{{-},{+}\})
    \end{align*}
    The result follows from a tweak of the standard proof and \cref{prop:fscd:4}.
    Let $\eta^p : \CFP\left( {-} \times \sembr{\alpha}, \moabcs \right) \nto \CFP\left( {-}, \moabcs \right)$ be the natural interpretation given by \cref{prop:fscd:4}.
    Then
    \begin{align*}
      &\sembr{\jstype{\Theta,\Xi}{\Trec{\alpha}{A}}}^p\\
      &= \eta^p_{\sembr{\Theta,\Xi}}\left(\sembr{\jstype{\Theta,\Xi,\alpha}{A}}^p\right)\\
      \shortintertext{which by the induction hypothesis:}
      &= \eta^p_{\sembr{\Theta,\Xi}}\left(\sembr{\jstype{\Xi,\alpha}{A}}^p \pi^{\Theta,\Xi,\alpha}_{\Xi,\alpha}\right)\\
      \shortintertext{which by a property of products and projections:}
      &= \eta^p_{\sembr{\Theta,\Xi}}\left(\sembr{\jstype{\Xi,\alpha}{A}}^p \left(\pi^{\Theta,\Xi}_{\Xi} \times \sembr{\alpha}\right)\right)\\
      \shortintertext{which by naturality:}
      &= \eta^p_{\sembr{\Xi}}\left(\sembr{\jstype{\Xi,\alpha}{A}}^p\right)\pi^{\Theta,\Xi}_{\Xi}\\
      &= \sembr{\jstype{\Xi}{\Trec{\alpha}{A}}}^p\pi^{\Theta,\Xi}_{\Xi}.
    \end{align*}

  \item[Case \rn{C$\Tds{}$}.]
    By the standard proof and \cref{prop:fscd:4}.

  \item[Case \rn{C$\Tus{}$}.]
    By the standard proof and \cref{prop:fscd:4}.

  \item[Case \rn{C$\Tplus$}.]
    By the standard proof and \cref{prop:fscd:4}.

  \item[Case \rn{C$\Tamp$}.]
    By the standard proof and \cref{prop:fscd:4}.

  \item[Case \rn{C$\Tot$}.]
    By the standard proof and \cref{prop:fscd:4}.

  \item[Case \rn{C$\Tlolly$}.]
    By the standard proof and \cref{prop:fscd:4}.

  \item[Case \rn{C$\Tand{}{}$}.]
    By the standard proof and \cref{prop:fscd:4}.

  \item[Case \rn{C$\Timp{}{}$}.]
    By the standard proof and \cref{prop:fscd:4}.\qedhere
  \end{description}
\end{proof}

\subsection{Semantic Substitution}
\label{sec:semant-subst-types}

Our goal in this section is to show that substitution is given by composition.
\defin{Context morphisms} give us a semantic account of substitutions.
To this end, we introduce explicit typing rules for type substitutions.
We use the judgment $\jcms{\sigma}{\Xi}{\Theta}$ to mean that the substitution $\sigma$ is a morphism of type variable contexts from $\Xi$ to $\Theta$.
It is inductively defined by
\[
  \infer[\rn{S-S-Empty}]{
    \jcms{\cdot}{\Theta}{\cdot}
  }{\mathstrut}
  \quad
  \infer[\rn{S-S-T$^p$}]{
    \jcms{\sigma,A}{\Theta}{\Xi,\jisst[p]{\alpha}}
  }{
    \jcms{\sigma}{\Theta}{\Xi}
    &
    \jstype[p]{\Theta}{A}
  }
\]
These rules ensure that substitutions of types for type variables respects polarities, \ie, that we only substitute positive session types for a positive type variable and negative session types for negative type variables.
Given a morphism $\sigma = A_1,\dots,A_n$ satisfying $\jcms{\sigma}{\Xi}{\alpha_1,\dotsc,\alpha_n}$ with $n \geq 0$, and a session type $\alpha_1,\dotsc,\alpha_n \vdash B$, we write $\sigma B$ for the result of the simultaneous substitution $[A_1,\dotsc,A_n/\alpha_1,\dotsc,\alpha_n]B$.

\begin{proposition}[Syntactic Substitution of Session Types]
  \label{prop:fscd:5}
  Let $\jcms{\sigma}{\Theta}{\Xi}$ be arbitrary.
  If $\jstype[p]{\Xi}{A}$, then $\jstype[p]{\Theta}{\sigma A}$.
\end{proposition}

Context morphisms are subject to the same classes of interpretation as session types.
A context morphism $\jcms{\sigma}{\Theta}{\Xi}$ gives rise to polarized interpretations, which are locally continuous functors $\sembr{\Theta} \to \sembr{\Xi}$.
The polarized interpretations of \rn{S-S-Empty} and \rn{S-S-T$^p$} are respectively
\begin{align*}
  \sembr{\jcms{\cdot}{\Theta}{\cdot}}^q &= \top_{\sembr{\Theta}}\\
  \sembr{\jcms{\sigma,A}{\Theta}{\Xi,\jisst[p]{\alpha}}}^q &= \langle \sembr{\jcms{\sigma}{\Theta}{\Xi}}^q, \alpha : \sembr{\jstype[p]{\Theta}{A}}^q \rangle
\end{align*}
where $\top_{\sembr{\Theta}}$ is the constant functor from $\sembr{\Theta}$ to the terminal category $\top$ and $q \in \{{-},{+}\}$.

\begin{lemma}[Weakening of Context Morphisms]
  \label{lemma:fscd:14}
  Let $\jcms{\sigma}{\Theta}{\Xi}$ be arbitrary and $\Theta,\Omega$ a context.
  Then $\jcms{\sigma}{\Omega,\Theta}{\Xi}$ and where $p$ ranges over $\{{-},{+}\}$,
  \[
    \sembr{\jcms{\sigma}{\Omega,\Theta}{\Xi}}^p = \sembr{\jcms{\sigma}{\Theta}{\Xi}}^p\pi^{\Omega,\Theta}_\Theta.
  \]
\end{lemma}

\begin{proof}
  By induction on the derivation of $\jcms{\sigma}{\Theta}{\Xi}$.

  \begin{description}[listparindent=\parindent]
  \item[Case \rn{S-S-Empty}.]
    Then $\jcms{\sigma}{\Omega,\Theta}{\cdot}$ by \rn{S-S-Empty}.
    Moreover,
    \[
      \sembr{\jcms{\sigma}{\Omega,\Theta}{\cdot}}^p = \top_{\CFP} = \top_{\Theta}\pi^{\Omega,\Theta}_\Theta = \sembr{\jcms{\sigma}{\Theta}{\cdot}}^p\pi^{\Omega,\Theta}_\Theta.
    \]
  \item[Case \rn{S-S-T${}^q$}.]
    \[
      \infer[\rn{S-S-T$^q$}]{
        \jcms{\sigma,A}{\Theta}{\Xi,\jisst[q]{\alpha}}
      }{
        \jcms{\sigma}{\Theta}{\Xi}
        &
        \jstype[q]{\Theta}{A}
      }
    \]
    Then by the induction hypothesis, $\jcms{\sigma}{\Omega,\Theta}{\Xi}$ and
    \[
      \sembr{\jcms{\sigma}{\Omega,\Theta}{\Xi}}^p = \sembr{\jcms{\sigma}{\Theta}{\Xi}}^p \pi^{\Omega,\Theta}_\Theta.
    \]
    By weakening, $\jstype[q]{\Omega,\Theta}{A}$, and by \cref{prop:fscd:8}
    \[
      \sembr{\jstype[q]{\Omega,\Theta}{A}}^p = \sembr{\jstype[q]{\Theta}{A}}^p \pi^{\Omega,\Theta}_\Theta.
    \]
    By \rn{S-S-T$^q$}, $\jcms{\sigma,A}{\Omega,\Theta}{\Xi,\jisst[q]{\alpha}}$.
    By the interpretation of \rn{S-S-T$^q$},
    \begin{align*}
      &\sembr{\jcms{\sigma,A}{\Omega,\Theta}{\Xi,\jisst[q]{\alpha}}}^p\\
      &= \langle \sembr{\jcms{\sigma}{\Omega,\Theta}{\Xi}}^p , \alpha : \sembr{\jstype[q]{\Omega,\Theta}{A}}^p \rangle\\
      &= \langle \sembr{\jcms{\sigma}{\Theta}{\Xi}}^p \pi^{\Omega,\Theta}_\Theta, \alpha : \sembr{\jstype[q]{\Theta}{A}}^p \pi^{\Omega,\Theta}_\Theta \rangle\\
      &= \langle \sembr{\jcms{\sigma}{\Theta}{\Xi}}^p, \alpha : \sembr{\jstype[q]{\Theta}{A}}^p\rangle \pi^{\Omega,\Theta}_\Theta\\
      &= \sembr{\jcms{\sigma,A}{\Theta}{\Xi,\jisst[q]{\alpha}}}^p\pi^{\Omega,\Theta}_\Theta .
    \end{align*}
  \end{description}
  We conclude the result by induction.
\end{proof}

\begin{proposition}[Semantic Substitution of Session Types]\label{prop:fscd:7}
  Let $\jcms{\sigma}{\Theta}{\Xi}$ be arbitrary and let $p$ range over $\{{-},{+}\}$.
  If $\Xi \vdash \jisst[q]{A}$, then
  \[
    \sembr{\Theta \vdash \jisst[q]{\sigma A}}^p = \sembr{\Xi \vdash \jisst[q]{A}}^p \circ \sembr{\jcms{\sigma}{\Theta}{\Xi}}^p.
  \]
\end{proposition}

\begin{proof}
  By induction on the derivation of $\jstype{\Xi}{A}$.
  As in \cref{prop:extended:10}, we describe the standard case.
  Consider a type-forming rule
  \[
    \infer{
      \jstype{\Xi}{F_{\mathcal{J}_1,\dotsc,\mathcal{J}_m}(A_1,\dotsc,A_n)}
    }{
      \jstype{\Xi}{A_1}
      &
      \cdots
      &
      \jstype{\Xi}{A_n}
      &
      \mathcal{J}_1
      &
      \cdots
      &
      \mathcal{J}_m
    }
  \]
  Assume $\sembr{\jstype{\Xi}{F_{\mathcal{J}_1,\dotsc,\mathcal{J}_m}(A_1,\dotsc,A_n)}}^p$ is given by a natural interpretation
  \[
    \sembr{F_{\mathcal{J}_1,\dotsc,\mathcal{J}_m}}^p_{\sembr{\Xi}} : \left( \prod_{i = 1}^n \CFP\left( \sembr{\Xi} , \moabcs \right) \right) \to \CFP\left( \sembr{\Xi}, \moabcs \right).
  \]
  We need to show that
  \begin{align*}
    &\sembr{\jstype{\Theta}{\sigma\left(F_{\mathcal{J}_1,\dotsc,\mathcal{J}_m}(A_1,\dotsc,A_n)\right)}}^p\\
    &= \sembr{\jstype{\Xi}{F_{\mathcal{J}_1,\dotsc,\mathcal{J}_m}(A_1,\dotsc,A_n)}}^p \circ \sembr{\jcms{\sigma}{\Theta}{\Xi}}^p.
  \end{align*}
  By the definition of syntactic substitution, we know that
  \[
    \sigma\left(F_{\mathcal{J}_1,\dotsc,\mathcal{J}_m}(A_1,\dotsc,A_n)\right) = F_{\mathcal{J}_1,\dotsc,\mathcal{J}_m}(\sigma A_1,\dotsc,\sigma A_n).
  \]
  By the induction hypothesis, we know for $1 \leq i \leq n$ that
  \[
    \sembr{\jstype{\Theta}{\sigma A_i}}^p = \sembr{\jstype{\Xi}{A_i}}^p \circ \sembr{\jcms{\sigma}{\Theta}{\Xi}}^p.
  \]
  Using these facts, we get:
  \begin{align*}
    &\sembr{\jstype{\Theta}{\sigma\left(F_{\mathcal{J}_1,\dotsc,\mathcal{J}_m}(A_1,\dotsc,A_n)\right)}}^p\\
    &= \sembr{\jstype{\Theta}{F_{\mathcal{J}_1,\dotsc,\mathcal{J}_m}(\sigma A_1,\dotsc,\sigma A_n)}}^p\\
    &= \sembr{F_{\mathcal{J}_1,\dotsc,\mathcal{J}_m}}^p_{\sembr{\Theta}}\left(\left(\sembr{\jstype{\Theta}{\sigma A_i}}^p\right)_{1 \leq i \leq n}\right)\\
    \shortintertext{which by the induction hypothesis,}
    &= \sembr{F_{\mathcal{J}_1,\dotsc,\mathcal{J}_m}}^p_{\sembr{\Theta}}\left(\left( \sembr{\jstype{\Xi}{A_i}}^p \circ \sembr{\jcms{\sigma}{\Theta}{\Xi}}^p\right)_{1 \leq i \leq n}\right)\\
    \shortintertext{which by naturality of $\sembr{F_{\mathcal{J}_1,\dotsc,\mathcal{J}_m}}^p$,}
    &= \sembr{F_{\mathcal{J}_1,\dotsc,\mathcal{J}_m}}^p_{\sembr{\Xi}}\left(\left( \sembr{\jstype{\Xi}{A_i}}^p\right)_{1 \leq i \leq n}\right)  \circ \sembr{\jcms{\sigma}{\Theta}{\Xi}}^p\\
    &= \sembr{\jstype{\Xi}{F_{\mathcal{J}_1,\dotsc,\mathcal{J}_m}(A_1,\dotsc,A_n)}}^p \circ \sembr{\jcms{\sigma}{\Theta}{\Xi}}^p.
  \end{align*}

  \begin{description}
  \item[Case \rn{C$\Tu$}.]
    \[
      \infer[\rn{C$\Tu$}]{\jstype[+]{\Xi}{\Tu}}{}
    \]
    By the standard proof and \cref{prop:fscd:4}.

  \item[Case \rn{CVar}.]
    \[
      \infer[\rn{CVar}]{\Xi,\jisst[p]{\alpha}\vdash\jisst[p]{\alpha}}{}
    \]
    Recall \cref{eq:fscd:7}:
    \begin{align*}
      \sembr{\Xi,\jisst[p]{\alpha}\vdash\jisst[p]{\alpha}}^q &= \pi^{\Xi,\alpha}_\alpha \quad (q \in \{{-},{+}\})
    \end{align*}
    Let $\jcms{\sigma,A}{\Theta}{\Xi,\alpha}$ be arbitrary.
    \begin{align*}
      &\sembr{\jstype{\Theta}{(\sigma,A)\alpha}}^q\\
      &= \sembr{\jstype{\Theta}{A}}^q\\
      &= \ms{id}_{\pi^{\Xi,\alpha}_\alpha} \circ \langle \sembr{\jcms{\sigma}{\Theta}{\Xi}}^q, \alpha : \sembr{\jstype{\Theta}{A}}^q \rangle\\
      &= \sembr{\Xi,\jisst[p]{\alpha} \vdash \jisst[p]{\alpha}}^q \circ \sembr{\jcms{\sigma,A}{\Theta}{\Xi,\alpha}}^q.
    \end{align*}

  \item[Case \rn{C$\rho$}.]
    \[
      \infer[\rn{C$\rho$}]{
        \Xi \vdash \jisst[p]{\Trec{\alpha}{A}}
      }{
        \Xi, \jisst[p]{\alpha} \vdash \jisst[p]{A}
      }
    \]
    Recall \cref{eq:20051}:
    \begin{align*}
      \sembr{\Xi\vdash\jisst{\Trec{\alpha}{A}}}^p &= \sfix{\left(\sembr{\Xi,\jisst{\alpha}\vdash \jisst{A}}^p\right)}\quad (p \in \{{-},{+}\})
    \end{align*}
    The result follows from a tweak of the standard proof and \cref{prop:fscd:4}.
    Let $\jcms{\sigma}{\Theta}{\Xi}$ be arbitrary.
    By \cref{lemma:fscd:14}, we can weaken it to $\jcms{\sigma,\alpha}{\Theta,\alpha}{\Xi,\alpha}$.
    Let $\eta^p : \CFP\left( {-} \times \sembr{\alpha}, \moabcs \right) \nto \CFP\left( {-}, \moabcs \right)$ be the natural interpretation given by \cref{prop:fscd:4}.
    Then
    \begin{align*}
      &\sembr{\jstype{\Theta}{\sigma\left(\Trec{\alpha}{A}\right)}}^p\\
      &= \sembr{\jstype{\Theta}{\Trec{\alpha}{(\sigma,\alpha)A}}}^p\\
      &= \eta^p_{\sembr{\Theta}}\left(\sembr{\jstype{\Theta,\alpha}{(\sigma,\alpha)A}}^p\right)\\
      \shortintertext{which by the induction hypothesis:}
      &= \eta^p_{\sembr{\Theta}}\left(\sembr{\jstype{\Xi,\alpha}{A}}^p \circ \sembr{\jcms{\sigma,\alpha}{\Theta,\alpha}{\Xi,\alpha}}^p\right)\\
      \shortintertext{which by the interpretation of \rn{S-S-T$^q$}:}
      &= \eta^p_{\sembr{\Theta}}\left(\sembr{\jstype{\Xi,\alpha}{A}}^p \circ \langle \sembr{\jcms{\sigma}{\Theta,\alpha}{\Xi}}^p, \alpha : \sembr{\jstype{\Theta,\alpha}{\alpha}}^p \rangle \right)\\
      \shortintertext{which by \cref{eq:20051}:}
      &= \eta^p_{\sembr{\Theta}}\left(\sembr{\jstype{\Xi,\alpha}{A}}^p \circ \langle \sembr{\jcms{\sigma}{\Theta,\alpha}{\Xi}}^p, \alpha : \ms{id}_{\pi^{\Theta,\alpha}_\alpha} \rangle \right)\\
      \shortintertext{which by \cref{lemma:fscd:14}:}
      &= \eta^p_{\sembr{\Theta}}\left(\sembr{\jstype{\Xi,\alpha}{A}}^p \circ \langle \sembr{\jcms{\sigma}{\Theta}{\Xi}}^p\pi^{\Theta,\alpha}_\Theta, \alpha : \ms{id}_{\pi^{\Theta,\alpha}_\alpha} \rangle \right)\\
      \shortintertext{which by a property of products and projections:}
      &= \eta^p_{\sembr{\Theta}}\left(\sembr{\jstype{\Xi,\alpha}{A}}^p \circ \left( \sembr{\jcms{\sigma}{\Theta}{\Xi}}^p \times \sembr{\alpha} \right) \right)\\
      \shortintertext{which by naturality:}
      &= \eta^p_{\sembr{\Xi}}\left(\sembr{\jstype{\Xi,\alpha}{A}}^p\right) \circ \sembr{\jcms{\sigma}{\Theta}{\Xi}}^p\\
      &= \sembr{\jstype{\Xi}{\Trec{\alpha}{A}}}^p \circ \sembr{\jcms{\sigma}{\Theta}{\Xi}}^p.
    \end{align*}

  \item[Case \rn{C$\Tds{}$}.]
    \[
      \infer[\rn{C$\Tds{}$}]{
        \jstype[+]{\Xi}{\Tds A}
      }{
        \jstype[-]{\Xi}{A}
      }
    \]
    By the standard proof and \cref{prop:fscd:4}.

  \item[Case \rn{C$\Tus{}$}.]
    \[
      \infer[\rn{C$\Tus{}$}]{
        \jstype[-]{\Xi}{\Tus A}
      }{
        \jstype[+]{\Xi}{A}
      }
    \]
    By the standard proof and \cref{prop:fscd:4}.

  \item[Case \rn{C$\Tplus$}.]
    \[
      \infer[\rn{C$\Tplus$}]{
        \jstype[+]{\Xi}{{\Tplus\{l : A_l\}}_{l \in L}}
      }{
        \jstype[+]{\Xi}{A_l}\quad(\forall l \in L)
      }
    \]
    By the standard proof and \cref{prop:fscd:4}.

  \item[Case \rn{C$\Tamp$}.]
    \[
      \infer[\rn{C$\Tamp$}]{
        \jstype[-]{\Xi}{{\Tamp\{l :A_l \}}_{l \in L}}
      }{
        \jstype[-]{\Xi}{A_l}\quad(\forall l \in L)
      }
    \]
    By the standard proof and \cref{prop:fscd:4}.

  \item[Case \rn{C$\Tot$}.]
    \[
      \infer[\rn{C$\Tot$}]{
        \jstype[+]{\Xi}{A \Tot B}
      }{
        \jstype[+]{\Xi}{A}
        &
        \jstype[+]{\Xi}{B}
      }
    \]
    By the standard proof and \cref{prop:fscd:4}.

  \item[Case \rn{C$\Tlolly$}.]
    \[
      \infer[\rn{C$\Tlolly$}]{
        \jstype[-]{\Xi}{A \Tlolly B}
      }{
        \jstype[+]{\Xi}{A}
        &
        \jstype[-]{\Xi}{B}
      }
    \]
    By the standard proof and \cref{prop:fscd:4}.

  \item[Case \rn{C$\Tand{}{}$}.]
    \[
      \infer[\rn{C$\Tand{}{}$}]{
        \jstype[+]{\Xi}{\Tand{\tau}{A}}
      }{
        \jisft{\tau}
        &
        \jstype[+]{\Xi}{A}
      }
    \]
    By the standard proof and \cref{prop:fscd:4}.

  \item[Case \rn{C$\Timp{}{}$}.]
    \[
      \infer[\rn{C$\Timp{}{}$}]{
        \jstype[-]{\Xi}{\Timp{\tau}{A}}
      }{
        \jisft{\tau}
        &
        \jstype[-]{\Xi}{A}
      }
    \]
    By the standard proof and \cref{prop:fscd:4}.\qedhere
  \end{description}
\end{proof}

\subsection{Interpretations are Well-Defined}
\label{sec:interpr-are-well}

In this section we show that our interpretations are all well-defined, \ie, that the polarized aspects are locally continuous functors.

\begin{proposition}[Functorial Interpretations are Well-Defined]
  \label{prop:fscd:9}
  If $\jstype{\Xi}{A}$, then the interpretations $\sembr{\jstype{\Xi}{A}}^-$ and $\sembr{\jstype{\Xi}{A}}^+$ are locally continuous functors from $\sembr{\Xi}$ to $\moabcs$.
\end{proposition}

\begin{proof}
  By induction on the derivation of $\jstype{\Xi}{A}$.
  We must show that each of the functor interpretations is locally continuous.

  \begin{description}
  \item[Case \rn{C$\Tu$}.]
    \[
      \infer[\rn{C$\Tu$}]{\jstype[+]{\Xi}{\Tu}}{}
    \]
    Recall \cref{eq:1502,eq:1501}:
    \begin{align*}
      \sembr{\jstype{\Xi}{\Tu}}^- &= \lambda \xi.\top_{\moabcs}\\
      \sembr{\jstype{\Xi}{\Tu}}^+ &= \lambda \xi.\{\ast\}_\bot
    \end{align*}
    Constant functors are easily seen to be locally continuous.

  \item[Case \rn{CVar}.]
    \[
      \infer[\rn{CVar}]{\Xi,\jisst[p]{\alpha}\vdash\jisst[p]{\alpha}}{}
    \]
    Recall \cref{eq:fscd:7}:
    \begin{align*}
      \sembr{\Xi,\jisst[p]{\alpha}\vdash\jisst[p]{\alpha}}^q &= \pi^{\Xi,\alpha}_\alpha \quad (q \in \{{-},{+}\})
    \end{align*}
    Projection functors are easily seen to be locally continuous.

  \item[Case \rn{C$\rho$}.]
    \[
      \infer[\rn{C$\rho$}]{
        \Xi \vdash \jisst[p]{\Trec{\alpha}{A}}
      }{
        \Xi, \jisst[p]{\alpha} \vdash \jisst[p]{A}
      }
    \]
    Recall \cref{eq:20051}:
    \begin{align*}
      \sembr{\jstype{\Xi}{\Trec{\alpha}{A}}}^p &= \sfix{\left(\sembr{\Xi,\jisst{\alpha}\vdash \jisst{A}}^p\right)}\quad (p \in \{{-},{+}\})
    \end{align*}
    By the induction hypothesis, $\sembr{\Xi,\jisst{\alpha}\vdash \jisst{A}}^p$ is locally continuous.
    By \cite[Proposition~5.2.7]{abramsky_jung_1995:_domain_theor}, we know that whenever $F : \sembr{\Xi} \times \moabcs \to \moabcs$ is locally continuous, then so is $\sfix{F} : \sembr{\Xi} \to \moabcs$.
    It follows that the interpretations are locally continuous.

  \item[Case \rn{C$\Tds{}$}.]
    \[
      \infer[\rn{C$\Tds{}$}]{
        \jstype[+]{\Xi}{\Tds A}
      }{
        \jstype[-]{\Xi}{A}
      }
    \]
    Recall \cref{eq:1506,eq:1507}:
    \begin{align*}
      \sembr{\jstype{\Xi}{\Tds A}}^- &= \sembr{\jstype{\Xi}{A}}^-\\
      \sembr{\jstype{\Xi}{\Tds A}}^+ &= \sembr{\jstype{\Xi}{A}}^+_\bot
    \end{align*}
    By the induction hypothesis, the interpretations for $\jstype{\Xi}{A}$ are locally continuous.
    The lifting functor is locally continuous.
    Locally continuous functors are closed under composition.
    This gives the result.

  \item[Case \rn{C$\Tus{}$}.]
    \[
      \infer[\rn{C$\Tus{}$}]{
        \jstype[-]{\Xi}{\Tus A}
      }{
        \jstype[+]{\Xi}{A}
      }
    \]
    This case is analogous to the case \rn{C$\Tds{}$}.

  \item[Case \rn{C$\Tplus$}.]
    \[
      \infer[\rn{C$\Tplus$}]{
        \jstype[+]{\Xi}{{\Tplus\{l : A_l\}}_{l \in L}}
      }{
        \jstype[+]{\Xi}{A_l}\quad(\forall l \in L)
      }
    \]
    Recall \cref{eq:202020,eq:222222}:
    \begin{align*}
      \sembr{\jstype{\Xi}{\Tplus\{l:A_l\}_{l \in L}}}^- &= \prod_{l \in L} \sembr{\jstype{\Xi}{A_l}}^-\\
      \sembr{\jstype{\Xi}{\Tplus\{l:A_l\}_{l \in L}}}^+ &= \bigoplus_{l \in L} \sembr{\jstype{\Xi}{A_l}}^+_\bot
    \end{align*}
    By the induction hypothesis, the interpretations for $\jstype{\Xi}{A}$ are locally continuous.
    The lifting, coalesced-sum, and product functors are locally continuous.
    Locally continuous functors are closed under composition.
    This gives the result.

  \item[Case \rn{C$\Tamp$}.]
    \[
      \infer[\rn{C$\Tamp$}]{
        \jstype[-]{\Xi}{{\Tamp\{l :A_l \}}_{l \in L}}
      }{
        \jstype[-]{\Xi}{A_l}\quad(\forall l \in L)
      }
    \]
    This case is analogous to the case \rn{C$\Tplus$}.

  \item[Case \rn{C$\Tot$}.]
    \[
      \infer[\rn{C$\Tot$}]{
        \jstype[+]{\Xi}{A \Tot B}
      }{
        \jstype[+]{\Xi}{A}
        &
        \jstype[+]{\Xi}{B}
      }
    \]
    Recall \cref{eq:13,eq:14}:
    \begin{align*}
      \sembr{\jstype{\Xi}{A \Tot B}}^- &= \sembr{\jstype{\Xi}{A}}^- \times \sembr{\jstype{\Xi}{B}}^-\\
      \sembr{\jstype{\Xi}{A \Tot B}}^+ &= \left(\sembr{\jstype{\Xi}{A}}^+ \times \sembr{\jstype{\Xi}{B}}^+\right)_\bot
    \end{align*}
    By the induction hypothesis, the interpretations for $\jstype{\Xi}{A}$ and $\jstype{\Xi}{B}$ are locally continuous.
    The lifting and product functors are locally continuous.
    Locally continuous functors are closed under composition.
    This gives the result.

  \item[Case \rn{C$\Tlolly$}.]
    \[
      \infer[\rn{C$\Tlolly$}]{
        \jstype[-]{\Xi}{A \Tlolly B}
      }{
        \jstype[+]{\Xi}{A}
        &
        \jstype[-]{\Xi}{B}
      }
    \]
    This case is analogous to the case \rn{C$\Tot$}.

  \item[Case \rn{C$\Tand{}{}$}.]
    \[
      \infer[\rn{C$\Tand{}{}$}]{
        \jstype[+]{\Xi}{\Tand{\tau}{A}}
      }{
        \jisft{\tau}
        &
        \jstype[+]{\Xi}{A}
      }
    \]
    Recall \cref{eq:15104,eq:15106}:
    \begin{align*}
      \sembr{\jstype{\Xi}{\Tand{\tau}{A}}}^- &= \sembr{\Xi\vdash \jisst{A}}^-\\
      \sembr{\jstype{\Xi}{\Tand{\tau}{A}}}^+ &= \left(\sembr{\tau}\times\sembr{\Xi\vdash \jisst{A}}^+\right)_\bot
    \end{align*}
    By the induction hypothesis, the interpretations for $\jstype{\Xi}{A}$ are locally continuous.
    The constant functor $\sembr{\tau}$, the lifting functor, and the product functors are locally continuous.
    Locally continuous functors are closed under composition.
    This gives the result.

  \item[Case \rn{C$\Timp{}{}$}.]
    \[
      \infer[\rn{C$\Timp{}{}$}]{
        \jstype[-]{\Xi}{\Timp{\tau}{A}}
      }{
        \jisft{\tau}
        &
        \jstype[-]{\Xi}{A}
      }
    \]
    This case is analogous to the case \rn{C$\Tand{}{}$}.\qedhere
  \end{description}
\end{proof}


\section{Illustrative Example: Flipping Bit Streams}
\label{sec:case-study-bit-str}

We give the full details for the case study in \cref{sec:bit-streams}.
We will show that flipping the bits in a stream twice is equivalent to forwarding it.
The bit stream protocol is specified by the session type $\mt{bits} = \Trec{\beta}{\Tplus \{ \mt{0} : \beta, \mt{1} : \beta \}}$.
It denotes the domains $\sembr{\mt{bits}}^+ = \FIX\left(X \mapsto \left(\left(\mt 0 : X_\bot\right) \oplus \left(\mt 1 : X_\bot \right)\right)\right)$ and $\sembr{\mt{bits}}^- = \{\bot\}$.
Its unfolding is $\mt{BITS} = \Tplus\{\mt{0} : \mt{bits}, \mt{1} : \mt{bits}\}$.
There are canonical isomorphisms
\begin{align*}
  &\ms{Unfold}^+ : \sembr{\mt{bits}}^+ \to \left(\left(\mt{0}:\sembr{\mt{bits}}^+_\bot\right) \oplus \left(\mt{1}:\sembr{\mt{bits}}^+_\bot\right)\right)\\
  &\ms{Unfold}^- : \{\bot\} \to \{(\mt{0} : \bot, \mt{1} : \bot)\}
\end{align*}
with respective inverses $\ms{Fold}^+$ and $\ms{Fold}^-$.
Write $\cons{0}{\alpha}$ and $\cons{1}{\alpha}$ for $\Fold^+((\mt{0},\upim{\alpha}))$ and $\Fold^+((\mt{1},\upim{\alpha}))$, respectively.

We desugar the bit-flipping process $\mt{flip}$.
Write $\tilde l$ for the complement of $l \in \{\mt{0},\mt{1}\}$.
Desugaring the quoted process gives a term $\jtypef{\cdot}{\mt{flip}}{\Tproc{f:\mt{bits}}{b:\mt{bits}}}$ where $\mt{flip}$ is:
\[
  \tFix{F}{
    \tProc{f}{\tSendU{f}{\tRecvU{b}{\tCase{b}{\{l \Rightarrow \tSendL{f}{\tilde l}{\tProc{f}{F}{b}}\}}_{l \in \{\mt{0},\mt{1}\}}}}}{b}
  }.
\]

We compute its denotation.
Let $\phi = \Tproc{f : \mt{bits}}{b : \mt{bits}}$.
The branch $B_{\mt 0} = \verb!f.1; f <- F <- b!$ of $\mt{flip}$'s case statement has the derivation:
\[
  \infer[\rn{$\Tplus R_{\mt 1}$}]{
    \jtypem{F : \phi}{b : \mt{bits}}{
      \tSendL{f}{\mt 1}{\tProc{f}{F}{b}}
    }{f}{\mt{BITS}}
  }{
    \infer[\rn{E-\{\}}]{
      \jtypem{F : \phi}{b : \mt{bits}}{
        \tProc{f}{F}{b}
      }{f}{\mt{bits}}
    }{
      \infer[\rn{F-Var}]{\jtypef{F : \phi}{F}{\phi}}{}
    }
  }
\]
We calculate that its denotation is:
\begin{align*}
  &\sembr{\jtypem{F : \phi}{b : \mt{bits}}{\tSendL{f}{\mt 1}{\tProc{t}{F}{b}}}{f}{\mt{BITS}}}(F : g)\left(b^+, \bot\right) = \left(\bot, \upim{\left(\mt 1, a^+\right)}\right)\\
  &\text{where }(\bot, a^+) = \down\left(g\right)(b^+, \bot).
\end{align*}
The branch $B_{\mt 1} = \verb!f.0; f <- F <- b!$ is symmetric.

The functional term $\mt{flip}$ has the derivation:
\[
  \adjustbox{width=0.99\linewidth,keepaspectratio}\bgroup
  \infer[\rn{F-Fix}]{
    \jtypef{\cdot}{
      \tFix{F}{
        \tProc{f}{\tSendU{f}{\tRecvU{b}{\tCase{b}{\{l \Rightarrow B_l\}}_{l \in \{\mt{0},\mt{1}\}}}}}{b}
      }
    }{\phi}
  }{
    \infer[\rn{I-\{\}}]{
      \jtypef{F : \phi}{
        \tProc{f}{\tSendU{f}{\tRecvU{b}{\tCase{b}{\{l \Rightarrow B_l\}}_{l \in \{\mt{0},\mt{1}\}}}}}{b}
      }{\phi}
    }{
      \infer[\rn{$\rho^+R$}]{
        \jtypem{F:\phi}{b : \mt{bits}}{
          \tSendU{f}{\tRecvU{b}{\tCase{b}{\{l \Rightarrow B_l\}}_{l \in \{\mt{0},\mt{1}\}}}}
        }{f}{\mt{bits}}
      }{
        \infer[\rn{$\rho^+L$}]{
          \jtypem{F:\phi}{b : \mt{bits}}{
            \tRecvU{b}{\tCase{b}{\{l \Rightarrow B_l\}}_{l \in \{\mt{0},\mt{1}\}}}
          }{f}{\mt{BITS}}
        }{
          \infer[\rn{$\Tplus$L}]{
            \jtypem{F:\phi}{b : \mt{BITS}}{
              \tCase{b}{\{l \Rightarrow B_l\}}_{l \in \{\mt{0},\mt{1}\}}
            }{f}{\mt{BITS}}
          }{
            \jtypem{F : \phi}{b : \mt{bits}}{
              B_l
            }{f}{\mt{BITS}}
            &
            (\forall l \in \{\mt{0},\mt{1}\})
          }
        }
      }
    }
  }
  \egroup
\]
We calculate its denotation as follows.
First, we find the denotation of the process:
\begin{align*}
  &\sembr{\jtypem{F:\phi}{b : \mt{bits}}{
    \tSendU{f}{\tRecvU{b}{\tCase{b}{\{l \Rightarrow B_l\}}_{l \in \{\mt{0},\mt{1}\}}}}
    }{f}{\mt{bits}}}(F : g)\\
  &= \strictfn_{b^+}\left(
    \begin{aligned}
      \lambda (b^+,\bot).&
      \begin{cases}
        \left(\bot, \cons{1}{a_{\mt 0}^+}\right) & \text{if } b^+ = \cons{0}{b_{\mt 0}^+}\\
        \left(\bot,  \cons{0}{a_{\mt 1}^+}\right) & \text{if } b^+ = \cons{1}{b_{\mt 1}^+}
      \end{cases}\\
      &\quad\text{where }\down\left(g\right)(b_l^+, \bot) = (\bot, a_l^+) \text{ for } l \in L
    \end{aligned}\right)
\end{align*}
Let $\Phi : \sembr{b : \mt{bits} \vdash a : \mt{bits}} \to \sembr{b : \mt{bits} \vdash a : \mt{bits}}$ be the function
\[
  \Phi(r) = \strictfn_{b^+} \left(\begin{aligned}
      \lambda (b^+,\bot).&
      \begin{cases}
        \left(\bot, \cons{1}{a_{\mt 0}^+}\right) & \text{if } b^+ = \cons{0}{b_{\mt 0}^+}\\
        \left(\bot,  \cons{0}{a_{\mt 1}^+}\right) & \text{if } b^+ = \cons{1}{b_{\mt 1}^+}
      \end{cases}\\
      &\quad\text{where }r(b_l^+, \bot) = (\bot, a_l^+) \text{ for } l \in L
    \end{aligned}
  \right)
\]
Quoting the process gives:
\begin{align*}
  &\sembr{\jtypef{F : \phi}{\tProc{f}{\tSendU{f}{\ldots}}{b}}{\phi}}\\
  &= \up \circ \sembr{\jtypem{F:\phi}{b : \mt{bits}}{\tSendU{f}{\ldots}}{f}{\mt{bits}}}\\
  &= \up \circ \Phi \circ \down
\end{align*}
By \cref{eq:28}, fixing the functional variable $F$ gives:
\begin{align*}
  &\sembr{\jtypef{\cdot}{\tFix{F}{\tProc{f}{\tSendU{f}{\ldots}}{b}}}{\phi}}\bot\\
  &= \sfix{\sembr{\jtypef{F : \phi}{\tProc{f}{\tSendU{f}{\ldots}}{b}}{\phi}}}\bot\\
  &= \sfix{(\up \circ \Phi \circ \down)}\bot\\
  &= \fix\left(\up \circ \Phi \circ \down\right)\\
  \shortintertext{which by \cref{prop:fscd:15}:}
  &= \up(\fix(\Phi)).
\end{align*}

Consider the composition $\jtypem{\cdot}{a : \mt{bits}}{\tProc{t}{\mt{flip}}{a}; \tProc{b}{\mt{flip}}{t}}{b}{\mt{bits}}$ of two bit-flipping processes.
We want to show that it is equivalent to forwarding, \ie, that
\[
  \sembr{\jtypem{\cdot}{a : \mt{bits}}{\tProc{t}{\mt{flip}}{a}; \tProc{b}{\mt{flip}}{t}}{b}{\mt{bits}}} = \sembr{\jtypem{\cdot}{a : \mt{bits}}{\tFwd{b}{a}}{b}{\mt{bits}}}.
\]
This means that we must show for all $(a^+,\bot) \in \sembr{a : \mt{bits}}^+ \times \sembr{b : \mt{bits}}^-$:
\begin{equation}
  \sembr{\jtypem{\cdot}{a : \mt{bits}}{\tProc{t}{\mt{flip}}{a}; \tProc{b}{\mt{flip}}{t}}{b}{\mt{bits}}}\bot(a^+,\bot) = (\bot, a^+).\label{eq:fscd:99}
\end{equation}

We first compute the denotation of the composition.
By \cref{eq:115},
\begin{align*}
  &\sembr{\jtypem{\cdot}{a : \mt{bits}}{\tProc{t}{\mt{flip}}{a}; \tProc{b}{\mt{flip}}{t}}{b}{\mt{bits}}}\bot(a^+, \bot)\\
  &= \Trop^{t^- \times t^+}\left(
    \down\left(\sembr{\jtypef{\cdot}{\mt{flip}}{\Tproc{t : \mt{bits}}{a : \mt{bits}}}}\bot\right)
    \times {}\right.\\
  &\qquad\qquad\qquad\left. {} \times
    \down\left(\sembr{\jtypef{\cdot}{\mt{flip}}{\Tproc{b : \mt{bits}}{t : \mt{bits}}}}\bot\right)
    \right)(a^+, \bot)\\
  &= \Trop^{t^- \times t^+}\left(\down(\up(\fix(\Phi))) \times \down(\up(\fix(\Phi)))\right)(a^+,\bot)\\
  &= \Trop^{t^- \times t^+}\left(\fix(\Phi) \times \fix(\Phi)\right)(a^+,\bot)\\
    \shortintertext{which by \cref{lemma:fscd:4} with $\sigma : \sembr{a : \mt{bits}}^- \times \sembr{t : \mt{bits}}^+ \to \sembr{t : \mt{bits}}^+ \times \sembr{b : \mt{bits}}^-$ the obvious relabelling isomorphism:}
  &= \left(\fix(\Phi) \circ \sigma \circ \fix(\Phi)\right)(a^+, \bot)
\end{align*}
Alternatively, we could compute the denotation as the least upper bound of an ascending chain:
\begin{align*}
  &\sembr{\jtypem{\cdot}{a : \mt{bits}}{\tProc{t}{\mt{flip}}{a}; \tProc{b}{\mt{flip}}{t}}{b}{\mt{bits}}}\bot(a^+, \bot)\\
  &= \Trop^{t^- \times t^+}\left(\fix(\Phi) \times \fix(\Phi)\right)(a^+,\bot)\\
  &= \Trop^{t^- \times t^+}\left(
    \down\left(\sembr{\jtypef{\cdot}{\mt{flip}}{\Tproc{t : \mt{bits}}{a : \mt{bits}}}}\bot\right)
    \times {}\right.\\
  &\qquad\qquad\qquad\left. {} \times
    \down\left(\sembr{\jtypef{\cdot}{\mt{flip}}{\Tproc{b : \mt{bits}}{t : \mt{bits}}}}\bot\right)
    \right)(a^+, \bot)\\
  &= \dirsup_{n \in \N}\; (b_n^-, a_n^+)
\end{align*}
where we inductively define
\begin{align*}
  (a_0^-, t_0^+) &= \sembr{\jtypef{\cdot}{\mt{flip}}{\Tproc{t : \mt{bits}}{a : \mt{bits}}}}\bot(a^+, \bot) = \fix(\Phi)(a^+,\bot)\\
  (t_0^-, b_0^+) &= \sembr{\jtypef{\cdot}{\mt{flip}}{\Tproc{b : \mt{bits}}{t : \mt{bits}}}}\bot(\bot, \bot) = \fix(\Phi)(\bot,\bot)\\
  (a_{n+1}^-, t_{n+1}^+) &= \sembr{\jtypef{\cdot}{\mt{flip}}{\Tproc{t : \mt{bits}}{a : \mt{bits}}}}\bot(a^+, t_n^-) =  \fix(\Phi)(a^+,t_n^-)\\
  (t_{n+1}^-, b_{n+1}^+) &= \sembr{\jtypef{\cdot}{\mt{flip}}{\Tproc{b : \mt{bits}}{t : \mt{bits}}}}\bot(t_n^+, \bot) = \fix(\Phi(t_n^+,\bot).
\end{align*}
We observe that  $a_n^- = t_n^- = \bot$  and $t_n^+ = t_{n + 1}^+$ for all $n \in \N$, so this chain stabilizes and stops evolving at $n = 2$.
We conclude that
\[
  \dirsup_{n \in \N}\; (b_n^-, a_n^+) = (b_2^-, a_2^+) = \left(\fix(\Phi) \circ \sigma \circ \fix(\Phi)\right)(a^+, \bot).
\]

Given an $a \in \sembr{\mt{bits}}^+$, let $\fd a$ be given by $\left(\fix(\Phi) \circ \sigma \circ \fix(\Phi)\right)(a, \bot) = (\bot, \fd a)$.
To show \cref{eq:fscd:99}, we must show that $a = \fd a$ for all $a \in \sembr{\mt{bits}}^+$.
We begin by characterizing $\fd a$.
We claim that:
\begin{enumerate}
\item if $a = \bot$, then $\fd a = \bot$;\label{item:fscd:1}
\item if $a = \cons{0}{\alpha}$, then $\fd a = \cons{0}{\fd \alpha}$;\label{item:fscd:2}
\item if $a = \cons{1}{\alpha}$, then $\fd a = \cons{1}{\fd \alpha}$.\label{item:fscd:3}
\end{enumerate}
Claim \ref{item:fscd:1} is clear by strictness.
By the fixed-point identity $\fix(\Phi) = \Phi(\fix(\Phi))$, we have for all $\alpha\in \sembr{\mt{bits}}^+$ that:
\begin{align*}
  &\fix(\Phi)(\cons{0}{\alpha}, \bot) = \Phi(\fix(\Phi))(\cons{0}{\alpha}, \bot) = (\bot, \cons{1}{\beta})\text{ where }(\bot, \beta) = \fix(\Phi)(\alpha, \bot)\\
  &\fix(\Phi)(\cons{1}{\alpha}, \bot) = \Phi(\fix(\Phi))(\cons{1}{\alpha}, \bot) = (\bot, \cons{0}{\beta})\text{ where }(\bot, \beta) = \fix(\Phi)(\alpha, \bot)
\end{align*}
These imply:
\begin{align}
  \left(\fix(\Phi) \circ \sigma \circ \fix(\Phi)\right)(\cons{0}{\alpha}, \bot) &= \fix(\Phi)(\cons{1}{\beta}, \bot) = (\cons{0}{\fd \alpha}, \bot)\label{eq:fossacs:8}\\
  \left(\fix(\Phi) \circ \sigma \circ \fix(\Phi)\right)(\cons{1}{\alpha}, \bot) &= \fix(\Phi)(\cons{0}{\beta}, \bot) = (\cons{1}{\fd \alpha}, \bot)\label{eq:fossacs:9}.
\end{align}
\Cref{eq:fossacs:8,eq:fossacs:9} respectively establish claims \ref{item:fscd:2} and \ref{item:fscd:3}.

To show that $a = \fd a$, we use a coinduction principle by Pitts~\cite{pitts_1994:_co_induc_princ}.
Given a relation $\R$ on $\sembr{\mt{bits}}^+_\bot$, let $\Rh$ be the relation on $\sembr{\mt{BITS}}^+_\bot$ where $a \Rh b$ if and only if for all $\alpha \in \sembr{\mt{bits}}^+$:
\begin{itemize}
\item if $a = \upim{(\mt{0}, \upim{\alpha})}$, then $b = \upim{(\mt{0}, \upim{\beta})}$ for some $\beta \in \sembr{\mt{bits}}^+$ with $\upim{\alpha} \R \upim{\beta}$;
\item if $a = \upim{(\mt{1}, \upim{\alpha})}$, then $b = \upim{(\mt{1}, \upim{\beta})}$ for some $\beta \in \sembr{\mt{bits}}^+$ with $\upim{\alpha} \R \upim{\beta}$.
\end{itemize}
The relation $\R$ is a \defin{simulation} if $a \R b$ implies $\ms{Unfold}^+_\bot(a) \Rh \ms{Unfold}^+_\bot(b)$.
By \cite[Theorem~2.5]{pitts_1994:_co_induc_princ}, $a \sqsubseteq b$ in $\sembr{\mt{bits}}^+_\bot$ if and only if there exists a simulation $\R$ such that $a \R b$.
Let ${\R} \subseteq \sembr{\mt{bits}}^+_\bot \times \sembr{\mt{bits}}^+_\bot$ be the relation:
\[
  {\R} = \left\{ (\upim{\bot},\upim{\bot}), (\upim{\cons{0}{\alpha}}, \upim{\cons{0}{\fd \alpha}}), (\upim{\cons{1}{\alpha}}, \upim{\cons{1}{\fd \alpha}}) \mid \alpha \in \sembr{\mt{bits}}^+ \right\}.
\]

Assume $\R$ is a simulation.
To show that $a \sqsubseteq \fd a$, it is sufficient to show that $\upim{a} \R \upim{\fd a}$.
\begin{enumerate}
\item If $a = \bot$, then $\fd a = \bot$ by strictness.
\item If $a = \cons{0}{\alpha}$, then $\fd a = \cons{0}{\fd \alpha}$ by \cref{eq:fossacs:8}.
\item If $a = \cons{1}{\alpha}$, then $\fd a = \cons{1}{\fd \alpha}$ by \cref{eq:fossacs:9}.
\end{enumerate}
So $\upim{a} \R \upim{\fd a}$ by definition of $\R$.
It follows that $a \sqsubseteq \fd a$.

The relation $\R$ is clearly a simulation if and only if $\op\R$ is a simulation.
So $\fd a \mathrel{\op\R} a$, which gives $\fd a \sqsubseteq a$.
We conclude $a = \fd a$ as desired.

We now show that $\R$ is a simulation.
Let $(\upim{a}, \upim{b}) \in {\R}$ be arbitrary.
We must show that $\Unfold^+_\bot(\upim{a}) \Rh \Unfold^+_\bot(\upim{b})$.
Assume first that $a = \cons{0}{\alpha}$.
By definition of $\R$, $b = \cons{0}{\fd \alpha}$.
We have $\Unfold^+_\bot(a) = \upim{(\mt{0}, \upim{\alpha})}$ and $\Unfold^+_\bot(b) = \upim{(\mt{0}, \upim{\fd \alpha})}$.
We must show that $\upim{\alpha} \R \upim{\fd \alpha}$.
We proceed by case analysis on $\alpha$.
\begin{enumerate}
\item If $\alpha = \bot$, then $\fd \alpha = \bot$ by strictness, and $\upim{\bot} \R \upim{\bot}$ by definition of $\R$.
\item If $\alpha = \cons{0}{\gamma}$, then $\fd \alpha = \cons{0}{\fd \gamma}$ by \cref{eq:fossacs:8}, and $\upim{\cons{0}{\gamma}} \R \upim{\cons{0}{\fd \gamma}}$ by definition of $\R$.
\item If $\alpha = \cons{1}{\gamma}$, then $\fd \alpha = \cons{1}{\fd \gamma}$ by \cref{eq:fossacs:9}, and $\upim{\cons{1}{\gamma}} \R \upim{\cons{1}{\fd \gamma}}$ by definition of $\R$.
\end{enumerate}
This completes the case of $a = \cons{0}{\alpha}$.
The case of $a = \cons{1}{\alpha}$ is symmetric.
We conclude that $\R$ is a simulation.

It follows that flipping a stream twice is equivalent to forwarding (the identity process).


\section{Proofs for Results in the Paper}
\label{sec:proofs-results-paper}

We restate the statements of the results in the main part of the paper and give their proofs.

\recsubst*

\begin{proof}
  Let $\Psi = \psi_1 : \tau_1, \dotsc, \psi_n : \tau_n$.
  By \cref{eq:28},
  \[
    \sembr{\jtypef{\Psi}{\tFix{x}{M}}{\tau}} = \sfix{\sembr{\jtypef{\Psi,x : \tau}{M}{\tau}}}.
  \]
  By the parametrized fixed-point property (\cref{def:fscd:2}),
  \[
    \sfix{\sembr{\jtypef{\Psi,x : \tau}{M}{\tau}}} = \sembr{\jtypef{\Psi,x : \tau}{M}{\tau}} \circ \langle \ms{id}_{\sembr{\Psi}}, \sembr{\jtypef{\Psi}{\tFix{x}{M}}{\tau}} \rangle.
  \]
  By \cref{eq:39} and properties of products, this equals
  \[
    \sembr{\jtypef{\Psi,x : \tau}{M}{\tau}} \circ \langle \sembr{\jtypef{\Psi}{\psi_1}{\tau_1}}, \dotsc, \sembr{\jtypef{\Psi}{\psi_n}{\tau_n}}, \sembr{\jtypef{\Psi}{\tFix{x}{M}}{\tau}} \rangle,
  \]
  which by \cref{prop:11} is $\sembr{\jtypef{\Psi}{[\tFix{x}{M}/x]M}{\tau}}$.
\end{proof}

\quoteunquote*

\begin{proof}
  Let $u \in \sembr{\Psi}$ be arbitrary.
  Using the fact that $(\up,\down)$ is a section-retraction pair, we get
  \begin{align*}
    &\sembr{\jtypem{\Psi}{\overline{a_i:A_i}}{\tProc{a}{\tProc{a}{P}{\overline a_i}}{\overline a_i}}{a}{A}}u\\
    &= \down\left(\sembr{\jtypef{\Psi}{\tProc{a}{P}{\overline a_i}}{\Tproc{a:A}{\overline{a_i:A_i}}}}u\right)\\
    &= \down\left(\left(\up \circ \sembr{\jtypem{\Psi}{\overline{a_i:A_i}}{P}{a}{A}}\right)u\right)\\
    &= (\down \circ \up)\left(\sembr{\jtypem{\Psi}{\overline{a_i:A_i}}{P}{a}{A}}u\right)\\
    &= \sembr{\jtypem{\Psi}{\overline{a_i:A_i}}{P}{a}{A}}u.\qedhere
  \end{align*}
\end{proof}

The proofs of the $\eta$-like equivalences make extensive use of the Knaster-Tarski formulation \eqref{eq:808} of the trace operator.
Though the proofs are detailed, they are not conceptually difficult.

\etaquote*

\begin{proof}
  Given $u \in \sembr{\Psi}$, write $\hat u$ for $\upd{u}{x \mapsto \sembr{\jtypef{\Psi}{M}{\tau}}u} \in \sembr{\Psi,x : \tau}$.
  Set
  \begin{align*}
    &l : \sembr{\Psi} \times \sembr{\Delta_1, \Delta_2, a:A}^+ \times \sembr{c:C}^- \to \sembr{\Delta_1, \Delta_2, a:A}^- \times \sembr{c:C}^+\\
    &l = \Lambda^{-1}\left(\lambda u \in \sembr{\Psi} . \sembr{\jtypem{\Psi}{\Delta_1}{P}{a}{A}}u \times \sembr{\jtypem{\Psi,x : \tau}{\Delta_2, a:A}{Q}{c}{C}}\hat u \right)\\
    &r : \sembr{\Psi} \times \sembr{\Delta_1, \Delta_2, a:\Tand{\tau}{A}}^+ \times \sembr{c:C}^- \to \sembr{\Delta_1, \Delta_2, a:\Tand{\tau}{A}}^- \times \sembr{c:C}^+\\
    &r = \Lambda^{-1}\left(\lambda u \in \sembr{\Psi} . \sembr{\jtypem{\Psi}{\Delta_1}{\tSendV{a}{M}{P}}{a}{\Tand{\tau}{A}}}u \times {} \right.\\
    &\qquad\qquad \left. {} \times \sembr{\jtypem{\Psi}{\Delta_2, a : \Tand{\tau}{A}}{\tRecvV{x}{a}{Q}}{c}{C}}u\right)
  \end{align*}
  Let $x_i : \tau_i$ be such that $\Psi = {x_1 : \tau_1},\dotsc, {x_n : \tau_n}$.
  By substitution (\cref{prop:extended:10}), we have
  \begin{align*}
    &\sembr{\jtypem{\Psi}{\Delta_1,\Delta_2}{\tCut{a}{P}{[M/x]Q}}{c}{C}}u\\
    &= \left(\sembr{\jtypem{\Psi, x : \tau}{\Delta_1,\Delta_2}{\tCut{a}{P}{Q}}{c}{C}} \circ {} \right.\\
    &\qquad\qquad\left. {} \circ \langle \sembr{\jtypef{\Psi}{x_1}{\tau_1}}, \dotsc, \sembr{\jtypef{\Psi}{x_n}{\tau_n}}, \sembr{\jtypef{\Psi}{M}{\tau}} \rangle\right)u\\
    &= \left(\sembr{\jtypem{\Psi, x : \tau}{\Delta_1,\Delta_2}{\tCut{a}{P}{Q}}{c}{C}} \circ \langle \ms{id}_{\sembr{\Psi}}, \sembr{\jtypef{\Psi}{M}{\tau}}\rangle\right)u\\
    &= \sembr{\jtypem{\Psi, x : \tau}{\Delta_1,\Delta_2}{\tCut{a}{P}{Q}}{c}{C}}\upd{u}{x \mapsto \sembr{\jtypef{\Psi}{M}{\tau}}u}\\
    &= \sembr{\jtypem{\Psi, x : \tau}{\Delta_1,\Delta_2}{\tCut{a}{P}{Q}}{c}{C}}\hat u\\
    \shortintertext{which by \cref{eq:11}:}
    &= \Trop^{a^- \times a^+}_{\Delta_1^+\times\Delta_2^+\times c^-,\Delta_1^-\times\Delta_2^-\times c^+}\left(\sembr{\jtypem{\Psi, x : \tau}{\Delta_1}{P}{a}{A}} \hat u \times {} \right.\\
    &\qquad\quad \left. {} \times \sembr{\jtypem{\Psi,x : \tau}{\Delta_2, a:A}{Q}{c}{C}}\hat u \right)\\
    \shortintertext{which by coherence (\cref{prop:extended:9}):}
    &= \Trop^{a^- \times a^+}_{\Delta_1^+\times\Delta_2^+\times c^-,\Delta_1^-\times\Delta_2^-\times c^+}\left(\sembr{\jtypem{\Psi}{\Delta_1}{P}{a}{A}}u \times {} \right.\\
    &\qquad\qquad \left. {} \times \sembr{\jtypem{\Psi,x : \tau}{\Delta_2, a:A}{Q}{c}{C}}\hat u \right)\\
    \shortintertext{which by \cref{prop:fscd:16}:}
    &=  \lambda (\delta_1^+, \delta_2^+, c^-) \in \sembr{\Delta_1, \Delta_2}^+ \times \sembr{c:C}^- . \Trop^{a^- \times a^+}_{\Psi\times\Delta_1^+\times\Delta_2^+\times c^-,\Delta_1^-\times\Delta_2^-\times c^+}\left(l\right)(u, \delta_1^+, \delta_2^+, c^-) .
  \end{align*}

  Similarly, we have
  \begin{align*}
    &\sembr{\jtypem{\Psi}{\Delta_1,\Delta_2}{\tCut{a}{(\tSendV{a}{M}{P})}{(\tRecvV{x}{a}{Q})}}{c}{C}}u\\
    &= \Trop^{a^- \times a^+}\left(\sembr{\jtypem{\Psi}{\Delta_1}{\tSendV{a}{M}{P}}{a}{\Tand{\tau}{A}}}u \times {} \right.\\
    &\qquad\qquad \left. {} \times \sembr{\jtypem{\Psi}{\Delta_2, a : \Tand{\tau}{A}}{\tRecvV{x}{a}{Q}}{c}{C}}u\right)\\
    \shortintertext{which by \cref{prop:fscd:16}:}
    &=  \lambda (\delta_1^+, \delta_2^+, c^-) \in \sembr{\Delta_1, \Delta_2}^+ \times \sembr{c:C}^- . \Trop^{a^- \times a^+}_{\Psi\times\Delta_1^+\times\Delta_2^+\times c^-,\Delta_1^-\times\Delta_2^-\times c^+}\left(r\right)(u, \delta_1^+, \delta_2^+, c^-) .
  \end{align*}

  We abbreviate judgments by their processes.
  We use the Knaster-Tarski formulation \eqref{eq:808} of the trace operator.
  Let $u \in \sembr{\Psi}$ and $(\delta_1^+,\delta_2^+,c^-) \in \sembr{\Delta_1,\Delta_2}^+ \times \sembr{c : C}^-$ be arbitrary.
  Using the above chains of equations, we calculate that
  \begin{align*}
    &\sembr{\jtypem{\Psi}{\Delta_1,\Delta_2}{\tCut{a}{P}{[M/x]Q}}{c}{C}}u(\delta_1^+,\delta_2^+,c^-)\\
    &= \pi_{\Delta_1^-,\Delta_2^-,c^+}\left(\bigsqcap L\right),\\
    &\sembr{\jtypem{\Psi}{\Delta_1,\Delta_2}{\tCut{a}{(\tSendV{a}{M}{P})}{(\tRecvV{x}{a}{Q})}}{c}{C}}u(\delta_1^+,\delta_2^+,c^-)\\
    &= \pi_{\Delta_1^-,\Delta_2^-,c^+}\left(\bigsqcap R\right)
  \end{align*}
  where
  \begin{align*}
    &L = \left\{ \left(\delta_1^-,\delta_2^-,a^-,a^+,c^+\right) \in \sembr{\Delta_1,\Delta_2,a:A}^-\times\sembr{a:A,c:C}^+ \mid {} \right.\\
    &\qquad\qquad\left.{} \mid l(u,\delta_1^+,\delta_2^+,a^+,a^-,c^-) \sqsubseteq (\delta_1^-,\delta_2^-,a^-,a^+,c^+) \right\}\\
    &R = \left\{ \left(\delta_1^-,\delta_2^-,a^-,a^+,c^+\right) \in \sembr{\Delta_1,\Delta_2,a:\Tand{\tau}{A}}^-\times\sembr{a:\Tand{\tau}{A},c:C}^+ \mid {} \right.\\
    &\qquad\qquad\left.{} \mid r(u,\delta_1^+,\delta_2^+,a^+,a^-,c^-) \sqsubseteq (\delta_1^-,\delta_2^-,a^-,a^+,c^+) \right\}.
  \end{align*}

  Infima of products are computed component-wise, so to show the semantic equivalence, it is sufficient to show that $\pi_{\Delta_1^-,\Delta_2^-,c^+} L = \pi_{\Delta_1^-,\Delta_2^-,c^+} R$.
  Set $v = \sembr{\jtypef{\Psi}{M}{\tau}}u$.

  To show $\pi_{\Delta_1^-,\Delta_2^-,c^+} L \subseteq \pi_{\Delta_1^-,\Delta_2^-,c^+} R$, let $(\delta_1^-,\delta_2^-,a^-,a^+,c^+) \in L$ be arbitrary.
  We show that there exists $(\alpha^-,\alpha^+) \in \sembr{a : \Tand{\tau}{A}}^- \times \sembr{a : \Tand{\tau}{A}}^+$ such that $(\delta_1^-,\delta_2^-,\alpha^-,\alpha^+,c^+) \in R$.
  Then where $\sembr{P}{u}(\delta_1^+,a^-) = (\delta_P^-,a_P^+)$ and $\sembr{Q}{\hat u}(\delta_2^+,a^+,c^-) = (\delta_Q^-,a_Q^-,c_Q^+)$,
  \[
    l(u,\delta_1^+,\delta_2^+,a^+,a^-,c^-)
    =
    (\delta_1^- : \delta_P^-,
    \delta_2^- : \delta_Q^-,
    a^- : a_Q^-,
    a^+ : a_P^+,
    c^+ : c_Q^+).
  \]
  By definition of $L$,
  \begin{equation}
    \label{eq:fscd:63}
    (\delta_P^-, \delta_Q^-, a_Q^-, a_P^+, c_Q^+) \sqsubseteq (\delta_1^-,\delta_2^-,a^-,a^+,c^+).
  \end{equation}
  Take $\alpha^- = a^-$ and $\alpha^+ = \upim{(v,a^+)}$.
  We show that $(\delta_1^-,\delta_2^-,a^-,\upim{(v,a^+)},c^+) \in R$.
  By assumption, $\sembr{\jtypef{\Psi}{M}{\tau}}u = v \neq \bot$.
  By definition of $r$,
  \begin{align*}
    &r(u,\delta_1^+,\delta_2^+,\upim{(v,a^+)},a^-,c^-)\\
    &= \left(\sembr{\tSendV{a}{M}{P}}u \times \sembr{\tRecvV{x}{a}{Q}}u\right)(\delta_1^+,\delta_2^+,\upim{(v,a^+)},a^-,c^-)\\
    &= (\sembr{\tSendV{a}{M}{P}}u(\delta_1^+,a^-), \sembr{\tRecvV{x}{a}{Q}}u(\delta_2^+,\upim{(v,a^+)},c^-))\\
    &= (\sembr{\tSendV{a}{M}{P}}u(\delta_1^+,a^-), \sembr{Q}\hat u(\delta_2^+,a^+,c^-))\\
    &= (\delta_1^- : \delta_P^-,
      \delta_2^- : \delta_Q^-,
      a^- : a_Q^-,
      a^+ : \upim{(v,a_P^+)},
      c^+ : c_Q^+).
  \end{align*}
  We must show that
  \[
    (\delta_P^-, \delta_Q^-, a_Q^-, \upim{(v,a_P^+)}, c_Q^+) \sqsubseteq (\delta_1^-,\delta_2^-,a^-,\upim{(v,a^+)},c^+).
  \]
  But this inequality follows readily from \cref{eq:fscd:63}.

  To show $\pi_{\Delta_1^-,\Delta_2^-,c^+} R \subseteq \pi_{\Delta_1^-,\Delta_2^-,c^+} L$, let $(\delta_1^-,\delta_2^-,a^-,a^+,c^+) \in R$ be arbitrary.
  We show that there exists $(\alpha^-,\alpha^+) \in \sembr{a : A}^- \times \sembr{a : A}^+$ such that $(\delta_1^-,\delta_2^-,\alpha^-,\alpha^+,c^+) \in L$.
  Where $\sembr{P}{u}(\delta_1^+,a^-) = (\delta_P^-,a_P^+)$,
  \[
    (\pi_{a^+} \circ r)(u,\delta_1^+,\delta_2^+,a^+,a^-,c^-) = \upim{(v,a_P^+)}.
  \]
  The definition of $R$ implies $\upim{(v,a_P^+)} \sqsubseteq a^+$.
  This implies that $a^+ = \upim{(v',a')}$ for some $v' \in \sembr{\tau}$ and $a' \in \sembr{A}^+$.
  So where $\sembr{Q}{\hat u}(\delta_2^+,a^+,c^-) = (\delta_Q^-,a_Q^-,c_Q^+)$,
  \[
    r(u,\delta_1^+,\delta_2^+,a^+,a^-,c^-)
    =
    (\delta_1^- : \delta_P^-,
    \delta_2^- : \delta_Q^-,
    a^- : a_Q^-,
    a^+ : \upim{(v,a_P^+)},
    c^+ : c_Q^+).
  \]
  By definition of $R$,
  \begin{equation}
    \label{eq:fscd:64}
    (\delta_P^-, \delta_Q^-, a_Q^-, \upim{(v,a_P^+)}, c_Q^+) \sqsubseteq (\delta_1^-,\delta_2^-,a^-,\upim{(v',a')},c^+).
  \end{equation}
  Take $\alpha^- = a^-$ and $\alpha^+ = a'$.
  Then by \cref{eq:fscd:64},
  \[
    (\delta_P^-, \delta_Q^-, a_Q^-, a_P^+, c_Q^+) \sqsubseteq (\delta_1^-,\delta_2^-,\alpha^-,\alpha^+,c^+).
  \]
  But
  \[
    l(u,\delta_1^+,\delta_2^+,\alpha^+,\alpha^-,c^-)
    =
    (\delta_1^- : \delta_P^-,
    \delta_2^- : \delta_Q^-,
    a^- : a_Q^-,
    a^+ : a_P^+,
    c^+ : c_Q^+).
  \]
  By transitivity,
  \[
    l(u,\delta_1^+,\delta_2^+,\alpha^+,\alpha^-,c^-) \sqsubseteq (\delta_1^-,\delta_2^-,\alpha^-,\alpha^+,c^+).
  \]
  We conclude that $(\delta_1^-,\delta_2^-,\alpha^-,\alpha^+,c^+) \in L$ and we are done.
\end{proof}

\assoccut*

\begin{proof}
  This is an immediate corollary of \cref{prop:fscd:13}.
\end{proof}

\etaunit*

\begin{proof}
  We use the Knaster-Tarski formulation \eqref{eq:808} of the trace operator.
  Let $u \in \sembr{\Psi}$ and $(\delta^+,c^-) \in \sembr{\Delta}^+ \times \sembr{c : C}^-$ be arbitrary.
  We abbreviate judgments by their processes and calculate that
  $\sembr{\tCut{a}{\tClose a}{(\tWait{a}{P})}}u(\delta^+,c^-) = \pi_{\Delta^- \times c^+}\left(\bigsqcap R\right)$
  where
  \begin{align*}
    r &= \sembr{\jtypem{\Psi}{\cdot}{\tClose a}{a}{\Tu}}u \times \sembr{\jtypem{\Psi}{\Delta,a : \Tu}{\tWait{a}{P}}{c}{C}}u\\
    R &= \left\{(\delta^+,a^-,a^+,c^-) \in \sembr{\Delta,a:A}^+ \times \sembr{a:A,c:C}^- \mid {}\right.\\
      &\qquad\qquad\left.{} \mid r(a^-,\delta^+,a^+,c^-) \sqsubseteq (a^+, \delta^-, a^-, c^+) \right\}.
  \end{align*}

  To show that
  \[
    \sembr{P}u(\delta^+,c^-) = \sembr{\tCut{a}{\tClose a}{(\tWait{a}{P})}}u(\delta^+,c^-)
  \]
  it is sufficient to show that
  \[
    \sembr{P}u(\delta^+,c^-) =  \pi_{\Delta^- \times c^+}\left(\bigsqcap R\right).
  \]
  Observe that for all $a^+, c^+, \delta^+$, we have $\pi_{a^+}(f(\bot, \delta^+, a^+, c^-)) = \ast$.
  Recall that $\ast$ is the top element in $\sembr{\Tu}^+$.
  By monotonicity of $f$, this implies that every element in $R$ is of the form $(\ast, c^+, \delta^-, a^-)$.
  By definition of $R$ and $f$, it then follows that every element in $R$ is of the form $(a^+ : \ast, c^+ : v, \delta^- : w, a^- : \bot)$ where $(\delta^- : w, c^+ : v) = \sembr{\jtypem{\Psi}{\Delta}{P}{c}{C}}u(\delta^+,c^-)$.
  This gives the result.
\end{proof}

\etatensor*

\begin{proof}
  We use the Knaster-Tarski formulation \eqref{eq:808} of the trace operator.
  Let $u \in \sembr{\Psi}$ and $(\delta_1^+,\delta_2^+,a^+,c^-) \in \sembr{\Delta_1,\Delta_2,a:A}^+ \times \sembr{c : C}^-$ be arbitrary.
  We abbreviate judgments by their processes and calculate that
  $\sembr{\tCut{b}{\left(\tSendC{b}{a}{P}\right)}{\left(\tRecvC{a}{b}{Q}\right)}}u(\delta_1^+,\delta_2^+,a^+,c^-) = \pi_{\Delta_1^-,\Delta_2^-,a^-,c^+}\left( \bigsqcap L \right)$
  and
  $\sembr{\tCut{b}{P}{Q}}u(\delta_1^+,\delta_2^+,a^+,c^-) = \pi_{\Delta_1^-,\Delta_2^-,a^-,c^+}\left( \bigsqcap R \right)$
  where
  \begin{align*}
    l &= \sembr{\jtypem{\Psi}{\Delta_1, a : A}{\tSendC{b}{a}{P}}{b}{A \Tot B}}u \times {}\\
      &\qquad\qquad {} \times \sembr{\jtypem{\Psi}{\Delta_2, b : A \Tot B}{\tRecvC{a}{b}{Q}}{c}{C}}u,\\
    L &= \left\{\left(\delta_1^-,\delta_2^-,a^-,b^-,b^+,c^+\right) \in \sembr{\Delta_1,\Delta_2,a : A,b : A \Tot B}^- \times \sembr{b : A \Tot B, c : C}^+ \mid {}\right.\\
      &\qquad\qquad\left.{} \mid l(\delta_1^+,\delta_2^+,a^+,b^+,b^-,c^-) \sqsubseteq (\delta_1^-,\delta_2^-,a^-,b^-,b^+,c^+) \right\},\displaybreak[0]\\
    r &= \sembr{\jtypem{\Psi}{\Delta_1}{P}{b}{B}}u \times \sembr{\jtypem{\Psi}{\Delta_2, a : A, b : B}{Q}{c}{C}}u,\\
    R &= \left\{\left(\delta_1^-,\delta_2^-,a^-,b^-,b^+,c^+\right) \in \sembr{\Delta_1,\Delta_2,a : A,b : B}^- \times \sembr{b : B, c : C}^+ \mid {}\right.\\
      &\qquad\qquad\left.{} \mid r(\delta_1^+,\delta_2^+,a^+,b^+,b^-,c^-) \sqsubseteq (\delta_1^-,\delta_2^-,a^-,b^-,b^+,c^+) \right\}.
  \end{align*}

  Infima of products are computed component-wise, so to show
  \[
    \sembr{\tCut{b}{\left(\tSendC{b}{a}{P}\right)}{\left(\tRecvC{a}{b}{Q}\right)}}u(\delta_1^+,\delta_2^+,a^+,c^-) = \sembr{\tCut{b}{P}{Q}}u(\delta_1^+,\delta_2^+,a^+,c^-)
  \]
  it is sufficient to show that $\pi_{\Delta_1^-,\Delta_2^-,a^-,c^+} L = \pi_{\Delta_1^-,\Delta_2^-,a^-,c^+} R$.

  To show $\pi_{\Delta_1^-,\Delta_2^-,a^-,c^+} L \subseteq \pi_{\Delta_1^-,\Delta_2^-,a^-,c^+} R$, let $\left(\delta_1^-,\delta_2^-,a^-,b^-,b^+,c^+\right) \in L$ be arbitrary, then $b^- = (b_a^-, b_b^-)$.
  We show that there exist $(\beta^-,\beta^+) \in \sembr{b : B}^- \times \sembr{b : B}^+$ such that $\left(\delta_1^-,\delta_2^-,a^-,\beta^-,\beta^+,c^+\right) \in R$.
  When $b^+ = (b_a^+, b_b^+)_\bot$, we take $\beta^+ = b_b^+$ and $\beta^- = b_b^-$.
  Then where $\sembr{P}u(\delta_1^+,b_b^-) = (\delta_P^-, b_P^+)$ and $\sembr{Q}u(\delta_2^+,b_a^+,b_b^+,c^-) = (\delta_Q^-,a_Q^-,b_Q^-,c_Q^+)$,
  \begin{align*}
    &l(\delta_1^+,\delta_2^+,a^+,(b_a^+,b_b^+)_\bot,(b_a^-,b_b^-),c^-)\\
    &=
      (\delta_1^- : \delta_P^-,
      \delta_2^- : \delta_Q^-,
      a^- : b_a^-,
      b^- : (a_Q^-, b_Q^-),
      b^+ : (a^+, b_P^+)_\bot,
      c^+ : c_Q^+).
  \end{align*}
  By definition of $L$,
  \begin{equation}
    \label{eq:sc:18}
    (\delta_P^-,\delta_Q^-,b_a^-,(a_Q^-, b_Q^-), (a^+, b_P^+)_\bot, c_Q^+)
    \sqsubseteq
    (\delta_1^-,\delta_2^-,a^-,(b_a^-,b_b^-),(b_a^+,b_b^+),c^+).
  \end{equation}
  We show $\left(\delta_1^-,\delta_2^-,a^-,\beta^-,\beta^+,c^+\right) \in R$.
  \[
    r\left(\delta_1^+,\delta_2^+,a^+,b_b^+,b_b^-,c^-\right) =
    (\delta_1^- : \delta_P^-,
    \delta_2^- : \delta_Q^-,
    a^- : a_Q^-,
    b^- : b_Q^-,
    b^+ : b_P^+,
    c^+ : c_Q^+).
  \]
  We must show
  \[
    (\delta_P^-, \delta_Q^-, a_Q^-, b_Q^-, b_P^+, c_Q^+)
    \sqsubseteq
    (\delta_1^-,\delta_2^-,a^-,b_b^-,b_b^+,c^+).
  \]
  All cases except $a_Q^- \sqsubseteq a^-$ are immediate from \cref{eq:sc:18}.
  That $a_Q^- \sqsubseteq a^-$ follows by transitivity from $a_Q^- \sqsubseteq b_a^-$ and $b_a^- \sqsubseteq a^-$, both by \cref{eq:sc:18}.
  This establishes the case when $b^+ = (b_a^+, b_b^+)_\bot$.
  The case $b^+ = \bot$ is impossible.
  When $b^+ = \bot$, $\sembr{\tRecvC{a}{b}{Q}}u(\delta_2^+,\bot,c^-) = (\bot,\bot,\bot)$.
  Then by definition of $l$ and $L$,
  \begin{align*}
    &l(\delta_1^+,\delta_2^+,a^+,(b_a^+,b_b^+)_\bot,(b_a^-,b_b^-),c^-)\\
    &= (\delta_1^- : \delta_P^-, \delta_2^- : \bot, a^- : b_b^-, b^- : \bot, b^+ : (a^+, b_P^+)_\bot, c^+ : \bot)\\
    &\sqsubseteq (\delta_1^-,\delta_2^-,a^-,(b_a^-,b_b^-),\bot,c^+).
  \end{align*}
  Looking at the $b^+$ component, we have $(a^+, b_P^+)_\bot \sqsubseteq \bot$, a contradiction.
  This establishes $\pi_{\Delta_1^-,\Delta_2^-,a^-,c^+} L \subseteq \pi_{\Delta_1^-,\Delta_2^-,a^-,c^+} R$.

  To show $\pi_{\Delta_1^-,\Delta_2^-,a^-,c^+} R \subseteq \pi_{\Delta_1^-,\Delta_2^-,a^-,c^+} L$, let $\left(\delta_1^-,\delta_2^-,a^-,b^-,b^+,c^+\right) \in R$ be arbitrary.
  We claim $\left(\delta_1^-,\delta_2^-,a^-,(a^-,b^-),(a^+,b^+)_\bot,c^+\right) \in L$.
  Where $\sembr{P}u(\delta_1^+,b^-) = (\delta_P^-, b_P^+)$ and $\sembr{Q}u(\delta_2^+,a^+,b^+,c^-) = (\delta_Q^-,a_Q^-,b_Q^-,c_Q^+)$,
  \[
    r(\delta_1^+,\delta_2^+,a^+,b^+,b^-,c^-) =
    (\delta_1^- : \delta_P^-,
    \delta_2^- : \delta_Q^-,
    a^- : a_Q^-,
    b^- : b_Q^-,
    b^+ : b_P^+,
    c^+ : c_Q^+).
  \]
  By definition of $R$,
  \begin{equation}
    \label{eq:sc:19}
    (\delta_P^-, \delta_Q^-, a_Q^-, b_Q^-, b_P^+, c_Q^+)
    \sqsubseteq
    (\delta_1^-,\delta_2^-,a^-,b^-,b^+,c^+).
  \end{equation}
  We compute
  \begin{align*}
    &l(\delta_1^+,\delta_2^+,a^+,(a^+,b^+)_\bot,(a^-,b^-),c^-)\\
    &=
      (\delta_1^- : \delta_P^-,
      \delta_2^- : \delta_Q^-,
      a^- : \bot,
      b^- : (a_Q^-, b_Q^-),
      b^+ : (a^+, b_P^+)_\bot,
      c^+ : c_Q^+).
  \end{align*}
  We must show
  \[
    (\delta_P^-, \delta_Q^-, \bot, (a_Q^-, b_Q^-), (a^+, b_P^+)_\bot, c_Q^+)
    \sqsubseteq
    (\delta_1^-,\delta_2^-,a^-,(a^-,b^-),(a^+,b^+)_\bot,c^+).
  \]
  This follows immediately from \cref{eq:sc:19}, so $\pi_{\Delta_1^-,\Delta_2^-,a^-,c^+} R \subseteq \pi_{\Delta_1^-,\Delta_2^-,a^-,c^+} L$.

  Having established $\pi_{\Delta_1^-,\Delta_2^-,a^-,c^+} R = \pi_{\Delta_1^-,\Delta_2^-,a^-,c^+} L$, it follows that
  \begin{align*}
    &\sembr{\jtypem{\Psi}{\Delta_1,\Delta_2,a:A}{\tCut{b}{\left(\tSendC{b}{a}{P}\right)}{\left(\tRecvC{a}{b}{Q}\right)}}{c}{C}}u(\delta_1^+,\delta_2^+,a^+,c^-)\\
    &= \sembr{\jtypem{\Psi}{\Delta_1,\Delta_2,a:A}{\tCut{b}{P}{Q}}{c}{C}}u(\delta_1^+,\delta_2^+,a^+,c^-).
  \end{align*}
  Because $(\delta_1^+,\delta_2^+,a^+,c^-)$ was an arbitrary element in their domain, we conclude the functions are equal.
\end{proof}

\etadshift*

\begin{proof}
  We abbreviate judgments by their processes.
  The function $\pi_{a^+} \circ \sembr{\tSendS{a}{P}}u$ is not strict because of the post-composition with $\up$ in the $a^+$ component.
  By \cref{prop:4503}, this implies we can drop the $\strictfn_{a^+}$ when computing the denotation of the cut processes:
  \begin{align*}
    &\sembr{\jtypem{\Psi}{\Delta_1,\Delta_2}{\tCut{a}{\left(\tSendS{a}{P}\right)}{\left(\tRecvS{a}{Q}\right)}}{c}{C}}u\\
    &= \Tr{\left(\sembr{\tSendS{a}{P}}u \times \sembr{\tRecvS{a}{Q}}u \right)}{a^- \times a^+}\\
    &= \Tr{\left(\left(\left(\ms{id} \times \left(a^+ : \up\right)\right) \circ \sembr{P}u\right) \times \strictfn_{a^+}\left(\sembr{Q}u \circ \left(\ms{id} \times \left(a^+ : \down\right)\right)\right)\right)}{a^- \times a^+}\\
    &= \Tr{\left(\left(\left(\ms{id} \times \left(a^+ : \up\right)\right) \circ \sembr{P}u\right) \times \left(\sembr{Q}u \circ \left(\ms{id} \times \left(a^+ : \down\right)\right)\right)\right)}{a^- \times a^+}\\
    &= \Tr{\left(\left(\ms{id} \times \left(a^+ : \up\right)\right) \circ \left(\sembr{P}u \times \sembr{Q}u\right) \circ \left(\ms{id} \times \left(a^+ : \down\right)\right)\right)}{a^- \times a^+}\\
    \shortintertext{Because $(\up,\down)$ is a section-retraction pair, by \cref{prop:fscd:12}:}
    &= \Tr{\left(\sembr{P}u\times\sembr{Q}u\right)}{a^-\times a^+}\\
    &= \sembr{\jtypem{\Psi}{\Delta_1,\Delta_2}{\tCut{a}{P}{Q}}{c}{C}}u.\qedhere
  \end{align*}
\end{proof}

\etaplus*

\begin{proof}
  We use the Knaster-Tarski formulation \eqref{eq:808} of the trace operator.
  Let $u \in \sembr{\Psi}$ and $(\delta_1^+,\delta_2^+,c^-) \in \sembr{\Delta_1,\Delta_2}^+ \times \sembr{c : C}^-$ be arbitrary.
  We abbreviate judgments by their processes and calculate that
  \begin{align*}
    \sembr{\tCut{a}{P}{Q_k}}u(\delta_1^+,\delta_2^+,c^-) &= \Tr{(l)}{b^- \times b^+}(\delta_1^+,\delta_2^+,c^-)\\
                                                         &= \pi_{\Delta_1^-,\Delta_2^-,c^+}\left(\bigsqcap L \right),\\
    \sembr{\tCut{a}{\left(\tSendL{a}{k}{P}\right)}{\left(\tCase{a}{\left\{l \Rightarrow Q_l\right\}_{l \in L}}\right)}}u(\delta_1^+,\delta_2^+,c^-) &= \Tr{(r)}{b^- \times b^+}(\delta_1^+,\delta_2^+,c^-)\\
                                                         &= \pi_{\Delta_1^-,\Delta_2^-,c^+}\left(\bigsqcap R \right)
  \end{align*}
  where
  \begin{align*}
    l &: \sembr{\Delta_1,\Delta_2,a : A_k}^+ \times \sembr{a : A_k, c : C}^- \to \sembr{\Delta_1,\Delta_2,a : A_k}^- \times \sembr{a : A_k, c : C}^+\\
    l &= \sembr{\jtypem{\Psi}{\Delta_1}{P}{a}{A_k}}u \times \sembr{\jtypem{\Psi}{\Delta_2,a:A_k}{Q_k}{c}{C}}u,\\
    L &= \left\{ \left(\delta_1^-,\delta_2^-,a^-,a^+,c^+\right) \in \sembr{\Delta_1,\Delta_2,a : A_k}^- \times \sembr{a : A_k, c : C}^+ \mid {} \right.\\
      &\qquad\qquad \left. {} \mid l\left(\delta_1^+,\delta_2^+,a^+,a^-,c^-\right) \sqsubseteq \left(\delta_1^-,\delta_2^-,a^-,a^+,c^+\right) \right\},\\
    r &: \sembr{\Delta_1,\Delta_2,a : \Tplus \{ l : A_l \}_{l \in L}}^+ \times \sembr{a : \Tplus \{ l : A_l \}_{l \in L}, c : C}^- \to {}\\
      &\qquad\qquad {} \to \sembr{\Delta_1,\Delta_2,a : \Tplus \{ l : A_l \}_{l \in L}}^- \times \sembr{a : \Tplus \{ l : A_l \}_{l \in L}, c : C}^+\\
    r &= \sembr{\jtypem{\Psi}{\Delta_1}{\tSendL{a}{k}{P}}{a}{\Tplus \{ l : A_l \}_{l \in L}}}u \times {}\\
      &\qquad\qquad {} \times \sembr{\jtypem{\Psi}{\Delta_2,a:\Tplus \{ l : A_l \}_{l \in L}}{\tCase{a}{\left\{l \Rightarrow Q_l\right\}_{l \in L}}}{c}{C}}u,\displaybreak[3]\\
    R &= \left\{ \left(\delta_1^-,\delta_2^-,a^-,a^+,c^+\right) \in \sembr{\Delta_1,\Delta_2,a : \Tplus \{ l : A_l \}_{l \in L}}^- \times \sembr{a : \Tplus \{ l : A_l \}_{l \in L}, c : C}^+ \mid {} \right.\\
      &\qquad\qquad \left. {} \mid r\left(\delta_1^+,\delta_2^+,a^+,a^-,c^-\right) \sqsubseteq \left(\delta_1^-,\delta_2^-,a^-,a^+,c^+\right) \right\}.
  \end{align*}

  Infima of products are computed component-wise, so to show
  \[
    \sembr{\tCut{a}{P}{Q_k}}u(\delta_1^+,\delta_2^+,c^-) = \sembr{\tCut{a}{\left(\tSendL{a}{k}{P}\right)}{\left(\tCase{a}{\left\{l \Rightarrow Q_l\right\}_{l \in L}}\right)}}u(\delta_1^+,\delta_2^+,c^-)
  \]
  it is sufficient to show that $\pi_{\Delta_1^-,\Delta_2^-,c^+} L = \pi_{\Delta_1^-,\Delta_2^-,c^+} R$.

  We first show $\pi_{\Delta_1^-,\Delta_2^-,c^+} L \supseteq \pi_{\Delta_1^-,\Delta_2^-,c^+} R$.
  Let $(\delta_1^-,\delta_2^-,(a_l^-)_{l \in L},a^+,c^+) \in R$ be arbitrary and let $\sembr{P}u(\delta_1^+,a_k^-) = (\delta_P^-,a_P^+)$.
  The case $a^+ = \bot$ is impossible.
  Indeed, when $a^+ = \bot$, strictness gives $\sembr{\tCase{a}{\left\{l \Rightarrow Q_l\right\}_{l \in L}}}u(\delta_2^+,\bot,c^-) = \bot$.
  This implies
  \[
    r(\delta_1^-,\delta_2^-,a^+ : \bot, (a_l^-)_{l \in L}, c^+) =
    (\Delta_1^- : \delta_P^-,
    \Delta_2^- : \bot,
    a^- : \bot,
    a^+ : (k, \upim{a_P^+}),
    c^+ : \bot),
  \]
  which by definition of $R$ implies the contradiction $(k, \upim{a_P^+}) \sqsubseteq \bot$.
  So $a^+ = (l, \upim{a_l^+})$ for some $l \in L$ and $a_l^+ \in \sembr{A_l}^+$.
  The definitions of $r$ and $R$ imply $(k, \upim{a_P^+}) \sqsubseteq (l, \upim{a_l^+})$, so we have $l = k$.
  We show that $(\delta_1^-,\delta_2^-,a_k^-,a_k^+,c^+) \in L$.
  Let $\sembr{Q_k}u(\delta_2^+,a_k^+,c^-) = (\delta_Q^-,a_Q^-,c_Q^+)$.
  We calculate
  \begin{align}
    &l(\delta_1^+,\delta_2^+,a_k^+,a_k^-,c^-)
      = (\Delta_1^- : \delta_P^-,
      \Delta_2^- : \delta_Q^-,
      a^- : a_Q^-,
      a^+ : a_P^+,
      c^+ : c_Q^+),\label{eq:4753}\\
    \begin{split}
      &r(\delta_1^-,\delta_2^-,(k,\upim{a_k^+}), (a_l^-)_{l \in L}, c^+)\\
      &= (\Delta_1^- : \delta_P^-,
      \Delta_2^- : \delta_Q^-,
      a^- : (k : a_Q^-, l \neq k : \bot),
      a^+ : (k, \upim{a_P^+}),
      c^+ : c_Q^+).\label{eq:5357}
    \end{split}
  \end{align}
  By definition of $R$, \cref{eq:5357} implies
  \begin{equation}
    \label{eq:66666}
    (\delta_P^-, \delta_Q^-, (k : a_Q^-, l \neq k : \bot), (k, \upim{a_P^+}), c_Q^+)
    \sqsubseteq
    (\delta_1^-,\delta_2^-,(a_l^-)_{l \in L},(k, \upim{a_k^+}),c^+).
  \end{equation}
  It immediately follows from \eqref{eq:66666} that
  \[
    (\delta_P^-, \delta_Q^-, a_Q^-, a_P^+, c_Q^+)
    \sqsubseteq
    (\delta_1^-,\delta_2^-,a_k^-, a_k^+,c^+),
  \]
  \ie, that $(\delta_1^-,\delta_2^-,a_k^-,a_k^+,c^+) \in L$.
  We deduce $\pi_{\Delta_1^-,\Delta_2^-,c^+} L \supseteq \pi_{\Delta_1^-,\Delta_2^-,c^+} R$.

  To show $\pi_{\Delta_1^-,\Delta_2^-,c^+} L \subseteq \pi_{\Delta_1^-,\Delta_2^-,c^+} R$, let $(\delta_1^-,\delta_2^-,a_k^-,a_k^+,c^+) \in L$ be arbitrary.
  Let $\sembr{P}u(\delta_1^+,a_k^-) = (\delta_P^-,a_P^+)$ and $\sembr{Q_k}u(\delta_2^+,a_k^+,c^-) = (\delta_Q^-,a_Q^-,c_Q^+)$.
  We calculate
  \begin{align}
    &l(\delta_1^+,\delta_2^+,a_k^+,a_k^-,c^-) =
      (\Delta_1^- : \delta_P^-,
      \Delta_2^- : \delta_Q^-,
      a^- : a_Q^-,
      a^+ : a_P^+,
      c^+ : c_Q^+),\label{eq:2752}\\
    \begin{split}
      &r(\delta_1^-,\delta_2^-,(k, \upim{a_k^+}),(k : a^-_k, l \neq k : \bot),c^+)\\
      &= (\Delta_1^- : \delta_P^-,
      \Delta_2^- : \delta_Q^-,
      a^- : (k : a_Q^-, l \neq k : \bot),
      a^+ : (k, \upim{a_P^+})
      c^+ : c_Q^+).\label{eq:37537}
    \end{split}
  \end{align}
  By definition of $L$, \cref{eq:2752} implies
  \begin{equation}
    \label{eq:252721}
    (\delta_P^-,\delta_Q^-,a_Q^-,a_P^+,c_Q^+) \sqsubseteq (\delta_1^-,\delta_2^-,a_k^-,a_k^+,c^+).
  \end{equation}
  It immediately follows from \eqref{eq:252721} that
  \[
    (\delta_P^-, \delta_Q^-, (k : a_Q^-, l \neq k : \bot), (k, \upim{a_P^+}), c_Q^+)
    \sqsubseteq
    (\delta_1^-,\delta_2^-,(k : a^-_k, l \neq k : \bot),(k, \upim{a_k^+}),c^+),
  \]
  \ie, that $(\delta_1^-,\delta_2^-,(k : a^-_k, l \neq k : \bot),(k,\upim{a_k^+}),c^+) \in R$.
\end{proof}

\etafunfold*

\begin{proof}
  We abbreviate judgments by their processes.
  \begin{align*}
    &\sembr{\tCut{a}{(\tSendU{a}{P})}{(\tRecvU{a}{Q})}}u\\
    &= \Trop^{a^- \times a^+}\left(\sembr{\tSendU{a}{P}}u \times \sembr{\tRecvU{a}{Q}}u\right)\\
    &= \Trop^{a^- \times a^+}\left(
      \left(
      \left(\ms{id} \times \left(a^+ : \ms{Fold}\right)\right)
      \circ \sembr{P}u
      \circ \left(\ms{id} \times \left(a^- : \ms{Unfold}\right)\right)
      \right)
      \times {} \right.\\
    &\qquad\qquad\qquad\left.{} \times \left(
      \left(\ms{id} \times \left(a^- : \ms{Fold}\right)\right)
      \circ \sembr{Q}u
      \circ \left(\ms{id} \times \left(a^+ : \ms{Unfold}\right)\right)
      \right)
      \right)\\
    &= \Trop^{a^- \times a^+}\left(
      \left(\ms{id} \times \left(a^+ : \ms{Fold}\right)  \times \left(a^- : \ms{Fold}\right) \right)
      \circ \left( \sembr{P}u \times \sembr{Q}u \right) \circ {} \right.\\
    &\qquad\qquad\qquad \left. {}
      \circ \left(\ms{id} \times \left(a^- : \ms{Unfold}\right) \times \left(a^+ : \ms{Unfold}\right)\right)
      \right)\\
    \shortintertext{which by \cref{prop:fscd:12}:}
    &= \Trop^{a^- \times a^+}\left(\sembr{P}u \times \sembr{Q}u\right)\\
    &= \sembr{\tCut{a}{P}{Q}}u.\qedhere
  \end{align*}
\end{proof}

\cohsestyp*

\begin{proof}
  This is exactly \cref{prop:fscd:8}.
\end{proof}

\cohterpro*

\begin{proof}
  This is exactly \cref{prop:extended:9}.
\end{proof}

\ssst*

\begin{proof}
  This is exactly \cref{prop:fscd:7}.
\end{proof}

\ssft*

\begin{proof}
  This is exactly \cref{prop:extended:10}.
\end{proof}


\end{document}